%% file: paper.tex
\documentclass[acmsmall,screen]{acmart}
\setcopyright{rightsretained}
\acmPrice{}
\acmDOI{10.1145/3571216}
\acmYear{2023}
\copyrightyear{2023}
\acmSubmissionID{popl23main-p120-p}
\acmJournal{PACMPL}
\acmVolume{7}
\acmNumber{POPL}
\acmArticle{23}
\acmMonth{1}
\received{2022-07-07}
\received[accepted]{2022-11-07}

\usepackage[T1]{fontenc}
\usepackage{booktabs}   
\usepackage{subcaption}
\usepackage{graphicx}
\usepackage{adjustbox}
\usepackage{listings}

\usepackage{stmaryrd}
\usepackage{arydshln}
\usepackage{proof}
\usepackage{xspace}
\usepackage{multirow}
\usepackage{pifont}
\usepackage{enumitem}
\setlist{nosep,leftmargin=\parindent}

\usepackage{tikz}
\usepackage[framemethod=tikz]{mdframed}
\usepackage{xcolor}
\usepackage{mathtools}
\usepackage{microtype}
\usepackage{stackengine}
\usepackage{hyperref}

\usepackage{wrapfig}

\usepackage[color]{changebar}
\cbcolor{red}

\newif\ifdispComments
\dispCommentsfalse

\input{defns.tex}
\usepackage{refs}
\begin{document}

%% Title information
\title[Unrealizability Logic]{Unrealizability Logic}         %% [Short Title] is optional;
                                        %% when present, will be used in
                                        %% header instead of Full Title.
%\titlenote{with title note}            %% \titlenote is optional;
                                        %% can be repeated if necessary;
                                        %% contents suppressed with 'anonymous'
%\subtitle{Subtitle}                     %% \subtitle is optional
%\subtitlenote{with subtitle note}       %% \subtitlenote is optional;
                                        %% can be repeated if necessary;
                                        %% contents suppressed with 'anonymous'

%% Author information
%% Contents and number of authors suppressed with 'anonymous'.
%% Each author should be introduced by \author, followed by
%% \authornote (optional), \orcid (optional), \affiliation, and
%% \email.
%% An author may have multiple affiliations and/or emails; repeat the
%% appropriate command.
%% Many elements are not rendered, but should be provided for metadata
%% extraction tools.

%% Author with single affiliation.
\author{Jinwoo Kim}                                        %% can be repeated if necessary
\authornote{Part of work done while a student at the University of Wisconsin-Madison.}
\affiliation{
  \institution{Seoul National University}            %% \institution is required
  \city{Seoul}
  \country{Republic of Korea}                  %% \country is recommended
}
\email{pl@cs.wisc.edu}          %% \email is recommended

\author{Loris D'Antoni}                                        %% can be repeated if necessary
\affiliation{
  \institution{University of Wisconsin-Madison}            %% \institution is required
  \city{Madison}
  %\state{Wisconsin}
  \country{USA}                  %% \country is recommended
}
\email{loris@cs.wisc.edu}          %% \email is recommended

\author{Thomas Reps}                                        %% can be repeated if necessary
\affiliation{
  \institution{University of Wisconsin-Madison}            %% \institution is required
  \city{Madison}
  %\state{Wisconsin}
  \country{USA}                  %% \country is recommended
}
\email{reps@cs.wisc.edu}          %% \email is recommended

%% Abstract
%% Note: \begin{abstract}...\end{abstract} environment must come
%% before \maketitle command
\input{abstract}

%% 2012 ACM Computing Classification System (CSS) concepts
%% Generate at 'http://dl.acm.org/ccs/ccs.cfm'.
\begin{CCSXML}
<ccs2012>
   <concept>
       <concept_id>10003752.10003790.10002990</concept_id>
       <concept_desc>Theory of computation~Logic and verification</concept_desc>
       <concept_significance>500</concept_significance>
       </concept>
   <concept>
       <concept_id>10003752.10003790.10011741</concept_id>
       <concept_desc>Theory of computation~Hoare logic</concept_desc>
       <concept_significance>500</concept_significance>
       </concept>
   <concept>
       <concept_id>10011007.10011074.10011092.10011782</concept_id>
       <concept_desc>Software and its engineering~Automatic programming</concept_desc>
       <concept_significance>500</concept_significance>
       </concept>
 </ccs2012>
\end{CCSXML}

\ccsdesc[500]{Theory of computation~Logic and verification}
\ccsdesc[500]{Theory of computation~Hoare logic}
\ccsdesc[500]{Software and its engineering~Automatic programming}
%% End of generated code

%% Keywords
%% comma separated list
% \keywords{keyword1, keyword2, keyword3}  
%% \keywords are mandatory in final camera-ready submission
\keywords{Unrealizability Logic, Unrealizability, Program Synthesis}

%% \maketitle
%% Note: \maketitle command must come after title commands, author
%% commands, abstract environment, Computing Classification System
%% environment and commands, and keywords command.
\maketitle

\input{1intro.tex}
\input{2motivating.tex}
\input{3logic.tex}
\input{4soundnesscompleteness.tex}

\input{5examples.tex}
\input{6related.tex}

\input{7conclusion.tex}

%% Acknowledgments
\begin{acks}
Supported, in part,
by a gift from Rajiv and Ritu Batra;
by ONR under grant N00014-17-1-2889;
by NSF under grants CCF-\{1750965, 1763871, 1918211, 2023222, 2211968, 2212558\}; 
by a Facebook Research Faculty Fellowship, 
by a Microsoft Research Faculty Fellowship, 
and a grant from the Korea Foundation of Advanced Studies.
Any opinions, findings, and conclusions or recommendations
expressed in this publication are those of the authors,
and do not necessarily reflect the views of the sponsoring
entities.
\end{acks}

%% Bibliography
%\bibliography{bibfile}
%\bibliographystyle{ACM-Reference-Format}
\bibliography{reference}

%% Appendix
 \newpage
 \appendix

 \input{appendix_proofs.tex}

\end{document}

%% file: defns.tex
\makeatletter
\newtheorem*{rep@theorem}{\rep@title}
\newcommand{\newreptheorem}[2]{%
\newenvironment{rep#1}[1]{%
 \def\rep@title{#2 \ref{##1}}%
 \begin{rep@theorem}}%
 {\end{rep@theorem}}}
\makeatother

\newcommand{\mypar}[1]{\vspace{1mm}\noindent\textit{#1.}}

\newreptheorem{theorem}{Theorem}
\newreptheorem{lemma}{Lemma}

\newcommand*\circled[1]{\tikz[baseline=(char.base)]{
            \node[shape=circle,fill,inner sep=2pt] (char) {\textcolor{white}{#1}};}}

\newcommand{\rone}{(\emph{i})~}
\newcommand{\rtwo}{(\emph{ii})~}
\newcommand{\rthree}{(\emph{iii})~}

\newcommand{\lbbar}{\{\kern-0.5ex|}
\newcommand{\rbbar}{|\kern-0.5ex\}}

\newcommand{\biglbbar}{\left\{\kern-0.5ex\left|}
\newcommand{\bigrbbar}{\right|\kern-0.5ex\right\}}

\newcommand{\ex}{\mathsf{e}}

\def\messy{\textsc{MESSY}\xspace}

\def\nope{\textsc{nope}\xspace}

\usepackage{color}

\def\nope{\textsc{Nope}\xspace}
\def\nay{\textsc{Nay}\xspace}

\newcommand{\semgus}{\textsc{SemGuS}\xspace}
\def\sygus{\textsc{SyGuS}\xspace}

\newcommand{\type}{\tau}

\newcommand{\RHS}{\textit{RHS}}           % a production right-hand side

\newcommand{\Etrue}{{{\mathsf{true}}}}
\newcommand{\Efalse}{{{\mathsf{false}}}}
\newcommand{\Et}{{{\mathsf{t}}}}
\newcommand{\Ef}{{{\mathsf{f}}}}
\newcommand{\Eifthenelse}[3]{{\mathsf{if}}\ {#1}\ {\mathsf{then}}\ {#2}\ {\mathsf{else}}\ {#3}}
\newcommand{\Eseq}[2]{{#1}{\mathsf{;}}\ {#2}}

\newcommand{\sem}[1]{\llbracket{#1}\rrbracket}       % [[ ]] 

\newcommand{\Eassign}[2]{{{#1}\ \mathsf{:=}\ {#2}}}

\newcommand{\Ewhile}[2]{\mathsf{while} \; {#1} \; \mathsf{do} \; {#2}}

\newcommand{\semantics}[1]{\llbracket {#1} \rrbracket}

\newcommand{\Omit}[1]{}

\usepackage[breakable]{tcolorbox}
\newenvironment{mybox}[1][gray!20]{
	\begin{tcolorbox}[   %% Adjust the following parameters at will.
		breakable,
		left=0pt,
		right=0pt,
		top=0pt,
		bottom=-1pt,
		colback=#1,
		colframe=#1,
		width=\dimexpr\textwidth\relax,
		boxsep=2pt,
		arc=0pt,outer arc=0pt,
		]
	}{
\end{tcolorbox}
}

\newcounter{resq}%[section]

\newcommand{\triple}[3]{\{{#1}\} \: #2 \: \{#3\}}

\newcommand{\bigutriple}[3]{\biglbbar {#1} \bigrbbar \: #2 \: \biglbbar #3 \bigrbbar}
\newcommand{\bigutripletight}[3]{\biglbbar {#1} \bigrbbar  #2 \biglbbar #3 \bigrbbar}
\newcommand{\utriple}[3]{\lbbar {#1} \rbbar \: #2 \: \lbbar #3 \rbbar}
\newcommand{\utriplenospace}[3]{\lbbar {#1} \rbbar  #2 \lbbar #3 \rbbar}

\def\state{r}

\def\stmt{\mathsf{s}}
\def\exp{\mathsf{e}}
\def\bexp{\mathsf{b}}

\newcommand{\Eskip}{\mathsf{Skip}}

\newcommand{\aP}{\mathcal{P}}
\newcommand{\aQ}{\mathcal{Q}}
\newcommand{\aR}{\mathcal{R}}

\newcommand{\aI}{\mathcal{I}}

\newcommand{\aB}{\mathcal{B}}

\newcommand{\vars}{\mathit{vars}}

\newcommand{\zero}{\mathsf{Zero}}
\newcommand{\one}{\mathsf{One}}
\newcommand{\varrule}{\mathsf{Var}}
\newcommand{\truerule}{\mathsf{True}}
\newcommand{\falserule}{\mathsf{False}}
\newcommand{\conj}{\mathsf{Conj}}
\newcommand{\notrule}{\mathsf{Not}}
\newcommand{\hp}{\mathsf{HP}}
\newcommand{\applyhp}{\mathsf{ApplyHP}}
\newcommand{\inv}{\mathsf{Inv}}
\newcommand{\weaken}{\mathsf{Weaken}}
\newcommand{\seqrule}{\mathsf{Seq}}
\newcommand{\assignrule}{\mathsf{Assign}}
\newcommand{\subone}{\mathsf{Sub1}}
\newcommand{\subtwo}{\mathsf{Sub2}}
\newcommand{\grmdisj}{\mathsf{GrmDisj}}
\newcommand{\while}{\mathsf{While}}
\newcommand{\whilerule}{\while}

\newcommand{\plusrule}{\mathsf{Plus}}
\newcommand{\minusrule}{\mathsf{Minus}}
\newcommand{\multrule}{\mathsf{Mult}}
\newcommand{\divrule}{\mathsf{Div}}
\newcommand{\andrule}{\mathsf{And}}
\newcommand{\siterule}{\mathsf{ITE}}
\newcommand{\ltrule}{\mathsf{LT}}

\newcommand{\eqrule}{\mathsf{Eq}}

\newcommand{\set}[1]{\{{#1}\}}
\newcommand{\stateset}[1]{\{{#1}\}}
\newcommand{\exset}[1]{\{{#1}\}}

\newcommand{\lrangle}[1]{\langle {#1} \rangle}

\newcommand{\syfirst}{\mathit{sy}_{\mathsf{first}}}
\newcommand{\gfirst}{G_{\mathsf{first}}}
\newcommand{\nstart}{\mathit{Start}}
\newcommand{\mymod}{\text{mod }}

\newcommand{\smod}[3]{{#1} \equiv_{#3} {#2}}
\newcommand{\snmod}[3]{{#1} \not \equiv_{#3} {#2}}

\newcommand{\ginf}{G_{\mathsf{inf}}}

\newcommand{\syu}{sy_{\mathsf{u}}}
\newcommand{\gu}{G_{\mathsf{u}}}
\newcommand{\syfin}{sy_{\mathsf{fin}}}
\newcommand{\gfin}{G_{\mathsf{fin}}}

\newcommand{\cinc}[1]{\mathsf{INC}\ {#1}}
\newcommand{\cdec}[1]{\mathsf{DEC}\ {#1}}
\newcommand{\cjez}[2]{\mathsf{JEZ}\ {#1}\ {#2}}
\newcommand{\fb}{\vec{\beta}}

\newcommand{\eq}{=}
\renewcommand{\vec}[1]{\bold{#1}}

\newcommand{\ext}[2]{\mathsf{ext}({#1}, {#2})}
\newcommand{\proj}[2]{\mathsf{proj_{e}}({#1}, {#2})}
\newcommand{\proji}[2]{\mathsf{proj_{i}}({#1}, {#2})}

\newcommand{\tbranch}{\semantics{s_1}(\semantics{b}(\pi))}
\newcommand{\fbranch}{\semantics{s_2}(\semantics{b}(\pi))}

\newcommand{\tbranchbeta}{\semantics{s_1}(\beta)}
\newcommand{\fbranchbeta}{\semantics{s_2}(\beta)}

\newcommand{\tbranchext}{\proj{\semantics{s_1}(\ext{\beta}{\beta_1})}{\vec{v_1}}}
\newcommand{\fbranchext}{\proj{\semantics{s_2}(\ext{\beta}{\beta_2})}{\vec{v_1}}}

\newcommand{\ostar}{\stackMath\mathbin{\stackinset{c}{0ex}{c}{0ex}{\star}{\bigcirc}}}
\newcommand{\semext}[1]{\semantics{{#1}}_{\mathsf{ext}}}

\newcommand{\subsubsubsection}[1]{\vspace{2pt plus 1pt minus 1pt}\noindent{\bf #1}}
\newcommand{\colormath}[2]{\begingroup\color{#1}#2\endgroup}

\newcommand{\syite}{\mathit{sy}_{\mathsf{id\_ite}}}
\newcommand{\gite}{\mathit{G}_{\mathsf{id\_ite}}}
\newcommand{\syconst}{\mathit{sy}_{\mathsf{id\_const}}}
\newcommand{\gconst}{\mathit{G}_{\mathsf{id\_const}}}
\newcommand{\sysum}{\mathit{sy}_{\mathsf{sum}}}
\newcommand{\gsum}{\mathit{G}_{\mathsf{sum}}}

\newcommand{\yaux}{y_{\mathit{aux}}}
\newcommand{\vaux}{v_{\mathit{aux}}}

\definecolor{dgreen}{RGB}{0,128,0}
\definecolor{dred}{RGB}{200,0,0}

\newcommand{\vx}{\vec{x}}
\newcommand{\vy}{\vec{y}}

\newcommand{\sstate}{a}

%% file: abstract.tex
% To cover: context, gap, innovation, significance, [evidence]

\begin{abstract}
% State the problem
We consider the problem of establishing that a program-synthesis problem is
\emph{unrealizable} (i.e., has no solution in a given search space of
programs).
% Say why it's an interesting problem
Prior work on unrealizability has developed some automatic techniques
to establish that a problem is unrealizable;
however, these techniques are all \emph{black-box}, meaning that they conceal the
reasoning behind \emph{why} a synthesis problem is unrealizable.

% Say what your solution achieves
In this paper, we present a Hoare-style reasoning system, called \emph{unrealizability logic}
for establishing that a program-synthesis problem is unrealizable.
To the best of our knowledge, unrealizability logic is the first proof system for
overapproximating the execution of an infinite set of imperative programs.
% Say what follows from your solution
The logic provides a general, logical system 
for building checkable proofs about
unrealizability.
Similar to how Hoare logic distills the fundamental concepts 
behind algorithms and tools to prove the correctness of programs,
unrealizability logic distills into a single logical system the
fundamental concepts that were hidden within prior tools capable
of establishing that a program-synthesis problem is unrealizable.
\end{abstract}

%%% Local Variables:
%%% mode: latex
%%% TeX-master: "paper.tex"
%%% End:

%% file: 1intro.tex
\section{Introduction}
\label{Se:Introduction}

Program synthesis refers to the task of discovering a program, within 
a given search space, that satisfies 
a behavioral specification (e.g., a logical formula, 
or a set of input-output examples).
While there have been many advances in program synthesis, especially in domain-specific 
settings~\cite{flashfill,swizzle,lambda2}, program synthesis 
remains a challenging task with many properties that are not yet well understood.
%Attesting to this fact is the existence of many different solvers and 
%solving strategies for program synthesis 
%(such as enumeration, constraint solving, reduction, etc.) each of which often utilize different 
%\emph{properties} that are hidden within a synthesis problem 
%(such as behavioral equivalence or quantifier elimination) for efficient solving. 

While tools are becoming better at synthesizing programs, one property that remains difficult to reason about is the 
\emph{unrealizability} of a synthesis problem, i.e., the non-existence of a solution that 
satisfies the behavioral specification within the search space of possible programs.
One of the ultimate goals behind studying unrealizability 
is to prune the search space when synthesizing programs, by showing that a certain  
subset of the search space does not contain the desired solution 
(see \S\ref{Se:RelatedWork} for a further discussion of this idea).
Unrealizability also has many applications;
for example, one can show that a certain synthesized solution is optimal with respect to some 
metric by proving that a better solution to the synthesis problem does not 
exist---i.e., by proving that the synthesis problem where the search space 
contains only programs of lower cost is unrealizable~\cite{qsygus}.
Unrealizability is also used to prune program paths in 
symbolic-execution engines for program repair~\cite{unrealse}.

%or ensuring that solvers do not run on end on unrealizable problems~\cite{pps}.
%For an example of an unrealizable synthesis problem, 
%consider the following synthesis problem $\syfirst$ in Example~\ref{ex:syfirst}, 
%which simply attempts to return the first variable $x$ out of two input variables $x$ and $y$.

\begin{example}
  \label{ex:syfirst}
  Consider the synthesis problem $\syfirst$ where the goal is to 
  synthesize a function $f$ that takes as input a state $(x, y)$, 
  and returns a state where $y = 10$.
  Assume, however, that the search space of possible programs in $\syfirst$ 
  is defined using the following grammar 
  $\gfirst$:
\[
  \begin{array}{ccc}
    \nstart  \rightarrow  \Eassign{y}{E}  & &
    E        \rightarrow  x \mid E + 1
  \end{array}
\]
  Clearly $\Eassign{y}{10} \not \in L(\mathit{Start})$; moreover, all programs in $L(\mathit{Start})$  are incorrect on at least one input.
  For example, 
  on the input $x = 15$
  every program in the grammar sets $y$ to a value greater than $15$.
  Consequently, $\syfirst$ is unrealizable.
\end{example}

While it is trivial for a human to establish that $\syfirst$
is indeed unrealizable, only a small number of known techniques 
can prove this fact automatically~\cite{nay,semgus,nope}.
However, a common drawback of these techniques is that they do not produce a proof artifact.
In particular, \citet{semgus} and \citet{nope} rely on external 
constraint solvers and program 
verifiers to prove unrealizability.
While a solver-based approach has worked well in practice, 
without extensive knowledge of the internals
of the external solvers, it is difficult to understand exactly
\emph{why} a synthesis problem is unrealizable.

This paper presents \emph{unrealizability logic},  a proof system
for reasoning about the unrealizability of synthesis problems.
In addition to the main goal of reasoning about unrealizability, 
unrealizability logic is designed with the following goals in mind:
\begin{itemize}
  \item to be a \emph{general} logic, capable of dealing with various synthesis problems;
  \item to be amenable to \emph{machine reasoning}, as to enable both automatic proof checking 
    and to open future opportunities for automation;
  \item to \emph{provide insight} into why certain synthesis problems are unrealizable through 
    the process of completing a proof tree.
\end{itemize}
Via unrealizability logic, one is able to
\rone reason about unrealizability in a principled, explicit fashion, and 
\rtwo produce concrete proofs about unrealizability.

In this paper, unrealizability logic is formulated for synthesis
problems over a deterministic, imperative programming language with
statements involving Boolean and integer expressions.
There are several challenges to creating unrealizability logic, 
which are illustrated in \S\ref{Se:MotivatingExample}.
One such challenge is already illustrated in Example~\ref{ex:syfirst}:
because of the recursive definition of nonterminal $E$, the search space
of programs $L(\gfirst)$
is an \emph{infinite} set.
In unrealizability logic, one reasons about infinite sets of programs
 en masse---as opposed to reasoning about each program in the set separately.
Proofs in unrealizability logic must thus establish judgements of the following kind:
\begin{mybox}
  For a given a set of input-output examples,
  no matter which program is chosen (out of a possibly infinite set of
  programs), there is at least one input-output example that is
  handled incorrectly.
\end{mybox}

The proof system for unrealizability logic has sound underpinnings,
and provides a way to build proofs of unrealizability similar to the way
Hoare logic~\cite{hoare} 
provides a way
to build proofs that a given program cannot reach a set of bad states.

%Although some previous work, such as \citet{nay}, \citet{semgus},
%or \citet{unrealwitness},
%can also prove that some synthesis problems are unrealizable,
%they do not make their reasoning steps explicit, whereas
%unrealizability logic can produce a tangible proof tree to justify the
%reasoning behind why a given problem is unrealizable (\S\ref{Se:MotivatingExample}).

\vspace{1mm}\noindent\textit{Contributions.} In summary, this paper makes the following contributions:
\begin{itemize}
  \item \emph{Unrealizability logic}, the first logical proof system for 
    overapproximating the execution of an infinite set of 
    imperative programs (\S\ref{Se:UnrealizabilityLogic}).
  \item A proof of soundness and relative completeness for unrealizability logic 
    (\S\ref{Se:SoundnessAndCompleteness}).
  \item
    Examples that illustrate the proving power of
    unrealizability logic: in particular, unrealizability logic can be used to prove 
    unrealizability for problems out of reach for previous methods; e.g., 
    those for which the proof requires reasoning about an infinite number of examples 
    (\S\ref{Se:Examples}).
\end{itemize}

\noindent
\S\ref{Se:RelatedWork} discusses related work.
\S\ref{Se:Conclusion} concludes.
Proofs for theorems in the paper may be found in the full version of this paper, 
available on arXiv~\cite{ularxiv}.

%%% Local Variables:
%%% mode: latex
%%% TeX-master: "paper.tex"
%%% End:

%% file: 2motivating.tex
\section{Motivating Examples}
\label{Se:MotivatingExample}

In this section, we give several key examples that describe how proof trees in 
unrealizability logic work, 
and also illustrate two key challenges 
behind unrealizability logic
(\S\ref{SubSe:InfiniteSetsOfPrograms}, \S\ref{SubSe:InputOutputPairs}).
Unrealizability logic shares much of the intuition 
behind Hoare logic and its extension toward recursive programs.
However, as we will illustrate in \S\ref{SubSe:WhyNotHoare}, 
these concepts alone are insufficient to model unrealizability, which motivated us to develop
the new concepts that we introduce in this paper.

Hoare logic is based on triples that overapproximate the set of states that can be reached 
by a program $\stmt$; i.e., the Hoare triple
$$
  \triple{P}{\stmt}{Q}
$$
asserts that $Q$ is an \emph{overapproximation} 
of all states that may be reached by executing $\stmt$, starting from a state in $P$.
The intuition in Hoare logic is that one will often attempt to prove 
a triple like $\triple{P}{\stmt}{\lnot X}$ for a set of bad states $X$, which 
ensures that execution of $\stmt$ cannot reach $X$.
%because $\lnot B$ is an overapproximation of all reachable states, the fact that 
%$\triple{P}{\stmt}{\lnot B}$ holds implies that, when started in $P$, $\stmt$ cannot reach $B$,
%and is thus safe.

Unrealizability logic operates on the same overapproximation principle, 
but differs in two main ways from standard Hoare logic.
The differences are motivated by how synthesis problems are typically defined,
using two components: \rone a search space $S$ (i.e., a set of programs), and 
\rtwo a (possibly infinite) set of related input-output pairs 
$\set{(i_1, o_1), (i_2, o_2), \cdots }$.

To reason about \emph{sets} of programs,
in unrealizability logic, the central element (i.e., the program $\stmt$) 
is changed to a \emph{set} of programs $S$.
The unrealizability-logic triple 
$$
  \utriple{P}{S}{Q}
$$
thus asserts that $Q$ is an overapproximation of all states that are 
reachable by executing \emph{any possible combination} of 
a pre-state $p \in P$ \emph{and} a program $\stmt \in S$.
More formally, the following theorem holds for any 
unrealizability triple $\utriple{P}{S}{Q}$:

\begin{theorem}
  \label{thm:unreal_hoare}
  The \emph{unrealizability triple} $\utriple{P}{S}{Q}$ for a precondition $P$, 
  postcondition $Q$, and set of programs $S$, 
  holds iff for each program $\stmt \in S$, 
  the Hoare triple $\triple{P}{\stmt}{Q}$ holds.
\end{theorem}

The second difference concerns \emph{input-output pairs}: 
in unrealizability logic, we wish to place the input states 
in the precondition, and overapproximate the set of states reachable from the input states 
(through a set of programs) as the postcondition.
Unfortunately, the input-output pairs of a synthesis problem
cannot be tracked using standard pre- and postconditions;
nor can they be tracked using auxiliary variables, 
because of a complication arising from the fact that unrealizability 
logic must reason about a \emph{set} of programs 
(see~\S\ref{SubSe:InputOutputPairs}).

To keep the input-output relations in check, the predicates of unrealizability logic 
talk instead about (potentially infinite) \emph{vector-states}, which are sets of states in which
each individual state is associated with a unique index.
Variable $x$ of the state with index $i$ is referred to as $x_i$.

\begin{example}[Vector-States]
\label{ex:vector_state}
Assume we are given a set of input-output state pairs
$\{(x = i_1, x = o_1)$, $(x = i_2, x = o_2)$, $\cdots$, $(x = i_n, x = o_n)\}$.
Then the input vector-state $I$ is denoted by
$(x_1 = i_1) \wedge (x_2 = i_2) \wedge \cdots \wedge (x_n = i_n)$.
The output vector-state $O$ is denoted by
$(x_1 = o_1) \wedge (x_2 = o_2) \wedge \cdots \wedge (x_n = o_n)$.
Once $I$ and $O$ have been constructed this way, the proof steps of
an unrealizability-logic proof would relate $x_i$ in the precondition
to $x_i$ in the postcondition;
one has effectively expressed the relation given by
the input-output pairs by \emph{renaming} each 
input-output pair to be unique.
\end{example}

Keeping these two differences in mind, let us turn to proving unrealizability.
Recall that a synthesis problem is given as a search space $S$ and a set of input-output 
pairs $\mathsf{Ex} = \set{(i_1, o_1), \cdots}$.
To prove unrealizability, one must prove that for every
program $\stmt \in S$, there 
is at least \emph{one} input-output pair $(i_k, o_k) \in \mathsf{Ex}$ such that 
$\stmt$ does not map $i_k$ to $o_k$.
Note that $\lnot O$ denotes the set of all output vector-states for which at 
least \emph{one} input-output pair does not hold.
Then, because of Theorem~\ref{thm:unreal_hoare}, proving 
$\triple{I}{\stmt}{\lnot O}$ for all $\stmt \in S$ is equivalent to proving the
unrealizability triple
$$
  \utriple{I}{S}{\lnot O}.
$$
The goal of unrealizability logic is to prove such triples in a principled logical system, 
\emph{without} having to descend to the level of individual programs or 
input-output pairs. 

\begin{figure}[tb!]
  \centering
  \[
    {
      \underbrace{\overbrace{
      \biglbbar
      \begin{array}{@{\hspace{0ex}}l}
      	x_1 = i_1 \land x_2 = i_2 \land \\ 
      	x_3 = i_3 \land \cdots 
		\end{array}			
		\bigrbbar
  }^{\textit{\mbox{\begin{tabular}{c}
     Precondition:\\ multiple possible\\ inputs\end{tabular}}}}}_{\textit{Vector-state} \: (I)}
    }\: 
    {
      \overbrace{
        L\underbrace{\left(\begin{array}{@{\hspace{0ex}}l@{\hspace{0.75ex}}c@{\hspace{0.75ex}}l@{\hspace{0ex}}}
                 \nstart & \rightarrow & \cdots  \\
                 E       & \rightarrow & \cdots 
               \end{array}\right)}_{\textit{Grammar}}
             }^{\textit{\mbox{\begin{tabular}{c} Body:\\ \textit{Set of programs} \end{tabular}}}}
    }\: 
    {
	\underbrace{\overbrace{
	\biglbbar
      \begin{array}{@{\hspace{0ex}}l}
      	x_1 \neq o_1 \vee x_2 \neq o_2 \vee \\ 
      	x_3 \neq o_3 \vee \cdots
		\end{array}	
	\bigrbbar				
      	}^{\textit{\mbox{\begin{tabular}{c}Postcondition:
        \\ at least one\\ incorrect output\end{tabular}}}}
          }_{\textit{Negated vector-state} \: (\lnot O)}
    }
  \]
  \caption{An unrealizability triple with a negated post-vector-state asserts that
    a synthesis problem is unrealizable.
    Observe how the vectorized version of $\lnot O$ asserts that there is at \emph{least one} 
    output example that does not satisfy the desired input-output relation.
  }
  \label{Fi:EssenceOfUnrealizabilityTriples}
\end{figure}

\figref{EssenceOfUnrealizabilityTriples} summarizes what we have
discussed so far.
In general, in unrealizability logic
the input vector-state need not be finite, and the spec need not be 
functional as well---for example, a synthesis problem 
that should assign some value $v > x$ to $x$,
for all possible $x$,
can be proved unrealizable via the triple 
$\utriple{\forall i. x_i = i}{L(G)}{\exists i. x_i \leq i}$.

\subsection{Challenge 1: Infinite Sets of Programs}
\label{SubSe:InfiniteSetsOfPrograms}

Compared to ordinary Hoare logic, the first challenge is that unrealizability logic needs to reason about infinite sets of programs. 
Fortunately, the set of programs in a program-synthesis problem is typically
formulated as a regular tree grammar (RTG), 
which defines the set of programs in an inductive manner.
Unrealizability logic uses the structure of the RTG to develop proof trees that mimic structural induction.

As illustrated earlier, the unrealizability triple
$\utriple{P}{S}{Q}$ 
captures information about the set of programs $S$.
If $S$ is specified in a recursive manner via an RTG, 
a triple $\utriple{P}{S}{Q}$ 
%, which we sometimes refer to as \emph{hypotheses},
can also be used as an induction hypothesis in a proof tree for unrealizability logic,
as we show in Example~\ref{ex:basic_unreal}.

\begin{example}
  \label{ex:basic_unreal}
  Consider a simple synthesis problem $sy_{\ex}$ with the following grammar:
$$
  \begin{array}{ccccc}
    \nstart  \rightarrow  S2 \mid S3  &&
    S2       \rightarrow  \Eseq{S2}{S2} \mid \Eassign{x}{x + 2} &&
    S3       \rightarrow  \Eseq{S3}{S3} \mid \Eassign{x}{x + 3}
  \end{array}
$$
  That is, $sy_{\ex}$ consists of programs that either \rone repeatedly add $2$ to $x$, 
  or \rtwo repeatedly add $3$ to $x$.

  Now suppose that the input specification to $sy_{\ex}$ was given as $\smod{x}{0}{6}$ (we use $\smod{x}{r}{p}$ as shorthand  for $x \equiv r (\mymod p)$; i.e., that $x$ is equivalent to $r$ modulo $p$) and the 
  output specification given as $\smod{x}{1}{6}$; that is, the goal is to reach a number of the
  form $6k' + 1$ when starting from a number $6k$.
  This problem is unrealizable, because one can only reach $6k'$, $6k' + 2$, $6k' + 3$, or $6k' + 4$ 
  by repeatedly adding 2 or 3  (but not both)  to a multiple of $6$.
  We formalize this reasoning as a proof tree in unrealizability logic.

  \subsubsubsection{Starting the Proof Tree.}
  To prove the synthesis problem $sy_{\ex}$ unrealizable, one wishes to prove the following triple,
  where the output is negated from the specification:
  $$
    \utriple{\smod{x}{0}{6}}{\nstart}{\snmod{x}{1}{6}}
  $$
  From this point on, a nonterminal as the center element of an unrealizability triple refers to the 
  language of that nonterminal; e.g., $\utriple{\smod{x}{0}{6}}{\nstart}{\snmod{x}{1}{6}}$ refers to 
  $\utriple{\smod{x}{0}{6}}{L(\nstart)}{\snmod{x}{1}{6}}$.

  Note that the postcondition $\snmod{x}{1}{6}$ is bigger than our 
  previously discussed reachable states of $6k', 6k' + 2, 6k' + 3, $ and $6k' + 4$. 
  Hence, the target triple can be proved by proving the following
  triple and then weakening it ($\weaken$ in Fig.~\ref{fig:struct_rules}):
  $$
    \utriple{\smod{x}{0}{6}}{\nstart}{\smod{x}{0}{6} \vee \smod{x}{2}{6} \vee \smod{x}{3}{6} \vee \smod{x}{4}{6}}
  $$
 
  %NOTE: Is validity the correct word?
  To prove this triple in unrealizability logic, we must first introduce the concept of a 
  \emph{context} $\Gamma$, which is a set of triples that stores 
  all the induction hypotheses that have been introduced
  up to some point in a proof tree. 
  Consequently, the judgements of the proof tree have the form:
  $$
    \Gamma \vdash \utriple{\aP}{S}{\aQ}
  $$
  The idea is that instead of reasoning about the provability of triples directly, 
  we wish to reason about the provability of a triple \emph{assuming} that every hypothesis inside $\Gamma$ is true.

  In our case, we wish to prove that 
  $\utriple{\smod{x}{0}{6}}{\nstart}{\smod{x}{0}{6} \vee \smod{x}{2}{6} \vee \smod{x}{3}{6} \vee \smod{x}{4}{6}}$ 
  without assuming anything (i.e., starting from the empty context).
  This triple can be established by proving that the following judgement holds (where the blank LHS denotes an empty context):
  $$
    \vdash \utriple{\smod{x}{0}{6}}{\nstart}{\smod{x}{0}{6} \vee \smod{x}{2}{6} \vee \smod{x}{3}{6} \vee \smod{x}{4}{6}}
  $$

  The key point in taking the next step is to notice that the language of the nonterminal $\nstart$, $L(\nstart)$, 
  is the \emph{union} of $L(S2)$ and $L(S3)$.
  Because unrealizability logic deals with sets of programs, it is equipped with rules for merging triples over 
  different sets of programs.
  In particular, we can apply the grammar-disjunction 
  rule of unrealizability logic ($\grmdisj$  in Fig.~\ref{fig:struct_rules}), which states that if 
  two program sets satisfy the same pre- and postcondition pair, 
  their union also satisfies the aforementioned pair.
  In this case, $\grmdisj$ is applied on our target triple as follows:
  {\small
  $$
    \infer[\grmdisj]
    {\vdash \utriple{\smod{x}{0}{6}}{\nstart}{\smod{x}{0}{6} \vee \smod{x}{2}{6} \vee \smod{x}{3}{6} \vee \smod{x}{4}{6}}}
    {
      \begin{aligned}
      \vdash \utriple{\smod{x}{0}{6}}{S2}{\smod{x}{0}{6} \vee \smod{x}{2}{6} \vee \smod{x}{3}{6} \vee \smod{x}{4}{6}} \\
      \vdash \utriple{\smod{x}{0}{6}}{S3}{\smod{x}{0}{6} \vee \smod{x}{2}{6} \vee \smod{x}{3}{6} \vee \smod{x}{4}{6}}
      \end{aligned}
    }
  $$
  }
  \noindent We are now faced with having to prove triples over the nonterminals $S2$ and $S3$.
  Because $S2$ and $S3$ are defined recursively, taking a naive 
  consideration of the productions from $S2$ and $S3$ will result in an infinite proof tree.
  To avoid this problem, one must introduce the triples one wishes to prove as
  \emph{hypotheses} in the context, and validate them in a procedure 
  similar to structural induction.

  \begin{figure*}[!tb]
  {\tiny
  $$
    \infer[\hp]{\vdash \utriple{\smod{x}{0}{2}}{S2}{\smod{x}{0}{2}}}{
      \infer[\weaken]{\Gamma_{S2} \vdash \utriple{\smod{x}{0}{2}}{\Eassign{x}{x + 2}}{\smod{x}{0}{2}}  \hspace{2mm} \circled{1}}{
        \infer[\assignrule]{\Gamma_{S2} \vdash \bigutripletight{\smod{x}{0}{2}}{\Eassign{x}{x + 2}}
        {\begin{array}{c}
          \exists x', e_t'. (\smod{x'}{0}{2} \wedge \\
          e_t = x' + 2) \wedge \\
          x = e_t
        \end{array}}}{
          \infer[]{\Gamma_{S2} \vdash \bigutripletight{\smod{x}{0}{2}}
          {x + 2}{\begin{array}{c}\exists e_t'. \: (\smod{x}{0}{2} ) \\ \wedge e_t = x + 2 \end{array}}}{\cdots}
        }
      }
      &
      \infer[\seqrule]{\Gamma_{S2} \vdash \utriple{\smod{x}{0}{2}}{\Eseq{S2}{S2}}{\smod{x}{0}{2}} \hspace{2mm} \circled{2}}{
        \infer[\applyhp]{\Gamma_{S2} \vdash \utriple{\smod{x}{0}{2}}{S2}{\smod{x}{0}{2}}}{} 
        &
        \infer[\applyhp]{\Gamma_{S2} \vdash \utriple{\smod{x}{0}{2}}{S2}{\smod{x}{0}{2}}}{}       
      }
    }
  $$
  }
    \caption[Simplified proof tree that proves the sub-goal 
             $\utriple{\smod{x}{0}{2}}{S2}{\smod{x}{0}{2}}$ in unrealizability logic.
    ]{Simplified proof tree that proves the sub-goal 
      $\utriple{\smod{x}{0}{2}}{S2}{\smod{x}{0}{2}}$ in unrealizability logic. $\Gamma_{S2}$ denotes the context $\stateset{\utriple{\smod{x}{0}{2}}{S2}{\smod{x}{0}{2}}}$.
      Labels $\circled{1}$ and $\circled{2}$ are names for the triples they are associated with.
     
   % \twr{Is it correct to have ``$\exists e_t'. \ldots$'' in the post-condition of the top-most triple (on the left)?  The instantiation of the Assign rule looks wrong to me.}
    }
    \label{fig:s2_prooftree}
  \end{figure*}
  
  \subsubsubsection{Introducing Hypotheses.}
  To see how hypotheses are introduced, consider the nonterminal $S2$.
  The idea is that one wishes to introduce the target triple 
  $\utriple{\smod{x}{0}{6}}{S2}{\smod{x}{0}{6} \vee \smod{x}{2}{6} \vee \smod{x}{3}{6} \vee \smod{x}{4}{6}}$
  as an \emph{induction hypothesis} about nonterminal $S2$, and prove that this triple holds in a way similar to 
  structural induction.
  Introducing a new hypothesis can be done using the $\hp$ rule in Fig.~\ref{fig:struct_rules}, which
  splits the proof according to nonterminal $N$'s productions:
  in this case, nonterminal $S2$ is split into $\Eassign{x}{x + 2}$ and $\Eseq{S2}{S2}$ 
  (the first application of $\hp$ in Fig.~\ref{fig:s2_prooftree}).

  As the hypothesis for $S2$, we introduce the triple $\utriple{\smod{x}{0}{2}}{S2}{\smod{x}{0}{2}}$: 
  observe that this triple may be used to prove the target triple 
  $\utriple{\smod{x}{0}{6}}{S2}{\smod{x}{0}{6} \vee \smod{x}{2}{6} \vee \smod{x}{3}{6} \vee \smod{x}{4}{6}}$ 
  via an application of $\weaken$
  (in fact, $\utriple{\smod{x}{0}{6}}{S2}{\smod{x}{0}{6} \vee \smod{x}{2}{6} \vee \smod{x}{3}{6} \vee \smod{x}{4}{6}}$ 
  will not work as a hypothesis directly for $S2$, as we will explain later in this section).
  
  We will use $\Gamma_{S2}$ to denote the singleton context $\stateset{\utriple{\smod{x}{0}{2}}{S2}{\smod{x}{0}{2}}}$.
  Fig.~\ref{fig:s2_prooftree} depicts the proof tree for
  $\vdash \utriple{\smod{x}{0}{2}}{S2}{\smod{x}{0}{2}}$, 
  where the application of $\hp$ at the root of the proof tree
  introduces  the $\Gamma_{S2}$ context in the two premises.

  \subsubsubsection{Proving Hypotheses.}
  If one were proving a property about $S2$ using standard structural induction,
  one would proceed to show that the property
  holds on the two sub-cases, $\Eassign{x}{x + 2}$ 
  and $\Eseq{S2}{S2}$---i.e., the two productions from $S2$---while assuming that 
  the property holds as a hypothesis.
  The same approach is used here: 
  we assume $\utriple{\smod{x}{0}{2}}{S2}{\smod{x}{0}{2}}$
  as a hypothesis (in context $\Gamma_{S2})$, and attempt to prove that 
  $\smod{x}{0}{2}$ is a valid pre- and postcondition for both $\Eassign{x}{x + 2}$ and $\Eseq{S2}{S2}$.
  
%  \jinwoo{Edits on line 457~492 to use $\circled{1}$ (Not sure if this was what Loris had in mind)}
  First, consider proof goal $\circled{1}$ for $\Eassign{x}{x + 2}$.
  The set of programs generated by the right-hand side of this production 
  is completely independent of how $S2$ is defined; 
  thus, for the purpose of performing structural induction on $S2$, it is a \emph{base case}.
  Here, we invoke the basic rule for assignment to prove $\circled{1}$.
%  Although the rule for assignment, and the method for dealing with expressions, are different in unrealizability logic 
%  compared to Hoare logic (due to the need to consider sets of programs), here $\Eassign{x}{x + 2}$
%  derives a single program, and thus one's intuition from Hoare logic applies here as well.
  Auxiliary variable $e_t$ stores the result of the expression $x + 2$, which is then assigned to 
  $x$ through the $\assignrule$ rule 
  (in which the postcondition of the conclusion mimics that of the forwards-based assignment rule in Hoare logic).
  The resulting postcondition can then be weakened to obtain $\smod{x}{0}{2}$, which completes the proof that 
  $\utriple{\smod{x}{0}{2}}{\Eassign{x}{x + 2}}{\smod{x}{0}{2}}$.

  Next, consider proof goal $\circled{2}$ for $\Eseq{S2}{S2}$.
  Here, nonterminal $S2$ appears directly, and provides us with a chance to apply the
  induction hypothesis via the rule $\applyhp$ in Fig.~\ref{fig:struct_rules}.
  In Figure~\ref{fig:s2_prooftree}, the derivation for $\circled{2}$ is the sub-proof tree on the right.
  Notice how $\Eseq{S2}{S2}$ is first decomposed using the $\seqrule$ rule for sequential composition 
  (identical to that in Hoare logic), which requires us to prove two instances of 
  the proof goal $\circled{2}$; in turn, these two instances are proved by directly
  applying the induction hypothesis through $\applyhp$.
% TWR omitted
%  In terms of normal structural induction, this is identical to invoking the induction hypothesis on each 
%  sub-nonterminal $S2$ in order to show that the induction hypothesis holds for the case $\Eseq{S2}{S2}$.

  Returning to nonterminal $\nstart$, the second premise
  $\utriple{\smod{x}{0}{6}}{S3}{\smod{x}{0}{6} \vee \smod{x}{2}{6} \vee \smod{x}{3}{6} \vee \smod{x}{4}{6}}$ 
  can be similarly proved by introducing the triple 
  $\utriple{\smod{x}{0}{3}}{S3}{\smod{x}{0}{3}}$ as a hypothesis for nonterminal $S3$, 
  and weakening it to get the target triple about $S3$.
  This step concludes our proof.
%  proving it, then applying $\weaken$ to obtain the final goal.
%  By applying $\grmdisj$ on $S2$ and $S3$ then, one can obtain the final target 
%  $\utriple{\smod{x}{0}{6}}{\nstart}{\smod{x}{0}{6} \vee \smod{x}{2}{6} \vee \smod{x}{3}{6} \vee \smod{x}{4}{6}}$, 
%  thereby proving that $\syfirst$ is unrealizable.
\end{example}

One take-away from Example~\ref{ex:basic_unreal} is the importance of choosing an 
appropriate induction hypothesis.
For instance, attempting to use the proof goal 
$\utriple{\smod{x}{0}{6}}{S2}{\smod{x}{0}{6} \vee \smod{x}{2}{6} \vee \smod{x}{3}{6} \vee \smod{x}{4}{6}}$
directly as a hypothesis for $S2$ would make the proof fail on the 
production $\Eseq{S2}{S2}$.
This example indicates that, similar to how identifying appropriate
invariants for loops is a key component of writing Hoare logic proofs, identifying
appropriate hypotheses for nonterminals is an essential part in
completing a proof in unrealizability logic.

%For example, if we had not strengthened the proof goal on $S2$, and instead introduced the triple 
%$\utriple{\smod{x}{0}{6}}{S2}{\smod{x}{0}{6} \vee \smod{x}{2}{6} \vee \smod{x}{3}{6} \vee \smod{x}{4}{6}}$ 
%as the hypothesis, the proof would have failed when trying to prove
%that the hypothesis holds for $\Eseq{S2}{S2}$.
%Such scenarios appear in ordinary mathematical induction proofs as
%well, where sometimes one must use a strengthened version of the proof
%goal in order to apply the induction hypothesis.
%Similarly, the choice of hypothesis dictates whether one can apply
%$\applyhp$ directly (i.e., invoke the induction hypothesis),
%whether one must introduce new hypotheses, or even
%whether the proof succeeds.
%From another point of view, one can also argue that the structure of the grammar guides the 
%hypothesis choice: for example, a production of the form $N \rightarrow \Eseq{N}{N}$ implies that 
%a hypothesis over $N$ must be an invariant---i.e., 
%a rule of the form $\utriple{\aP}{N}{\aP}$---in order to use $\applyhp$ directly.
%Note that $\utriple{\aP}{N}{\aP}$ specifies a \emph{language invariant}---an invariant
%of all programs in $L(N)$.

\subsection{Challenge 2: Tracking (Infinite) Input-Output Relations}
\label{SubSe:InputOutputPairs}

The second challenge in unrealizability logic is that 
the specification of a synthesis problem is typically given as a set of 
input-output \emph{pairs}: a specific input value is associated with 
a specific output value.
(Even if the specification is given as a universally quantified formula, one can 
understand such a formula as an infinite set of input-output pairs.)
  The standard way to address this problem in Hoare logic is to introduce 
  auxiliary variables that freeze the values of program variables in the precondition, 
  and to allow the postcondition to refer to these auxiliary variables.
  However, this approach alone is unsuitable for unrealizability logic,
  as illustrated by the following example.
%  \rone unrealizability logic deals with sets of programs, and 
%  \rtwo one requires some form of assurance that results in the postcondition are \emph{not} generated 
%  by different examples executing different programs.

\begin{example}
  \label{ex:univ_problem}
  Consider the identity assignment $\Eassign{x}{x}$.
  The Hoare triple for this program is typically given as $\triple{x = x_{aux}}{\Eassign{x}{x}}{x = x_{aux}}$;
  that is, starting from $x = x_{aux}$, where $x_{aux}$ is an \emph{auxiliary variable}, one ends up in 
  $x = x_{aux}$.
  Strictly speaking, the auxiliary variable $x_{aux}$ is quantified \emph{outside} the Hoare triple:
  $\forall x_{aux}. \triple{x = x_{aux}}{\Eassign{x}{x}}{x = x_{aux}}$; 
  this position for the quantifier indicates that the triple should hold for \emph{every} value of $x_{aux}$.

  Now suppose that we are given a synthesis problem $sy_{\mathsf{id}}$, 
  where the goal is to synthesize a program equivalent to $\Eassign{x}{x}$,
  using the grammar $G_{\mathsf{id}}$ below:
  $$
  \begin{array}{ccc}
    \nstart  \rightarrow  \Eassign{x}{E}&&
    E       \rightarrow 0 \mid E + 1 \mid E - 1 \\
  \end{array}
  $$
  For every integer $i$, the set $L(E)$ contains a constant expression that evaluates
  to $i$, and thus $L(\nstart)$ consists of exactly the set of all constant assignments.
  However, synthesizing a statement that is computationally equivalent
  to $\Eassign{x}{x}$ is \emph{impossible} because
  the set only contains constant assignments---i.e., $sy_{\mathsf{id}}$ is unrealizable.

  Suppose that one tries to specify $sy_{\mathsf{id}}$ using an auxiliary
  variable, so that the goal is to prove $\utriple{x = x_{aux}}{\nstart}{x \neq x_{aux}}$.
  Unfortunately, this triple is \emph{invalid} in unrealizability logic, 
  in the sense that starting from the precondition $x = x_{aux}$, 
  one can actually reach a state where still, $x = x_{aux}$!

  To understand this seemingly counterintuitive fact, 
  we will attempt to prove the triple 
  $\utriple{x = x_{aux}}{\nstart}{\exists k. x = k}$ 
  (where $k$ indicates that $x$ may have any value in the post-state).
  To do so,
  let us first characterize the behavior of all programs in 
  the language of $E$.
  In doing so, one will generate the following hypothesis
  (where, again, $e_t$ is an auxiliary variable for storing the value
  obtained from executing $E$):
  $$
    \utriple{x = x_{aux}}{E}{x = x_{aux} \wedge \exists k. e_t = k}.
  $$
%\twr{The post-condition of the triple we claimed we wanted to prove is $x \neq x_{aux}$, so shouldn't the post-condition
%on the previous line be $x \neq x_{aux} \wedge \exists k. e_t = k$?
%}
  The best one can say about terms in $E$ 
  is that there exists an integer $k$ whose value is $e_t$.
  %This hypothesis is \emph{precise}.
  Applying $\assignrule$ and $\weaken$ to the triple, one can derive the following 
  (precise) triple:
  $$
    \utriple{x = x_{aux}}{\nstart}{\exists k. x = k}.
  $$
  Like in Hoare logic, the quantification for $x_{aux}$ is \emph{outside} the 
  unrealizability triple, which yields:
  $$
    \forall x_{aux}. \utriple{x = x_{aux}}{\nstart}{\exists k. x = k}.
  $$
This triple says that for every value of
  $x_{aux}$, there is \emph{some} $k$ (and a corresponding program)
  that works; thus, this triple actually asserts that the problem
  is realizable (which is clearly not true)!

  The problem is that, while for every \emph{individual} value of $x_{aux}$
  there does indeed exist a suitable $k$ (corresponding to a specific expression
  $t_k \in L(E)$ that always evaluates to $k$),
  any such value of $k$ (and expression $t_k$)
  fails to work for other values of $x_{aux}$---i.e., the use of a single auxiliary variable fails to capture the fact that 
  the multiple input-output pairs must all be satisfied \emph{simultaneously}.
  In terms of synthesis, for each input, some program in the 
  search space will work:
  but this does not guarantee the existence of 
  a program that works for \emph{all} inputs
  (such a program does not exist).
  % The problem is that the triple states that for all $x_{aux}$, there is \emph{some} $k$ that works.
  % What we want, however, is a \emph{single} $k$ that works for \emph{all} $x_{aux}$---i.e., \emph{all} examples---and in fact,
  % such a program does not exist in $L(G_{\mathsf{id}})$.
\end{example}

The issue illustrated in Example~\ref{ex:univ_problem}---where 
one must ensure that
\emph{each} input must map to some \emph{specific} output, but all via
the \emph{same} program---has appeared in other work on unrealizability.
For example, Nay~\cite{nay} uses semi-linear sets over LIA to
capture this relation (albeit for a specific class of synthesis problems), 
while Nope~\cite{nope} constructs a program that executes
each example in lockstep.

In this paper, we show that these relations can be elegantly expressed as
unrealizability triples by \rone merely renaming the set of variables to which each example
refers to be unique; \rtwo conjoining all the renamed states into a single, big state, and 
\rthree modifying the semantics of programs to execute over the renamed variables.
As discussed in Example~\ref{ex:vector_state}, we 
refer to the renamed, conjoined big states as \emph{vector-states} because the renaming and 
conjoining process can be intuitively understood as 
vectorizing the input states, and then executing the program on the single vector.
As shown in Example~\ref{ex:univ_problem}, and as we will later show in \S\ref{Se:Examples}, 
having a simple logical representation allows us to extend the vector-states toward 
infinite examples, which allows us to prove unrealizability for 
problems beyond the reach of previous work~\cite{nay, nope, semgus}.

\begin{example}[Two Input-Output Pairs]
  \label{ex:cartesian2}
  Consider again the synthesis problem $sy_{\mathsf{id}}$ and grammar $G_{\mathsf{id}}$
  from Example~\ref{ex:univ_problem}.
  Suppose that the specification for $sy_{\mathsf{id}}$ is given as a set of two input-output 
  pairs: $\exset{[x = 1 \mapsto x = 1], [x = 2 \mapsto x = 2]}$.
  Then the input and output specifications can both be expressed as the vector-state 
  $\stateset{x_1 = 1 \wedge x_2 = 2}$, where $x_1$ corresponds to the first example 
  and $x_2$ corresponds to the second example.
  The following triple is valid in unrealizability logic, 
  i.e., the two examples suffice to demonstrate that $sy_{\mathsf{id}}$ is unrealizable:
  $$
  \utriple{x_1 = 1 \wedge x_2 = 2}{L(\nstart)}{x_1 \neq 1 \vee x_2 \neq 2}.
  $$
  In turn, this triple may be derived from the triple:
  $$
  \utriple{x_1 = 1 \wedge x_2 = 2}{L(\nstart)}{\exists k. x_1 = k \wedge x_2 = k},
  $$
  which states that there must exist a $k$ for which both $x_1$ and $x_2$ are 
  equivalent to $k$ in the post-state (i.e., that $L(\nstart)$ is a constant program).
  This condition implies that $x_1$ and $x_2$ are identical, and thus implies that 
  $x_1 \neq 1 \vee x_2 \neq 2$.
\end{example}

Example~\ref{ex:cartesian2} illustrates
how unrealizability can be proved using two examples;
in \S\ref{Se:Examples}, we present problems for which proving unrealizability 
requires infinitely many examples.

The fact that the input-output examples can be packed into a \emph{single} 
vector-state is important:
because the starting precondition contains only a single vector-state, 
\textit{all} examples inside the single 
vector-state are guaranteed to be executed on the same program in the grammar.
Likewise, because all the output examples are packed into a single vector-state, the negation 
of this vector-state is guaranteed to contain all vector-states that are wrong on \textit{at least 
one} input example.

\subsection{Why Not Recursive Hoare Logic?}
\label{SubSe:WhyNotHoare}

At this point, readers familiar with Hoare logic extended toward recursive procedures~\cite{10years, nipkow}
may notice that the proof structure described in \S\ref{SubSe:InfiniteSetsOfPrograms} is similar 
to proofs in recursive Hoare logic.
In fact, the similarity between program-synthesis problems and nondeterministic, 
recursive procedures has already been exploited in Nope~\cite{nope}, which 
constructs a recursive program from a synthesis problem, then relies on an 
external verifier to check whether the problem is 
unrealizable.

Then what is the problem with applying recursive Hoare logic to programs translated in this way?
The answer is that some features of synthesis problems are difficult to 
model as a nondeterministic recursive program---e.g., loops and infinite examples 
(both of which Nope cannot support).

With while loops, the problem is that a nondeterministic modelling as described in~\cite{nope} 
must have a way of \emph{recording} the unbounded choices that were selected when 
synthesizing the while-loop body, due to the fact that a loop body must stay \emph{fixed} 
throughout multiple iterations.

\begin{example}[Multiple Loop Bodies]
  \label{ex:loops_bad}
  Consider a variant of $sy_{\ex}$ from Example~\ref{ex:basic_unreal}, 
  where the grammar has been modified to contain while loops instead of 
  sequential composition:
$$
  \begin{array}{ccccc}
    \nstart  \rightarrow  \Ewhile{B}{S}  &&
    B      \rightarrow  x < 100 &&
    S      \rightarrow  \Eassign{x}{x + 2} \mid \Eassign{x}{x + 3} \\
  \end{array}
$$
  $sy_{\ex}$ still consists of programs that repeatedly add either $2$ or $3$ to $x$.
  However, the repeated addition is performed within a loop (until $x \geq 100$).
  The search space consists of exactly two programs:
  one in which the loop body is $\Eassign{x}{x + 2}$, and
  one in which the loop body is $\Eassign{x}{x + 3}$.
  Starting from $\smod{x}{0}{6}$, $\smod{x}{0}{2}$ is an invariant of
  the first program, and $\smod{x}{0}{3}$ is an invariant of the second program.
\end{example}  

\begin{wrapfigure}{r}{0.49\textwidth}
\vspace{-6mm}
\begin{lstlisting}[linewidth=\linewidth]
int x, e_t, b_t; // Vars store results of execution
void Start() { While(); }
void B() { b_t = x < 100; }
void S() {
  int select = nondet(); // Nondeterminsitic choice
  if (select = 1) x = x + 2; else x = x + 3; 
}
void While() {
  B();
  if (b_t) S(); While(); else skip;
}
\end{lstlisting}
\vspace{-4mm}
  \caption{Encoding of $sy_{\ex}$ as a nondeterministic program.}
\vspace{-3mm}    
  \label{fig:while_naive}
\end{wrapfigure}
  Nope~\cite{nope} constructs a nondeterministic, recursive program from a synthesis problem, 
  where each nonterminal is translated into a procedure.
  The basic idea is that executing the translated procedure returns 
  the result of executing that nonterminal---where the 
  nondeterministic choices of an RTG are mimicked by nondeterminsitic choices 
  in the procedure.

  A naive translation of $sy_{\ex}$ into a nondeterminsitic program 
  following this idea 
  would result in a program like the one in Figure~\ref{fig:while_naive}.
  The problem with Figure~\ref{fig:while_naive} is that there is no machinery 
  present to ensure that the \emph{same} loop body is repeated for each iteration---thus, 
  an analysis of 
  Figure~\ref{fig:while_naive}  
  would report that states in which, e.g., $\smod{x}{1}{6}$ holds are possible!
    To address this issue, Figure~\ref{fig:while_naive} 
    would require some data structure capable of recording the 
    $\texttt{nondet()}$ choices in $\texttt{S()}$ directly embedded in the program.
    While a single $\mathsf{Int}$ may suffice for 
    Figure~\ref{fig:while_naive}, as the loop body is decided via a single nondeterministic 
    choice, other synthesis problems in general would require 
    data structures such as lists, as constructing their loop bodies 
    may require the application of multiple productions.
    In particular, 
    such a list must be of \emph{unbounded} size, because, in general, a term within an RTG may be 
  of unbounded size.

%  To address this issue, Figure~\ref{fig:while_naive} 
%  would require, e.g., a list capable of recording the 
%  $\texttt{nondet()}$ choices in $\texttt{S()}$ directly embedded in the program---moreover, 
%  this list must be 
%  \emph{unbounded}, because, in general, a term within an RTG may be 
%  of unbounded size.\footnote{
%    \begin{changebar}[6pt]
%    Although a single int would suffice for Figure~\ref{fig:while_naive}, 
%    this is only because the loop body here can always be fixed through a single production.
%    \end{changebar}
%  } 

  In unrealizability logic, one does not need such complex machinery---instead, 
  one can perform simple invariant-based reasoning.
  Because unrealizability logic is tailored for synthesis problems,
  one can first split the search space through $\grmdisj$, 
  as illustrated in Example~\ref{ex:basic_unreal},
  to consider the two loops ``$\Ewhile{B}{\Eassign{x}{x + 2}}$'' and 
  ``$\Ewhile{B}{\Eassign{x}{x + 3}}$'' separately.
  At this point, it becomes clear that 
  $\utriple{\smod{x}{0}{2}}{\Eassign{x}{x + 2}}{\smod{x}{0}{2}}$ and 
  $\utriple{\smod{x}{0}{3}}{\Eassign{x}{x + 3}}{\smod{x}{0}{3}}$ are invariants for 
  each of the loops.
  An additional application of $\weaken$ yields the desired triple 
  $\utriple{\smod{x}{0}{6}}{\nstart}{\smod{x}{0}{2} \vee \smod{x}{0}{3}}$.
  (While it is possible to reason about both loops via a single invariant, as 
  suggested in \S\ref{SubSe:NoWhileLogic},  
  such an invariant is likely to be too complex; 
  see \S\ref{SubSe:NoWhileLogic} and Lemma~\ref{thm:precision} for details.)

The second problem with trying to apply recursive Hoare logic has to do
with dealing with infinitely many examples, or more generally, with vector-states themselves:
vector-states provide unrealizability logic a way to ensure that multiple examples 
are executed along the same program.
Nope~\cite{nope}  creates copies of 
sets of variables for each different example when translating the 
synthesis problem into a program.
However, creating copies of variables for each example 
requires an \emph{infinite} number of variables in the program to support 
infinite examples.\footnote{
Unlike the case with while loops, where an \emph{unbounded} list suffices 
(e.g., an OCaml list), infinite examples require the use of a 
truly infinite number of variables.
}
Programming languages, much less 
program verifiers, typically do not support---or struggle to support---programs 
with infinite variables,
which is why previous attempts such as Nope fail to support cases with infinite examples.

%However, an \emph{infinite} number of examples under this scheme requires 
%the translated program itself to 
%explicitly use an infinite number of variables 
%(or a similar construct, like infinite arrays).
%However, 

It is true that, if one is willing to develop a logic supporting all necessary 
features: recursion, nondeterminism, both local and global variables, 
infinite data structures, and infinite vector-states,
then it would be possible to prove unrealizability via an extended Hoare logic.
However, to the best of our knowledge, there is no comprehensive 
study of a system containing 
\emph{all} of these features at once, 
whereas this paper proves soundness, relative completeness, 
and some other decidability results related to unrealizability and 
unrealizability logic.

Moreover, such a logical system would be \emph{more} complex than required to 
reason about unrealizability---for example, a record-and-replay encoding of while loops 
as suggested in Example~\ref{ex:loops_bad} completely hides the fact that one is processing 
a while-loop production, and thus \emph{prevents} simple 
invariant-based reasoning about loops.
Such drawbacks are
in direct conflict with the goal of unrealizability logic, 
which---as stated in \S\ref{Se:Introduction}---is to provide a simple logical system 
that distills the essence of proving unrealizability.
True to this goal, the $\whilerule$ rule for loops in unrealizability logic 
does allow one to perform invariant-based reasoning for loops, and provides 
similar direct reasoning principles for other constructs as well.

%%% Local Variables:
%%% mode: latex
%%% TeX-master: "paper.tex"
%%% End:

%% file: 3logic.tex
\section{Unrealizability Logic}
\label{Se:UnrealizabilityLogic}

In this section, we formally define unrealizability logic and some necessary 
preliminaries.

\subsection{Preliminary Definitions}

In this paper, we consider synthesis problems where the search space 
is a subset of all terms producible by the imperative grammar 
$G_{imp}$ (Figure~\ref{fig:base_grammar})---that is, 
deterministic, expression-based or imperative synthesis problems defined over 
integer arithmetic expressions.
We assume a 
function $\semantics{\cdot}$ that defines the standard semantics of every term $t \in L(G)$.
We leave supporting synthesis problems involving more complex operations, e.g., 
recursive data structures, to future work;
this paper lays the foundations for making such extensions possible.

\begin{figure}[bt!]
  {\small
	\[
	\begin{array}{llcl}
		\mathit{Stmt} & S & ::= &
			\Eassign{x}{E} \mid \Eseq{S}{S} \mid \Eifthenelse{B}{S}{S}  \mid \Ewhile{B}{S}\\
		\mathit{IntExpr} & E & ::= &
			f \mid 1 \mid x \mid E + E \mid E - E \mid E \cdot E \mid E / E \mid \Eifthenelse{B}{E}{E}\\
		\mathit{BoolExpr} & B & ::= &
			\Et \mid \Ef \mid \lnot B \mid B \wedge B \mid E < E \mid E == E
	\end{array}\]}
	\vspace{-2mm}
  \caption{The base grammar $G_{imp}$, which defines a universe of terms that we are 
  interested in for this paper.}
  \label{fig:base_grammar}
\end{figure}

\def\alphabet{\mathcal{A}}
\def\nonterminals{\mathcal{N}}

\begin{definition}[Regular Tree Grammar]
A \emph{(typed) regular tree grammar} (RTG) is a tuple $G = (\nonterminals,\alphabet,S,a,\delta)$, 
where $\nonterminals$ is a finite set of nonterminals;
$\alphabet$ is a ranked alphabet; $S \in \nonterminals$ is an initial nonterminal; 
$a$ is a type assignment that gives types for members of $\alphabet \cup \nonterminals$; 
and $\delta$ is a finite set of productions of the form
$N_0 \rightarrow \alpha^{(i)}(N_1,...,N_i)$, 
where for $0 \leq j \leq i$, each $N_j \in \nonterminals$ is a nonterminal 
such that if $a_{\alpha^{(i)}} = (\type_0,\type_1,...,\type_i)$ then $a_{N_j} = \type_j$.
\end{definition}

In our setting, the alphabet consists of 
constructors for each of the constructs of $G_{imp}$ 
(although for simplicity we write, e.g., $\Eassign{x}{0}$, rather 
than ${:=^{(2)}}(x^{(0)},0^{(0)})$).
$G_{imp}$ contains three types of nonterminals:
statement nonterminals, expression nonterminals, and Boolean-expression nonterminals.

We assume grammars use productions that differ from the ones in 
$G_{imp}$ only due to the nonterminal names.
We say that a production $N_0 ::= \alpha^{(i)}(N_1, \ldots, N_i)$ 
is \textit{valid with respect to} $G_{imp}$ iff 
replacing each nonterminal $N_j$ with the nonterminal of the same type 
in the set $\{S,E,B\}$ yields a production in $G_{imp}$,
e.g., production $S_1::= \Eseq{S_2}{S_3}$ is valid with respect to $G_{imp}$.\footnote{
In this paper, we sometimes write a production in which some terminals are in-lined in place of 
nonterminals---e.g., we write $S::=\Eassign{x}{x}$ 
in place of the two productions $S::=\Eassign{x}{E_x}, E_x::=x$.}

\begin{definition}[Synthesis Problem]
  \label{def:program_synthesis}
  A synthesis problem is a $4$-tuple $\textit{sy} = \langle G, f, I, \psi \rangle$, 
  where
  \begin{itemize}
    \item
      $G$ is a regular-tree grammar consisting of productions that 
      are valid with respect to $G_{imp}$,
    \item
      $f$ is the name of the function to synthesize,
    \item
      $I$ is the set of allowed input states of $f$ (i.e., the domain of $f$),
    \item
      $\psi$ is a Boolean formula that describes a behavioral specification that $f$ must satisfy.
  \end{itemize}
  The goal of a synthesis problem is to find $f^* \in L(G)$ such that 
  $\forall\state \in I.\ \psi(\semantics{f^*}(\state),\state)$.
  If such an $f^* \in L(G)$ exists, synthesis problem $sy$ is \emph{realizable};
  otherwise, $sy$ is \emph{unrealizable}.
\end{definition}
In this paper, $f$ may be thought as of a state transformer that has type of 
$\mathsf{State} \rightarrow \mathsf{State}$.

A key point from Definition~\ref{def:program_synthesis} is that the search space 
is defined as an RTG.
As noted in \S\ref{SubSe:InfiniteSetsOfPrograms}, this property is used 
by unrealizability logic to build proof trees that mimic structural induction.
  Also, note that the specification $\psi$ is given as a Boolean formula, i.e., a logical
  specification.
  So far in the paper,
  we have only considered synthesis problems where the specification 
  is given as a set of input-output examples, which are subsumed by 
  logical specifications.
  However, unrealizability logic is also capable of dealing with logical specifications 
  directly, because the pre- and postcondition of an unrealizability triple are 
  logical predicates.
%The fact that the search space is given as a RTG also guides the
%development of the inductive inference rules in
%\S\ref{SubSe:NoWhileLogic}.

\subsection{Basic Definitions for Unrealizability Logic}

We now introduce some concepts that are specific to unrealizability logic.

In this paper, we will use (and have used) the following notation for different 
kinds of states:
\begin{itemize}
  \item Lowercase letters (e.g., $p, q, \state$) for ordinary states,
  \item Uppercase letters (e.g., $P, Q$) for sets of ordinary states 
    (and also predicates over states),
  \item Lowercase Greek letters (e.g., $\sigma, \pi$) for vector-states 
    (as shown in \S\ref{Se:MotivatingExample}, defined in Definition~\ref{def:vectorstate}),
  \item Calligrahic uppercase letters (e.g., $\aP, \aQ$) 
    for sets of (and predicates describing) vector-states.
\end{itemize}
Sometimes, uppercase letters are used also to denote nonterminals, i.e., sets of programs, 
and lowercase letters used to denote single programs.
It should be clear from context, in these cases, what the upper- and lowercase letters denote.

In unrealizability logic, the semantics of expressions are extended such that an expression 
does not evaluate to a value $v$, but evaluates to a \emph{state} that stores 
the value $v$ inside a (reserved) auxiliary variable $e_t$ (or $b_t$, for 
Boolean expressions).
This non-standard approach is taken because, as we will see in
\S\ref{SubSe:NoWhileLogic}, simply generating values will lose some
information that is required to precisely compute the set of states
that a nonterminal can generate.

\begin{definition}[Semantics of Expressions]
  \label{def:alt_semantics}
  Let $\semext{\cdot}$ be a semantics for expressions that 
  evaluates an expression to a state instead of a value.
  Intuitively, we will store the value resulting from computing an expression 
  to a reserved variable $e_t$ in the state.

  Given a `standard' semantics $\semantics{\cdot}$  
  that evaluates expressions to values and an arbitrary state $\state$,
  let $\semext{e}(\state) = \state[e_t \mapsto \sem{e}(\state)]$ if 
  $e$ is an atomic expression ($0, 1,$ or $x$).
  (If $e$ is an atomic Boolean expression, the assignment is to $b_t$ 
  instead.)

  If $e$ is a $n$-ary operation, with $e = e_1 + e_2$ as an example,
  let $\semext{e_1 + e_2}(\state) = 
  \state[e_t \mapsto \semext{e_1}(\state)[e_t] + \semext{e_2}(\state)[e_t]]$; 
  that is, $\semext{\cdot}$ is defined recursively such that 
  the extended semantics of operators such as $+$, i.e., $\semext{+}$, 
  are state transformers as well.
  These semantics work
  by referencing the value of $e_t$ stored in the state that 
  each sub-expression $e_1$ and $e_2$ evaluates to 
  (likewise, if $e$ is a $n$-ary Boolean operation, the final assignment is to $b_t$ instead; 
  Boolean sub-expressions will also reference $b_t$ instead of $e_t$).
\end{definition}

We define the extended semantics recursively, instead of simply stating that 
$\semext{e}(\state) = \state[e_t \mapsto \sem{e}(\state)]$, because we wish the 
rules of unrealizability logic to be recursive, and thus need to define operations 
such as addition on states.
In the rest of this paper, we will take $\semext{\cdot}$ from Definition~\ref{def:alt_semantics} to 
be the standard semantics given to a term, and thus drop the subscript $\mathsf{ext}$.

In this paper, we consider a state to be a finite mapping from variables to values,
which can also be understood as a finite set of $\langle$variable, value$\rangle$ pairs 
(where the variables are unique).

As we showed in \S\ref{Se:MotivatingExample}, unrealizability logic relies on 
\emph{vector-states} and a semantics modified to run on vector-states
to capture input-output relations.
Vector-states may simply be understood as ordinary states  
defined over an extended set of variables; we formalize this notion below.

\begin{definition}[Vector-State]
  \label{def:vectorstate}
  Let $\state_1, \state_2, \cdots $ be a finite or countably infinite collection of 
  states indexed by the natural numbers, that are defined over the 
  same set of variables $V$.
  Then the vector-state $\sigma$ is a state defined over the set of variables 
  $\bigcup_{i \in D} \set{v_i \mid v \in V}$, 
  such that for all $v \in V$ and $i \in D$, we have that $\sigma(v_i) = \state_i(v)$.
%  Let $\state_1, \cdots, \state_n$ 
%  be a collection of states defined over the same set of variables $V$.
%  A \emph{vector-state} $\sigma$ over $\state_1, \cdots, \state_n$, denoted by
%  $\sigma = \lrangle{\state_1, \cdots, \state_n}$, is a state defined over 
%  the set of variables $\bigcup_{1 \leq i \leq n} V_i$, where
%  \begin{itemize}
%    \item
%      $V_i = \set{v_i \mid v \in V}$ is the set of variables in which
%      each variable in $V$ is subscripted with symbol $i$.
%    \item
%       For all $v \in V$ and $1 \leq i \leq n$,
%       $\sigma(v_i) = \state_i(v)$.
%  \end{itemize}
%  A vector-state may also be defined over an infinite collection of states indexed by 
%  the elements of an (infinite) domain $D$, in which case $\sigma$ is a state 
%  defined over the set of variables $\bigcup_{i \in D} V_i$, for 
%  $V_i = \set{v_i \mid v \in V}$, and for all $v \in V$ and $i \in D$, 
%  $\sigma(v_i) = \state_i(v)$.
\end{definition}
%Although it would be possible to define uncountably infinite vector-states, 
%in this paper we will only be concerned with vector-states defined over a 
%countable collection of states as in Definition~\ref{def:vectorstate}.
%Although it is possible for $D$ to be uncountable, in this paper we will only 
%be concerned with vector-states defined over an countably infinite set of states.
%In this case, $D = \mathbb{N}$ and we will 
%denote the vector-state as $\sigma = \lrangle{r_1, r_2, \cdots}$, 
%indexing each individual state by the natural numbers.

Intuitively, given a finite collection of states $\state_1, \cdots, \state_n$, 
the vector-state $\sigma = \lrangle{\state_1[v_1 / v], \cdots, \state_n[v_n / v]}$ is 
simply a `vectorization' of the renamed individual states. 
Executing a program $t$ on this vector-state produces a vector-state 
equivalent to the vector obtained by running each individual state 
through $t$ separately: that is, 
$\semantics{t}(\sigma) = \lrangle{\semantics{t}(\state_1)[v_1 / v], \cdots, 
\semantics{t}(\state_n)[v_n / v]}$.
This intuition carries over
to when the given collection of states is infinite, in which case 
the resulting vector-state will simply be of infinite length.
Def.~\ref{def:vectorsemantics} formalizes this intuition, 
and extends the program semantics to a semantics 
that works on possibly infinite vector-states, and also to a semantics of a set of programs 
on a set of possibly infinite vector-states.

%(This idea is reflected in Def.~\ref{def:vectorsemantics} below.)
%We now extend the standard program semantics to a semantics that
%operates on vector-states, and further extend that definition to work
%with sets of vector-states and programs, as well.

%Definition~\ref{def:vectorsemantics} gives a definition of
%program semantics on vector-states, as well as a semantics for
%sets of  programs and sets of vector-states.
\begin{definition}[Vector-State Semantics]
  \label{def:vectorsemantics}
  Let $\sigma = \lrangle{\state_1[v_1 / v], \state_2[v_2 / v], \cdots}$ be a vector-state 
  indexed by the natural numbers, 
  where $\state_1, \state_2, \cdots$ are defined over variables $v \in V$.

  Then the semantics of a term $t$ on
  ${\sigma}$ is defined as 
  $\semantics{t}(\sigma) = \bigcup_{i \in N, v \in V} 
  \set{v_i \mid \sigma(v_i) = \semantics{t}(\state_i)(v)}$.
%
%  Intuitively, if $\sigma = \lrangle{\state_1, \cdots, \state_n}$, then 
%  $\semantics{t}(\sigma) = \lrangle{\semantics{t}(\state_1), \cdots, \semantics{t}(\state_n)}$.
%
  The semantics of $t$ on a \emph{set} 
  of vector-states $\aP$ is defined as 
  $\semantics{t}(\aP) = \bigcup_{\sigma \in \aP} \{\semantics{t}(\sigma)$\}.
  The semantics of a \emph{set} of programs $T$ on a \emph{set} of vector-states $\aP$ 
  is defined as
  $\semantics{T}(\aP) = \bigcup_{t \in T} \semantics{t}(\aP)$. 
\end{definition}
The semantics of a set of programs on a set of input states 
can be understood as producing the set of all states that 
may arise from taking \emph{any} combination of a program $t \in T$ and a 
(vector-) state $\sigma \in \aP$.
This semantics allows us to define the validity of unrealizability triples:
a triple $\utriple{\aP}{T}{\aQ}$ 
is valid iff $\aQ$ overapproximates the states that may result from 
any combination of $\sigma \in \aP$ and $t \in T$.

\begin{definition}[Validity]
  \label{def:triple}
  Given a precondition over vector-states $\aP$, and a postcondition over vector-states $\aQ$, 
  an unrealizability triple
  $\utriple{\aP}{T}{\aQ}$ is \emph{valid} for a set of programs $T$,
  denoted by $\models \utriple{\aP}{T}{\aQ}$, iff $\semantics{T}(\aP) \subseteq \aQ$.
\end{definition}

%As stated above, that a triple $\utriple{\aP}{T}{\aQ}$ 
%is valid means that 
%$\aQ$ is an overapproximation of the result of $\semantics{T}(\aP)$.
As previously stated as Theorem~\ref{thm:unreal_hoare},
$\models \utriple{\aP}{T}{\aQ}$ iff for all $t \in T$, $\triple{\aP}{\stmt}{\aQ}$
is a valid Hoare triple.
Connecting all the definitions, we state that:
\begin{mybox}
A synthesis problem $\textit{sy} = \langle G, f, I, \psi \rangle$ is unrealizable
if $\models \utriple{I(\state)}{N}{\neg\psi(f(\state),\state)}$ 
is a valid unrealizability triple.
Here, $N$ is the initial nonterminal of $G$, and $I$ and $\psi$ are predicates
that describe the vectorized input and output conditions
of the synthesis problem.
\end{mybox}

\subsection{The Rules of Unrealizability Logic}
\label{SubSe:NoWhileLogic}

In this section, we present the inference rules of unrealizability logic.

\mypar{The Assertion Language}
One issue to discuss
before presenting the rules is the choice of 
\emph{assertion language}---that is, the language used for the pre- and 
postconditions of unrealizability triples.
Unrealizability logic is parametric in the choice of assertion language, 
as long as the assertion language is as least as expressive 
as first-order Peano arithmetic (FO-PA).
This requirement is due to the fact that the addtional 
predicates that unrealizability logic 
adds to the given pre- and postconditions when constructing a proof tree 
require FO-PA---an assertion language less expressive than FO-PA will fail to encode 
these additional predicates. 
% making the language incompatible with unrealizability logic.

This characteristic makes FO-PA a natural choice for the assertion language 
in unrealizability-logic proofs, especially as it is well-studied and 
provides opportunities for automation.
However, some synthesis problems will require a stronger assertion language 
than FO-PA for a proof of unrealizability to be completed in unrealizability logic.
For example, proofs that require an infinite number of 
examples will require a logic capable of supporting an infinite number of 
variables (such an example is given in \S\ref{Se:DealingWithAnInfiniteNumberOfExamples}).
One example of such a logic is FO-PA extended with infinite arrays.

\begin{example}[FO-PA with Infinite Arrays]
  Consider the assertion language of FO-PA, extended with axioms from the 
  first-order theory of arrays~\cite{mccarthyarray}. 
  Variables may be arrays ranging over an 
  infinite range of indices (indexed by the natural numbers), 
  and $\vec{x}[i]$ indicates reading the array $\vec{x}$ at index $i$.

  Then one can encode vector-states of infinite length by introducing an 
  array for each variable:
  the array $\vec{x}$ encodes an infinite-length 
  vector-variable $\lrangle{x_1, x_2, \cdots}$, 
  where $\vec{x}[i]$ corresponds to the variable $x_i$.
  For example, $\forall i. \vec{x}[i] = i$ encodes the condition that 
  $x_i = i$, i.e., that the value of $x$ in the $i$-th example is $i$ 
  (we will use the two notations $x_i$ and $\vec{x}[i]$ interchangably in the rest of this paper).
\end{example}

As stated above, unrealizability logic is parametric in the choice 
of assertion language; if one wishes to write proofs that necessitate the 
use of predicates that cannot be expressed in FO-PA, it is entirely possible 
for one to use an assertion language capable of expressing the required 
predicates.
The examples used in this paper rely on two assertion languages: standard FO-PA for 
proofs that do not require infinite examples, and 
FO-PA extended with infinite arrays for cases that do require 
an infinite number of examples to prove unrealizability (Example~\ref{ex:inf_example}).

%The additional predicates generated by the rules of unrealizability logic 
%itself are within first-order Peano arithmetic (FO-PA), which means
%that one must choose an assertion language at least as expressive 
%as FO-PA; if not, the predicates occurring in a proof tree may be beyond the 
%power of the assertion language.
%If one does choose an assertion language at least as expressive as FO-PA, 
%then the assertions occurring in a proof tree are guaranteed to stay within 
%that language.
%Unrealizability logic itself is parametric in the choice of assertion language, 
%as long as it is at least as expressive as FO-PA---e.g., one may choose 
%to use FO-PA, or use a logic that extends FO-PA with
%infinite variables to model infinite input examples 
%(as we will do in \S\ref{Se:Examples}, Example~\ref{ex:inf_example}).

%In this section, we consider unrealizability logic for programs that do not 
%contain while loops.
The goal of unrealizability logic is to prove unrealizability for 
synthesis problems defined over an RTG.
Therefore, the rules that we state in Figures \ref{fig:exp_rules}, 
\ref{fig:stmt_rules}, and \ref{fig:struct_rules}
are for sets of programs defined as $L(N)$ for some RTG nonterminal $N$, 
or as the language of the right-hand side of an RTG production.

%NOTE: Delete part about predicated execution, unneeded now
%One thing to notice from Figure~\ref{fig:rules} is that the nonterminals, 
%as well as some of the operators, such as $=$, are labeled with the superscript $\fb$.
%Informally, $\fb$ is a Boolean vector, and the superscript expresses
%an idea similar to \emph{predicated execution} on vector-states: a predicated program $\stmt^{\fb}$ executed on a 
%vector-state $\lrangle{\state_1, \cdots \state_n}$ will only 
%execute $\stmt$ on each $\state_i$ for which $\beta_i$ is $\Et$ (true), and leave the other states unchanged.
%Similarly, the operator $\eqb$ over vector-variables and vector-values 
%only generates a predicate for indices where $\fb$ is $\Et$.
%Predicated programs and their semantics are required to deal with branches and loops (see Definition~\ref{def:proj-semantics}).
%For all other rules, the intuition above suffices for explanatory purposes.

% All places where $\fb$ is omitted are places where $\fb = \vec{\Et}$;
% i.e., every state in the vector-state is executed.

\subsubsection{Rules for Expressions}

  {\footnotesize
  \begin{figure}[bt!]
	$$ 
  \begin{array}{c}
    \infer[\mathsf{Zero}]{\Gamma \vdash \utriple{\aP}{0}
    {\exists \vec{e_t '}. \aP[\vec{e_t '} / \vec{e_t}] \wedge \vec{e_t} \eq \vec{0}}}
		{} \quad
    \infer[\mathsf{One}]{\Gamma \vdash \utriple{\aP}{1}
    {\exists \vec{e_t '}. \aP[\vec{e_t '} / \vec{e_t}] \wedge \vec{e_t} \eq \vec{1}}}
		{} \\[3mm]
    \infer[\mathsf{True}]{\Gamma \vdash \utriple{\aP}{\Et}
    {\exists \vec{b_t '}. \aP[\vec{b_t '} / \vec{b_t}] \wedge \vec{b_t} \eq \vec{\Et}}}
    {} \quad
    \infer[\mathsf{False}]{\Gamma \vdash \utriple{\aP}{\Ef}
    {\exists \vec{b_t '}. \aP[\vec{b_t '} / \vec{b_t}] \wedge \vec{b_t} \eq \vec{\Ef}}}
    {} \\[3mm]
		\infer[\mathsf{Var}]{\Gamma \vdash \utriple{\aP}{x}
    {\exists \vec{e_t '}. \aP[\vec{e_t '} / \vec{e_t}] \wedge \vec{e_t} \eq \vec{x}}}
		{} \quad
    \infer[\mathsf{Not}]
    {\Gamma \vdash 
      \utriple{\aP}{!B}{\exists \vec{b_t '}. 
      \aQ[ \vec{b_t '} / \vec{b_t}] \wedge \vec{b_t} = \lnot \vec{b_t '}}}
    {
      \Gamma \vdash \utriple{\aP}{B}{\aQ}
    }
    \\[3mm]

    \infer[\mathsf{Bin}]
    {\Gamma \vdash \utriple{\aP}{(E_1 \oplus E_2)}
    {
      \exists \vec{e_t'}, \vec{v_1}, \vec{v_2}, \vec{v_1'}, \vec{v_2'}
      .
       (      \aP
       \wedge \aQ_1[\vec{v_1'} / \vec{v}]
       \wedge \aQ_2[\vec{v_2'} / \vec{v}]
       \wedge (\vec{v} = \vec{v_1} = \vec{v_2}))[\vec{e_t'} / \vec{e_t}]
       \wedge \vec{e_t} = \vec{e_{t_1}'} \oplus \vec{e_{t_2}'}
    }}
    {
      \begin{array}{l}
      \Gamma \vdash \utriple{
        \aP \wedge (\vec{v_1} = \vec{v}) }
      {E_1}{\aQ_1} \\
      \Gamma \vdash \utriple{
        \aP \wedge (\vec{v_2} = \vec{v})}
      {E_2}{\aQ_2}
      \end{array} \quad
      \begin{array}{l}
        \vec{v_1}, \vec{v_2}, \vec{v_1'}, \vec{v_2'} \ \text{fresh renames of} \ \vec{v} \\
        \vec{v} \ \text{refers to full set of vars in} \ \aP
      \end{array}
    }\\[3mm]

    \infer[\mathsf{Comp}]
    {\Gamma \vdash \utriple{\aP}{(E_1 \odot E_2)}
    {
      \exists \vec{b_t'}, \vec{v_1}, \vec{v_2}, \vec{v_1'}, \vec{v_2'} . 
      (\aP \wedge \aQ_1[\vec{v_1'} / \vec{v}] \wedge \aQ_2[\vec{v_2'} / \vec{v}] \wedge 
      (\vec{v} = \vec{v_1} = \vec{v_2}))[{\vec{b_t'} / \vec{b_t}}] \wedge 
      \vec{b_t} = {\vec{e_{t_1}'} \odot \vec{e_{t_2}'}}
    }}
    {
      \begin{array}{l}
      \Gamma \vdash \utriple{
        \aP \wedge (\vec{v_1} = \vec{v}) }
      {E_1}{\aQ_1} \\
      \Gamma \vdash \utriple{
        \aP \wedge (\vec{v_2} = \vec{v})}
      {E_2}{\aQ_2}
      \end{array} \quad
      \begin{array}{l}
        \vec{v_1}, \vec{v_2}, \vec{v_1'}, \vec{v_2'} \ \text{fresh renames of} \ \vec{v} \\
        \vec{v} \ \text{refers to full set of vars in} \ \aP
      \end{array}
    } \\[3mm]
    \infer[\mathsf{And}]
    {\Gamma \vdash \utriple{\aP}{(B_1 \wedge B_2)}
    {
      \exists \vec{b_t'}, \vec{v_1}, \vec{v_2}, \vec{v_1'}, \vec{v_2'} . 
      (\aP \wedge \aQ_1[\vec{v_1'} / \vec{v}] \wedge \aQ_2[\vec{v_2'} / \vec{v}] \wedge 
      (\vec{v} = \vec{v_1} = \vec{v_2}))[\vec{b_t'} / \vec{b_t}] \wedge 
      \vec{b_t} = \vec{b_{t_1}'} \wedge \vec{b_{t_2}'}
    }}
    {
      \begin{array}{l}
      \Gamma \vdash \utriple{
        \aP \wedge (\vec{v_1} = \vec{v}) }
      {B_1}{\aQ_1} \\
      \Gamma \vdash \utriple{
        \aP \wedge (\vec{v_2} = \vec{v})}
      {B_2}{\aQ_2}
      \end{array} \quad
      \begin{array}{l}
        \vec{v_1}, \vec{v_2}, \vec{v_1'}, \vec{v_2'} \ \text{fresh renames of} \ \vec{v} \\
        \vec{v} \ \text{refers to full set of vars in} \ \aP
      \end{array}
    }  
  \end{array}
	$$
\caption{Inference rules for expressions in unrealizability logic. $\mathsf{Bin}$ 
represents rules for binary expressions, where the operator $\oplus$ is one of 
$+$ ($\plusrule$), $-$ ($\minusrule$), $\cdot$ ($\multrule$), or $/$ ($\divrule$).
$\mathsf{Comp}$ represents rules for binary comparators, where the operator 
$\odot$ is one of $<$ ($\ltrule$) or $==$ ($\eqrule$).
}
\label{fig:exp_rules}
  \end{figure}
}

Recall that in Definition~\ref{def:alt_semantics} we defined the semantics of 
expressions to update the value of a reserved auxiliary variable ${e_t}$ instead of 
generating a value.
The rules for expressions reflect this fact: in every 
expression rule, 
the postcondition of the conclusion takes the form 
$\exists \vec{e_t'}. \aP[\vec{e_t '} / \vec{e_t}] \wedge \vec{e_t} = \vec{e}$, 
where $\aP$ is some predicate and $\vec{e}$ is a vector 
that encodes the values that may be generated by an expression in the set of expressions considered.
Observe that the form of this predicate mimics the postcondition in the 
forwards-assignment rule of Hoare logic~\cite{floyd} applied to
  $\vec{e_t} = \vec{e}$ (e.g., $\vec{e_t} = \vec{0}$ in the case of $\mathsf{Zero}$):
the assignment is to a vector-variable because we are working with vector-states
(as opposed to single states, as in Definition~\ref{def:alt_semantics}).

The intuition that $\vec{e_t}$ contains the value generated by the expression nonterminals
can be formalized into an invariant about expression nonterminals:
\begin{lemma}
  \label{lem:expression_invariant}
  Given an expression nonterminal $E$,
  if an unrealizability triple $\Gamma \vdash \utriple{\aP}{E}{\aQ}$ is derivable using the rules of unrealizability logic 
  under some context $\Gamma$,
  then for every $\exp\in L(E)$ and for every $\sigma \in \aP$, 
  the formula $\aQ[\semantics{\exp}(\sigma)/\vec{e_t}]$ is true 
  assuming all triples in $\Gamma$ are true
  (where in this case, $\semantics{\exp}(\sigma)$ refers to the standard semantics that evaluates 
  $\exp$ to a value).
\end{lemma}
Lemma~\ref{lem:expression_invariant} shows that, 
$\aQ$ at the very least contains states where $\vec{e_t} = \semantics{\exp}(\sigma)$: the variable $\vec{e_t}$ 
overapproximates the set of values obtainable by executing an expression $\exp \in L(E)$.
A similar lemma also holds for Boolean-expression nonterminals,
where the variable $\vec{e_t}$ is replaced with $\vec{b_t}$.

Lemma~\ref{lem:expression_invariant} is important, 
because it allows the postcondition of the conclusion to refer to the 
values computed by expression nonterminals through use of the variable $\vec{e_t}$ (or $\vec{b_t}$).
In particular, we assume that $\vec{e_t}$ and $\vec{b_t}$ 
are \emph{reserved} for the purposes of storing these values.
(More precisely, we assume that $\vec{e_t}$ and $\vec{b_t}$ are included in the set of 
program variables $\vars(N)$, and thus rules such as $\subone$ cannot
substitute these variables.)

Figure~\ref{fig:exp_rules} presents the rules for each expression-based operator in unrealizability logic.
Having articulated the key intuition behind the rules, let us now look at each rule in detail.

\mypar{0-ary operators ($0, 1, x, \Et, \Ef$)}
Rules for 0-ary operators 
do not rely on the result of another set of programs, and can be computed directly by 
assigning the correct values ($0, 1, x,$ etc.) to the auxiliary (vector-)variable $\vec{e_t}$.
The symbol $\vec{x}$ in the rule $\mathsf{Var}$ refers to the vector of values generated by taking $x$ from 
each sub-state of $\sigma = \lrangle{\state_1, \state_2, \cdots}$;
i.e., $\langle \state_1(x_1), \state_2(x_2), \cdots \rangle$.

\mypar{Unary Operators ($!B$)}
$G_{imp}$ only contains a single instance of a unary operator, Not ($!$), which 
takes the value produced by a nonterminal $B$ and negates it.
The rule $\mathsf{Not}$ mirrors this behavior: 
it first requires that the behavior of the nonterminal $B$ is captured by the premise 
$\utriple{\aP}{B}{\aQ}$.

Following Lemma~\ref{lem:expression_invariant}, the result of executing $B$ is 
stored in the vector-variable $\vec{b_t}$; that is, by referring to the 
vector-variable $\vec{b_t}$ in $\aQ$, one can refer to the (set of possible) 
values created by $B$.
The conclusion of the rule negates the value in $\vec{b_t}$ 
and re-assigns it to $\vec{b_t}$ (captured in the postcondition predicate 
$\exists \vec{b_t '}. \aQ[ \vec{b_t '} / \vec{b_t}] \wedge \vec{b_t} = \lnot \vec{b_t '}$).

\mypar{Binary Operators ($E_1 + E_2, E_1 - E_2, E_1 \cdot E_2, E_1 / E_2$, $B_1 \wedge B_2$, $E_1 < E_2$, $E_1 == E_2$)}
Binary operators operate in a similar manner to unary operators in that 
they rely on the premise triples to refer to the values generated by the 
sub-nonterminals $E_1$ and $E_2$.
However, they also pose a unique challenge in that one needs to be careful 
in how these values are composed, as (partially) described in 
\S\ref{Se:MotivatingExample}.

Take $\plusrule$ as an example.
A naive version of $\plusrule$ might be written as following:
$$
  \infer[\plusrule\textsf{-bad}]{\utriple{\aP}{E_1 + E_2}{\aQ_1 + \aQ_2}}
  {
    \utriple{\aP}{E_1}{\aQ_1} \quad
    \utriple{\aP}{E_2}{\aQ_2}
  }
$$
The problem arises in resolving the set corresponding to $\aQ_1 + \aQ_2$: 
the simplest way is to generate pairs by 
taking the Cartesian product of the two sets, then add the values in each 
pair to produce a set of values.
However, this approach is imprecise, because it also generates pairs generated by \emph{different} 
vector-states in the precondition.
\begin{example}[Plus-bad]
  \label{ex:why_not_plus}
  Consider an application of the $\plusrule$ rule where 
  $E_1 ::= x$, $E_2 ::= y$, and the
  input precondition is $\aP = \set{\lrangle{(1, 3), (2, 4)}, \lrangle{(3, 1), (4, 2)}}$;
  i.e., there are two possible two-example vector-states:
  $(\vx[1], \vy[1]) = (1, 3) \wedge (\vx[2], \vy[2]) = (2, 4)$ and 
  $(\vx[1], \vy[1]) = (3, 1) \wedge (\vx[2], \vy[2]) = (4, 2)$,
  where $\vec{v}[i]$ denotes the $i$-th entry of $\vec{v}$.
  (We omit $e_t$ in input states; its value is irrelevant.)

  Because the only term in $E_1 + E_2$ is $x + y$, the postcondition 
  should only contain states in which $e_t = \lrangle{4, 6}$.
  However, because there are two possible input states in $\aP$, the predicate
  $\aQ_1$ contains two states---one in which $e_t = \lrangle{1, 2}$ and another
  in which $e_t = \lrangle{3, 4}$.
  Likewise, the predicate $\aQ_2$ contains two states, one in which $e_t$ is $\lrangle{3, 4}$ and
  one in which it is $\lrangle{1, 2}$.
  If one takes the Cartesian product of the states in $\aQ_1$ and 
  $\aQ_2$, then the predicate $\aQ$ ends up having $3$ 
  possible values of $e_t$: $\lrangle{4, 6}$, $\lrangle{2, 4}$, or $\lrangle{6, 8}$.
  The problem is that the two possible input states $\lrangle{(1, 3), (2, 4)}$ and 
  $\lrangle{(3, 1), (4, 2)}$ are 
  \emph{mixed} by taking the Cartesian product: for example, 
  $x = \lrangle{1, 2}$ from the first input vector-state and $y = \lrangle{1, 2}$ from the 
  second input vector-state are added, these come from different input states!
\end{example}

To remedy this problem, in the $\plusrule$ rule in unrealizability logic, 
  each original state in the precondition 
  is \emph{appended} with a \emph{copy} of itself, 
  where we refer to the copied part as the `\emph{ghost (vector)-state}', 
  and the original part as the original (vector)-state.
This idea is captured in the preconditions of the premises, e.g.,
$\aP \wedge (\vec{v_1} = \vec{v})$;
by taking a set of fresh variables $\vec{v_1}$ and setting $\vec{v_1} = \vec{v}$, 
one has essentially extended each state in $\aP$ with a copy of each variable.

\begin{example}[Extended States]
  \label{ex:extended_state}
  Consider the precondition $\aP = \set{\lrangle{(1, 3), (2, 4)}, \lrangle{(3, 1), (4, 2)}}$, the 
  same set of vector-states used in Example~\ref{ex:why_not_plus}.
  Observe that $\aP \wedge (\vec{v_1} = \vec{v})$ 
  is the set of vector-states 
  $(\vx[1], \vy[1], \vx_1[1], \vy_1[1], \vx[2], \vy[2], \vx_1[2], \vy_1[2]) = 
  \set{\lrangle{(1, 3, 1, 3, 2, 4, 2, 4)}, \lrangle{(3, 1, 3, 1, 4, 2, 4, 2)}}$; note that 
  states such as $\lrangle{(1, 3, 3, 1, 2, 4, 4, 2)}$ 
  are not included in $\aP \wedge (\vec{v_1} = \vec{v})$
  because they violate the conjunct ``$(\vec{v_1} = \vec{v})$.''
\end{example}

Because the variables of the ghost state have been \emph{renamed}, 
execution through a program in $E$ will leave these variables \emph{unchanged} 
according to the semantics of $E$.
Thus, one can say that a state in, e.g., $\aQ_1$ is an extended state---where the 
`original' part of the state has executed through $E$, and the `ghost' part 
is left unchanged.
By this device, one can check \emph{via the ghost parts} whether two
extended states from $\aQ_1$ and $\aQ_2$ should be added.
Because the ghost-state portion of each extended state remains unchanged,
one can deduce that two extended states must be added 
if and only if the ghost-state portions of the states are identical.

In the $\plusrule$ rule, the postcondition of the conclusion illustrates the
idea of matching vector-states on ghost-state portions.
In particular, the constraint $\vec{v_1} = \vec{v_2}$ filters out those
pairs of extended states in which the ghost-state portions are different!
The vocabulary shifts $\aQ_1[\vec{v_1'} / \vec{v}]$ and $\aQ_2[\vec{v_2'} / \vec{v}]$
move the `original' 
program variables in $\aQ_1$ and $\aQ_2$ to $\vec{v_1'}$ and $\vec{v_2'}$, respectively. 
The values of the original variables are unneeded, 
except for the value of $\vec{e_t}$ inside these states
(available as $\vec{e_{t_1}'}$ and $\vec{e_{t_2}'}$, respectively, after the vocabulary shifts.)
The conjunction with $\aP \land (\vec{v} = \vec{v_1} = \vec{v_2})$
restores the values of the original program variables.
The value of $\vec{e_t}$ is established by the conjunction 
with $\vec{e_t} = \vec{e_{t_1}'} + \vec{e_{t_2}'}$.
Finally, all of the additionally introduced variables are quantified out, leaving
(vector-)states over the original variables, plus $\vec{e_t}$.

\begin{example}[Correcting Plus with Extended States]
  \label{ex:correcting_plus}
  Consider again the grammar from Example~\ref{ex:why_not_plus}.
  In this example, we will remove the second example from Examples~\ref{ex:extended_state}
  and~\ref{ex:why_not_plus} for simplicity: that is, 
  the input precondition $\aP$ consists of the length-1 vector-states 
  $\set{\lrangle{(1, 3)}, \lrangle{(3, 1)}}$.

  As illustrated in Example~\ref{ex:extended_state}, the extended precondition
  for $E_1$ is a set of vector-states of the form 
  $(x, y, x_1, y_1)$, namely, 
  $\set{\lrangle{(1, 3, 1, 3)}, \lrangle{(3, 1, 3, 1)}}$;
  Similarly, the extended precondition for $E_2$ is a set of vector-states of the form
  $(x, y, x_2, y_2)$, namely, 
  $\set{\lrangle{(1, 3, 1, 3)}, \lrangle{(3, 1, 3, 1)}}$.

  Because $E_1 ::= x$, 
  the predicate $\aQ_1$---the result of executing $E_1$---has vector-states of the form
  $(x, y, x_1, y_1, e_t)$, 
  and equals 
  $\set{\lrangle{(1, 3, 1, 3, 1)}, \lrangle{(3, 1, 3, 1, 3)}}$.\footnote{
    As in Example~\ref{ex:why_not_plus}, $e_{t_1}$ is omitted in input vector-states
    because its value is irrelevant in the example.
  }
  Similarly, because $E_2 ::= y$, the predicate $\aQ_2$---the 
  result of executing $E_2$---has vector-states of the form
  $(x, y, x_2, y_2, e_{t})$ and equals $\set{\lrangle{(1, 3, 1, 3, 3)}, \lrangle{(3, 1, 3, 1, 1)}}$.

  The predicate $\aQ_1[\vec{v} / \vec{v_1'}]$ creates a set of vector-states of the form
  $(x_1', y_1', x_1, y_1, e_{t_1}')$, i.e., 
  $\set{\lrangle{(1, 3, 1, 3, 1)}, \lrangle{(3, 1, 3, 1, 3)}}$
  (\emph{original} program variables $\vec{v}$ are shifted to $\vec{v_1'}$).
  Similarly, the predicate
  $\aQ_2[\vec{v_2'} / \vec{v}]$ creates a set of vector-states of the form
  $(x_2', y_2', x_2, y_2, e_{t_2}')$, i.e., 
  $\set{\lrangle{(1, 3, 1, 3, 3)}, \lrangle{(3, 1, 3, 1, 1)}}$.

  These are conjoined with $\aP \land (\vec{v} = \vec{v_1} = \vec{v_2})$,
  leaving $\vec{e_t}$ unconstrained (as $\vec{e_t}$ was already unconstrained in $\aP$),
  and produces a set of vector-states of the form
  $(\colormath{dgreen}{x, y,} \colormath{cyan}{x_1', y_1', x_1, y_1,} 
    \colormath{dred}{e_{t_1}',} \colormath{cyan}{x_2', y_2', x_2, y_2,} \colormath{dred}{e_{t_2}'})$,
  i.e.,
    $\set{\lrangle{(\colormath{dgreen}{1, 3,} \colormath{cyan}{1, 3, 1, 3,} 
          \colormath{dred}{1,} \colormath{cyan}{1, 3, 1, 3,} \colormath{dred}{3})}, 
          \lrangle{(\colormath{dgreen}{3, 1,} \colormath{cyan}{3, 1, 3, 1,} 
          \colormath{dred}{3,} \colormath{cyan}{3, 1, 3, 1,} \colormath{dred}{1})}
    }$.
  Original program variables are colored {\color{dgreen} green}, 
  introduced ghost variables are colored {\color{cyan} cyan}, 
  and the results of executing $E_1$ and $E_2$ ($e_{t_1}'$ and $e_{t_2}$', respectively)
  are colored {\color{dred} red}.

  Once the set of states is conjoined with $\vec{e_t} = \vec{e_{t_1}'} + \vec{e_{t_2}'}$,
  and the auxiliary variables are quantified out, 
  we obtain exactly the set of states that $E_1 + E_2$ can produce: $(x, y, e_t) = 
    \set{\lrangle{(1, 3, 4)}, 
         \lrangle{(3, 1, 4)}}$.
\end{example}

At this point, it becomes possible to answer the question: why do we not simply produce 
sets of values as the postcondition for expression nonterminals, and instead introduce states with an 
auxiliary variable $\vec{e_t}$?
The reason is that, as illustrated in Example~\ref{ex:why_not_plus},  
a scheme that merely returns a set of values is unable to track the originating state, 
as done in the $\plusrule$ rule 
(and other instantiations of the $\mathsf{Bin}$ rule of Figure~\ref{fig:exp_rules}), 
and hence imprecise.
Extending the semantics of expressions to return the
extended state allows us to keep track of this information by extending the input and 
output states with ghost variables as illustrated in 
Example~\ref{ex:correcting_plus}, and leads to a precise rule---which is why we opt to 
update the (vector-)variable $\vec{e_t}$ in our rules.

{\small
  \begin{figure}[bt!]
  $$
  \begin{array}{c}
  \\[3mm]
    \infer[\assignrule]{\Gamma \vdash \utriple{\aP}{\Eassign{x}{E}}
    {\exists \vec{x'}. \aQ [\vec{x'} / \vec{x}] \wedge \vec{x} = \vec{e_t}}}
  {
    \Gamma \vdash \utriple{\aP}{E}{\aQ}
  }\quad
    \infer[\seqrule]{\Gamma \vdash \utriple{\aP}{\Eseq{S_1}{S_2}}{\aR}}
  {
    \Gamma \vdash \utriple{\aP}{S_1}{\aQ} \quad 
    \Gamma \vdash \utriple{\aQ}{S_2}{\aR}
  }
  \\[3mm]
    \infer[\siterule]{\Gamma \vdash \bigutriple{\aP}{\Eifthenelse{B}{S_1}{S_2}}{
      \exists \vec{v_1}, \vec{v_2}, \vec{v_1'}, \vec{v_2'}.
      \begin{array}{l}
        \aQ_1[\vec{v_1'}[i] / \vec{v}[i] \ \text{where} \ \vec{b_{t_1}}[i] = \Efalse] \wedge \\
        \aQ_2[\vec{v_2'}[i] / \vec{v}[i] \ \text{where} \ \vec{b_{t_2}}[i] = \Etrue] 
      \end{array}
      \wedge (\vec{v_1} = \vec{v_2})
    }}
  {
    \begin{array}{l}
    \Gamma \vdash \utriple{\aP}{B}{\aP_B} \\
    \Gamma \vdash \utriple{
      \aP_B \wedge (\vec{v_1} = \vec{v})
    }{S_1}{\aQ_1} \\
    \Gamma \vdash \utriple{
      \aP_B \wedge (\vec{v_2} = \vec{v})
    }{S_2}{\aQ_2}
    \end{array} \quad 
    \begin{array}{l}
      \vec{v_1}, \vec{v_2}, \vec{v_1'}, \vec{v_2'} \ \text{fresh renames of} \ \vec{v} \\
      \vec{v} \ \text{refers to full set of vars in} \ \aP
    \end{array}
  }\\[5mm]

    \hspace{-10mm}\infer[\mathsf{While}] {\Gamma \vdash \utriple{\aI}{\Ewhile{B}{S}}{\aI_B \wedge \vec{b_t} = \Ef}}
    {
      \begin{array}{l}
      \vec{b_{loop}}, \vec{v_1}, \vec{v_2} \ \text{fresh} \\
      \Gamma \vdash \utriple{\aI}{B}{\aI_B} \\
      \Gamma \vdash \utriple{\aI_B \wedge \vec{b_{loop}} = 
        \vec{b_t} \wedge \vec{v} = \vec{v_1}}{S}{\aI_B'}
      \end{array} \quad
      \begin{array}{l}
        (\exists \vec{v_1}, \vec{v_2}, \vec{v_1'}, \vec{v_2'}, \vec{b_t}. \aI_B'[\vec{v_{1}'}[i] / \vec{v}[i] \ \text{where} \  \vec{b_{loop}}[i] = \Efalse] \wedge \\
        (\aI_B \wedge \vec{b_{loop}} = \vec{b_t}\wedge \vec{v} = \vec{v_2}) 
        [\vec{v_2'}[i] / \vec{v}[i] \ \text{where} \ \vec{b_{loop}}[i] = \Etrue] \wedge \\
        (\vec{v_1} = \vec{v_2}))
        \implies 
        (\exists \vec{b_t}. \aI_B)
      \end{array}
    }\\[5mm]

    \infer[\mathsf{While\_Exact}]
    {\Gamma \vdash \bigutriple{\aP}{\Ewhile{B}{S}}{
      \begin{array}{l}
        \exists \sigma_{\mathit{start}} \in \aP. b \in B. s \in S. 
          \: \forall j. 
          \: \exists l. \: \exists \sstate_0, \cdots, \sstate_l. \\
          \sigma_{\mathit{start}}[j] = \sstate_0 \wedge \\
          \forall i, 0 \leq i < l.
          \sem{b}(\sstate_i) = \Etrue \wedge \sem{s}(\sstate_i) = \sstate_{i + 1} \wedge \\
          \sigma[j] = \sstate_l \wedge \sem{b}(\sstate_l) = \Efalse
      \end{array}
      }
    }
    {}
  \end{array}
  $$
  \caption{Inference rules for statements in unrealizability logic.
  The postcondition of $\mathsf{While\_Exact}$ is not in the standard assertion language 
  of FO-PA extended with arrays, but there exists an encoding to FO-PA based on the Gödel 
  $\beta$-function~\cite{godelbeta}.
}
  \label{fig:stmt_rules}
\end{figure}
}

\subsubsection{Rules for Statements}
Unrealizability logic considers four kinds of statements: 
assignment, sequential composition, branches, and loops,
as shown in Figure~\ref{fig:stmt_rules}.

$\assignrule$ and $\seqrule$ can be explained 
using the same principles that we used to explain rules for expressions.
In $\assignrule$, the result of the nonterminal $E$ is stored in the variable $\vec{e_t}$ of 
$\aQ$; this value gets referenced and assigned to $\vec{x}$.
$\seqrule$ is the same as in Hoare logic, where the sequentiality of the statements 
are captured by sharing the pre/postcondition $\aQ$.

%$\siterule$ and $\whilerule$ are more complex, so we explain them in detail below.

\mypar{If-Then-Else}
Branches pose a similar challenge to operators such as $\plusrule$: 
one requires a mechanism for composing the states generated by two different 
nonterminals while ensuring that they come from the same `origin' state.
Similarly to what we did for $\plusrule$, the solution is to extend the input states 
with a ghost state: observe that the premise triples 
are of the same form 
$\aP_B \wedge (\vec{v_1} = \vec{v})$, 
where $\aP_B$ is identical to $\aP$ except that $\vec{b_t}$ now stores the 
result of the branch condition.

The additional complication in the $\siterule$ rule is that in a given vector-state, 
certain examples must pass through the $\Etrue$ branch, while others 
must pass through the $\Efalse$ branch.
In the $\siterule$ rule, we achieve this synchronization by first passing 
an input state through \emph{both} the true and false branches (separately)---this
step is captured by the two premises about $S_1$ and $S_2$, which simply push the 
ghost-state extended preconditions through the respective branches
to obtain $\aQ_1$ and $\aQ_2$.

When \emph{merging} the states from $\aQ_1$ and $\aQ_2$, 
we project out the examples that took the wrong branch---e.g., 
in $\aQ_1$, the postcondition of the $\Etrue$ branch, 
the examples that should go through the $\Efalse$ branch are projected out.

This projecting out is captured by the conditional substitution 
$[\vec{v_1'}[i] / \vec{v}[i] \ \text{where} \ \vec{b_{t_1}}[i] = \Efalse]$, 
which substitutes a variable $v[i]$ with a fresh variable 
${v_1[i]'}$ \emph{only if} 
$\vec{b_{t_1}}[i] = \Efalse$---that is, if the branch condition 
evaluated to $\Efalse$ for that example.
(We reference $\vec{b_{t_1}}$, the ghost-version of $\vec{b_t}$, 
because  the value of $\vec{b_t}$ may change through the execution of $S_1$.)

The quantifier-free segment of the postcondition 
$\aQ_1[\vec{v_i'}[i] / \vec{v}[i] \ \text{where} \ \vec{b_{t_1}}[i] = \Efalse] \wedge 
 \aQ_2[\vec{v_i'}[i] / \vec{v}[i] \ \text{where} \ \vec{b_{t_2}}[i] = \Etrue] \wedge 
 (\vec{v_1} = \vec{v_2})
$
is thus an extended state with five `copies' of state, of which 
only one is retained:
the examples that took the correct 
branches (named via $\vec{v}$, and the ones that will be retained), 
the examples that took the incorrect branches (named via 
$\vec{v_1'}$ and $\vec{v_2'}$), and the two ghost states (named via $\vec{v_1}$ and $\vec{v_2}$).
Quantifying out the auxiliary variables 
$\vec{v_1'}, \vec{v_2'}, \vec{v_1},$ and $\vec{v_2}$ leaves just the examples that took the 
correct branches.

\begin{example}[If-then-Else with Extended States]
  \label{ex:eite_example}
  Consider the set of states $\aP$ given by 
  $(x_1, x_2) = \set{\lrangle{(-1, 1)}, \lrangle{(2, -2)}}$; that is, 
  $\aP$ denotes a set of two-example configurations, 
  which could be specified by the formula $ (x_1 = -1 \land x_2 = 1) \lor (x_1 = 2 \land x_2 = -2)$.
  Consider the If-Then-Else statement
  $\Eifthenelse{x > 0}{\Eassign{x}{x}}{\Eassign{x}{0-x}}$ (we use a single program as 
  it is sufficient to illustrate the $\siterule$ rule);
  this program sets $x$ to its absolute value.
  Observe that after executing $x > 0$ (corresponding to $B$), 
  $\aP_B$ becomes the set of states 
  $(x_1, x_2, b_{t_1}, b_{t_2}) = 
    \set{\lrangle{(-1, 1, \Efalse, \Etrue)}, \lrangle{(2, -2, \Etrue, \Efalse)}}$.
  After conjoining the ghost state and running through $S_1$, we get the following set
  (the ghost state is marked in {\color{cyan}cyan}):
  {\small
  $$
    (x_1, x_2, b_{t_1}, b_{t_2}, \colormath{cyan}{x_{1_1}, x_{2_1}, b_{t_{1_1}}, b_{t_{2_1}}}) 
    \! = \!
    \set{\lrangle{(-1, 1, \Efalse, \Etrue, \colormath{cyan}{-1, 1, \Efalse, \Etrue})}, 
         \lrangle{(2, -2, \Etrue, \Efalse, \colormath{cyan}{2, -2, \Etrue, \Efalse})}
    }
  $$
  }
  Applying the conditional substitution to this set yields the following set of states, where 
  non-substituted variables are {\color{dgreen} green}, variables 
  substituted with $\vec{v_1'}$ are {\color{dred} red}, and the ghost state is {\color{cyan} cyan}:
  $$
    \set{\lrangle{(\colormath{dred}{-1}, \colormath{dgreen}{1}, \colormath{dred}{\Efalse}, \colormath{dgreen}{\Etrue}, \colormath{cyan}{-1, 1, \Efalse, \Etrue})}, 
         \lrangle{(\colormath{dgreen}{2}, \colormath{dred}{-2}, \colormath{dgreen}{\Etrue}, \colormath{dred}{\Efalse}, \colormath{cyan}{2, -2, \Etrue, \Efalse})}}
  $$
  Similarly, executing $S_2$ and applying substitution yields the following set:
   $$
    \set{\lrangle{(\colormath{dgreen}{1}, \colormath{dred}{-1}, \colormath{dgreen}{\Efalse}, \colormath{dred}{\Etrue}, \colormath{cyan}{-1, 1, \Efalse, \Etrue})}, 
         \lrangle{(\colormath{dred}{-2}, \colormath{dgreen}{2}, \colormath{dred}{\Etrue}, \colormath{dgreen}{\Efalse}, \colormath{cyan}{2, -2, \Etrue, \Efalse})}}
   $$ 
  Again, note that the substitution happens for variables highlighted in {\color{dred} red}, while the 
  {\color{dgreen} green} variables are the original program variables.
  Taking the conjunction and quantifying out all unnecessary variables leaves only the variables 
  highlighted in {\color{dgreen} green}, which yields:
  $$
    (\colormath{dgreen}{x_1}, \colormath{dgreen}{x_2}, 
    \colormath{dgreen}{b_{t_1}}, \colormath{dgreen}{b_{t_2}}) = \set{
      \lrangle{(\colormath{dgreen}{1}, \colormath{dgreen}{1}, \colormath{dgreen}{\Efalse}, \colormath{dgreen}{\Etrue}}, 
      \lrangle{(\colormath{dgreen}{2}, \colormath{dgreen}{2}, \colormath{dgreen}{\Etrue}, \colormath{dgreen}{\Efalse})}
    },
  $$
  Which is exactly the set of states that are computed by the given If-Then-Else statement.
\end{example}

  In general, conditionally substituting a variable $\vec{x}[i]$ with 
  $\vec{x}'[i]$, 
  if a condition $\vec{b}[i]$ is false may be performed by replacing all occurences 
  of $\vec{v}[i]$ in a formula with 
  $\Eifthenelse{\vec{b}[i]}{\vec{x}[i]}{\vec{x'}[i]}$: this expression is equivalent to
  $\vec{x'}[i]$ when the branch condition $\vec{b}[i]$ is false, and equivalent to 
  $\vec{x}[i]$ when the expression is true.
  Note that this substitution is possible even if 
  the predicate is used to encode 
  infinite vector-states, as long as the predicate itself is finite.

  \begin{example}[Conditional Substitution Over Infinite Vector-States]
    Consider the infinite vector-state $\forall i. \vec{x}[i] = -1 \wedge  \vec{y}[i] = i$.
    Although this vector-state is infinite, one can conditionally substitute 
    $\vec{x}$ with a new vector-variable $\vec{x'}$ for entries where
    $\vec{b}[i]$ is $\Efalse$ with the expression 
    $\forall i. (\Eifthenelse{\vec{b}[i]}{\vec{x}[i]}{\vec{x'}[i]}) = -1 \wedge  \vec{y}[i] = i$.
    Observe that $\vec{x}[i]$ is now unconstrained for entries 
    where $\vec{b}[i] = \Efalse$, 
    whereas
    $\vec{x'}[i]$ is unconstrained instead for entries where $\vec{b}[i] = \Etrue$.
  \end{example}

\mypar{Loops}
The basic intuition behind the $\whilerule$ rule is identical 
to that in Hoare logic---it requires the existence of an \emph{invariant} $\aI$.
The difference with Hoare logic is that, in our work, $\aI$ must work
as an invariant for \emph{all} possible completions of the loop body
$S$.

%In unrealizability logic, a vector-state must repeat the loop body until 
%\emph{all} examples cause the loop condition to 
%evaluate to $\Efalse$, which is why the postcondition of
%the $\whilerule$ rule
%has the condition $\vec{b_t} = \Ef$:
%all the loop conditions of all examples (stored in $\vec{b_t}$ of $\aI_B$) evaluate to $\Efalse$.

Because the invariant $\aI$ spans multiple examples (i.e., vector-states), 
the way the invariant is checked is also 
different from in Hoare logic.
First, we take a similar approach as in $\siterule$, 
and extend the invariant $\aI_B$ with the ghost state $\vec{v_1}$.
This extended invariant is pushed through the loop body $S$regardless of the 
loop guard (in the premise 
$\utriple{\aI_B \wedge \vec{b_{loop} = \vec{b_t}} \wedge \vec{v} = \vec{v_1}}{S}{\aI_B'}$).

Observe that $\aI_B'$ contains examples that were pushed through the loop body even 
when they do not satisfy the loop guard.
The implication on the right of the rule is thus designed to check whether 
$\aI_B' \implies \aI_B$ 
\emph{only on the examples where the loop guard is} $\Etrue$.
To achieve this, a pattern similar to the $\plusrule$ rule is applied: 
the premise of the implication composes $\aI_B'$ and $\aI$ to obtain 
a mixed vector where examples that pass the loop guard are drawn from $\aI_B'$ and 
examples that do not pass the guard are drawn from $\aI$.
This vector is checked against $\exists \vec{b_t}. \aI_B$ 
to establish that $\aI$ is indeed an invariant.

%this is again achieved through conditional substitution,
%which substitutes variables for which $\vec{b_{loop}[i] = \Efalse}$, and quantifies
%out the substituted variables again (as well as the auxiliary variables $\vec{e_t}$
%and $\vec{b_t}$, which may change).\footnote{
%  The additional assignment $\vec{b_{loop}} = \vec{b_t}$ before executing the loop body, 
%  as well as referencing $\vec{b_{loop}}$ in the postcondition, is again due to the fact that 
%  $\vec{b_t}$ may change during the execution of $S$.
%}
%For examples that are substituted out, $\aI_B$ is an invariant because 
%these examples are left unchanged---thus, one only needs to check the invariant condition for 
%examples that pass the loop guard and execute the body, 
%which the implication premise captures.
%
%Finally, the condition $\vec{b_t} = \Ef$ in the postcondition captures vector-states 
%where the loop guard of all examples (stored in $\vec{b_t}$ of $\aI_B$) evaluate to $\Efalse$:
%i.e., vector-states for which the loop terminates 
%on all individual examples within the vector-state.
%In other words, this postcondition checks for \emph{partial correctness}, 
%and does not provide any guarantees about traces that contain an example that 
%fail to terminate.

Unfortunately, $\whilerule$ is incomplete despite the fact that the original while 
rule for Hoare logic is complete.
Unrealizability logic thus includes an additional rule $\mathsf{While\_Exact}$ to 
ensure (relative) completeness of the logic; 
the postcondition of $\mathsf{While\_Exact}$ \emph{precisely} captures the set of 
vector-states that may occur from executing a loop from $\Ewhile{B}{S}$ on a 
input vector-state from $\aP$.
$\mathsf{While\_Exact}$ does deviate from the standard invariant-based reasoning 
scheme of Hoare logic, but we observe that 
\rone invariant-based reasoning is still possible in Unrealizability logic through the 
standard $\whilerule$ rule, and 
\rtwo even in standard Hoare logic, to achieve complete `invariant-based' reasoning 
one must come up with invariants that are similar to the postcondition of 
$\mathsf{While\_Exact}$ anyways.
This is because the proof of completeness in standard Hoare logic depends on the 
existence of an invariant that precisely captures the sequence of states that may be 
generated by the execution of a loop, similar to $\mathsf{While\_Exact}$.

%The reader may find it surprising that a single invariant capturing the behavior of all 
%possible loop bodies suffices for relative completeness.
%A single invariant does suffice for relative completeness because one
%can actually write an invariant that talks about the choice of loop body \emph{itself}.
%However, to pull off this trick, one needs a highly complex encoding
%whose details are beyond the scope of this paper (see \S\ref{SubSe:Completeness} for details).
%In practice, the use of this complex encoding will often result in an invariant too complex 
%for a proof to be completed.
%\jinwoo{This is due to the fact that, through a complex encoding, one can actually talk about the 
%choice of loop body \emph{within the invariant}.}
%It is true that such a single invariant does suffice for completeness, 
%in practice, the required invariant 
%may become extremely complex depending on the loop body.
When writing actual proofs in unrealizability logic, it is often beneficial to 
split the loop bodies into smaller sets, and reason about them separately (as we have done in 
Example~\ref{ex:loops_bad} of \S\ref{Se:MotivatingExample}) through the 
use of the structural rules.

{\small
\begin{figure}[bt!]
  $$
    \begin{array}{c}
      \\[3mm]
      \infer[\weaken]{\Gamma \vdash \utriple{\aP'}{N_1}{\aQ'}}
      {
        \Gamma \vdash \utriple{\aP}{N}{\aQ} \quad 
        \aP' \implies \aP \quad 
        \aQ \implies \aQ' \quad 
        N_1 \subseteq N
      }\quad
      \infer[\conj]{\Gamma \vdash \utriple{\aP}{N}{\aQ_1 \wedge \aQ_2}}
      {
        \Gamma \vdash \utriple{\aP}{N}{\aQ_1} \quad 
        \Gamma \vdash \utriple{\aP}{N}{\aQ_2}
      }\\[3mm]
      \infer[\grmdisj]{\Gamma \vdash \utriple{\aP}{(N_1 \cup N_2)}{\aQ}}
      {
        \Gamma \vdash \utriple{\aP}{N_1}{\aQ} \quad 
        \Gamma \vdash \utriple{\aP}{N_2}{\aQ}
      }\quad
      \infer[\inv]{\Gamma \vdash \utriple{\aP}{N}{\aP}}
      {vars(\aP) \: {\cap} \: vars(N) = \emptyset}\quad
      \\[3mm]
      \infer[\subone]{\Gamma \vdash \utriple{\aP [\vec{y} / \vec{z}]}{N}{\aQ [\vec{y} / \vec{z}]}}
      {
        \Gamma \vdash \utriple{\aP}{N}{\aQ} & 
        \vec{y} \: {\cap} \: vars(N) = \emptyset & 
        \vec{z} \: {\cap} \: vars(N) = \emptyset & \
        (\vec{y}, \vec{z}) \text{ are not } \vec{e_t} \text{ or } \vec{b_t}
      }
      \\[3mm]
      \infer[\subtwo]{\Gamma \vdash \utriple{\aP [\vec{y} / \vec{z}]}{N}{\aQ}}
      {\Gamma \vdash \utriple{\aP}{N}{\aQ} & \vec{z} \: {\cap} \: (vars(N) \: {\cup} \: vars(\aQ)) = \emptyset }\\[3mm]
      \infer[\hp]{\Gamma \vdash \utriple{\aP}{N}{\aQ}}
    {
      \begin{array}{@{\hspace{0ex}}l}
        \Gamma, \utriplenospace{\aP}{N}{\aQ} \vdash \utriplenospace{\aP}{\RHS_1}{\aQ} \\
        \cdots \\
        \Gamma, \utriplenospace{\aP}{N}{\aQ} \vdash \utriplenospace{\aP}{\RHS_n}{\aQ}
      \end{array}
      ~~
      \begin{array}{@{\hspace{0ex}}l@{\hspace{-0.4ex}}}
        \\
        N \rightarrow \RHS_1  \mid  \cdots \mid \RHS_n\\
                \text{is exhaustive}
      \end{array}
    }
		\quad
      \infer[\applyhp]{\Gamma, \utriple{\aP}{N}{\aQ} \vdash \utriple{\aP}{N}{\aQ}}
      {}
    \end{array}
  $$
  \caption{Structural rules in unrealizability logic.
           $\vars(N)$ refers to the entire set of variables that may appear in 
           a program $t_N \in N$, with the addition of the reserved auxiliary variables 
           $\vec{e_t}$ and $\vec{b_t}$.
             }
  \label{fig:struct_rules}
\end{figure}
}

\subsubsection{Structural Rules}
%NOTE: Fill
Finally, unrealizability logic has a set of structural 
rules that operate on conditions, hypotheses, and sets of 
programs.
We list them in Figure~\ref{fig:struct_rules}.

\mypar{Weaken, Conj} 
These rules operate like their counterpart in Hoare logic. 
$\weaken$ can
 shrink the precondition or enlarge the postcondition following the principle that 
the postcondition overapproximates the set of all states that the precondition can generate.
$\weaken$ can also shrink the set of  considered programs.
$\conj$ takes the conjunction of the two postconditions.

\mypar{GrmDisj} $\grmdisj$ allows us to split sets of programs: if two sets of programs, $N_1$ and $N_2$, satisfy the same pre- and postcondition,
then their union also satisfies those same pre- and postcondition.
This rule is not required for the completeness of the logic. 
However, it helps greatly
in simplifying actual unrealizability proofs: a clever division of the search space often results in simpler predicates.

\mypar{Inv, SubOne, SubTwo} 
$\inv$, $\subone$, and $\subtwo$ are 
 rules that are necessary for our completeness proof.
$\inv$ states that if a precondition does not share any free variables with any program 
in the set $N$ 
(i.e., $vars(\aP) \cap vars(N) = \emptyset$), then it is an invariant---clearly this holds, because 
any program $N$ will be unable to access or modify the variables used in $\aP$.
$\subone$ performs $\alpha$-renaming of variables inside the 
pre-and postconditions; the side conditions state that one can only rename variables that are 
not `captured' by the set of programs $N$.
$\subtwo$ states that variables from the postcondition that do not appear in 
the set of programs $N$ or the postcondition may be renamed to any new name.
Intuitively, this is because a variable that only occurs in the precondition will remain 
untouched during program execution; see the proof of $\subtwo$ in the full version of the 
paper for details.

\mypar{HP, ApplyHP} 
$\hp$ and $\applyhp$ allow us to 
perform structural induction in the proof tree.
%All other rules described above do not interact with the hypothesis context $\Gamma$, and simply carry it around.
$\hp$ introduces a new hypothesis into the context $\Gamma$, while $\applyhp$ allows us to 
apply an induction hypothesis. (These rules were explained in~\S\ref{SubSe:InfiniteSetsOfPrograms}).

%When an unrealizability triple is derivable using the rules in Figures~\ref{fig:exp_rules}, \ref{fig:stmt_rules}, and 
%\ref{fig:struct_rules}, we say that the said triple is \emph{derivable}.
%NOTE: Add stuff about the context
\begin{definition}[Derivability]
  \label{def:derivable}
  Given a precondition over vector-states $\aP$, a postcondition over vector-states $\aQ$, a set of programs $S$, 
  and a set of hypotheses $\Gamma$, 
  we say that a triple $\utriple{\aP}{S}{\aQ}$ is \emph{derivable assuming $\Gamma$}, 
  denoted by $\Gamma \vdash \utriple{\aP}{S}{\aQ}$, 
  if there exists a proof tree deriving $\Gamma \vdash \utriple{\aP}{S}{\aQ}$ using the rules of 
  unrealizability logic in Figures~\ref{fig:exp_rules}, \ref{fig:stmt_rules}, and 
  \ref{fig:struct_rules}.

  We say that a triple $\utriple{\aP}{S}{\aQ}$ is \emph{derivable}, if it can be 
  derived with an empty set of hypotheses.
\end{definition}

%%% Local Variables:
%%% mode: latex
%%% TeX-master: "paper.tex"
%%% End:

%% file: 4soundnesscompleteness.tex
\section{Soundness, Completeness, and Other Theoretical Results}
\label{Se:SoundnessAndCompleteness}

Unrealizability logic is a sound logic in the sense that 
any derivable triple $\utriple{\aP}{N}{\aQ}$ is also valid.

\begin{theorem}[Soundness]
  \label{thm:soundness}
  Given a nonterminal $N$ in a grammar $G$, the following property holds:
\[\vdash \utriple{\aP}{N}{\aQ} \implies \models \utriple{\aP}{N}{\aQ}\]
\end{theorem}

Soundness can be proved via structural induction, which shows that the conclusion
triple of each rule is sound given that the premises are sound
(see the Appendix in the full version of this paper~\cite{ularxiv} for the full proof).
While reasoning about some cases is complex due to the vectorized semantics, 
the overall structure of the proof is a simple structural-induction argument.

Surprisingly, unrealizability logic is also relatively complete, 
in the sense that 
all valid triples $\utriple{\aP}{N}{\aQ}$ are derivable, assuming 
that there is an oracle capable of verifying all true sentences in the assertion 
language.
\begin{theorem}[Relative Completeness]
	\label{thm:completeness}
  Let $\aP$ and $\aQ$ be predicates in an assertion language $L$.
  Assuming the existence of an oracle capable of verifying all true sentences in 
  $L$, the following property holds for each nonterminal $N$ in a grammar $G$:
\[\models \utriple{\aP}{N}{\aQ} \implies \vdash \utriple{\aP}{N}{\aQ}\]
\end{theorem}

Completeness follows from two lemmas that we state later 
(Lemmata~\ref{lem:strongest_triple} and~\ref{lem:general_triple}), which state that 
\rone a \emph{strongest triple} that precisely describes the behavior of a nonterminal is derivable, 
and \rtwo that any sound triple about a nonterminal can be derived from the strongest 
triple.

%Soundness follows from the soundness of each rule; we refer the reader to Appendix~\ref{App:Proofs}
%for the full proof.

%\begin{proof}
%  Soundness can be proved by showing the soundness of each rule; 
%  i.e., assuming that the premise triple is valid, then the conclusion triple is 
%  also valid.
%  This holds for all of our non-structural rules as the postconditions are precise 
%  (except $\while$, which is sound but overapproximating).
%  For structural rules, all other rules except $\hp$ and $\applyhp$ 
%  may be proved valid by showing that the postcondition of the conclusion is 
%  overapproximating according to Definition~\ref{def:triple}.
%  $\hp$ and $\applyhp$ can be proved sound using the soundness of structural induction.
%\end{proof}

In the rest of this section, we give a sketch of the proof of completeness, 
and discuss some other theoretical results.
We will present proof sketches for the theorems listed in this section, for the 
full proofs, we again refer the reader to the full version of the paper~\cite{ularxiv}.

\subsection{A Sketch of the Proof of Completeness}
\label{SubSe:Completeness}

Completeness of unrealizability logic is a two-step process: first we wish to 
prove that the non-structural rules (expression and statement rules) of 
unrealizability logic are complete, in the sense that,
given a set of programs generated by a production $op(N_1, \cdots, N_k)$, 
if all sound premise triples about $N_1, \cdots, N_k$ are derivable, 
then all sound triples about $op(N_1, \cdots, N_k)$ are also derivable.

\begin{lemma}[Completeness of Expression and Statement Rules]
	\label{thm:precision}
  Let $G$ be a regular tree grammar and $N$ be a set of programs
  generated by either a nonterminal or the RHS of a production.
  Then, the expression and statement rules of unrealizability logic
  \emph{preserve completeness}, in the sense that,
  if all sound required premise triples are derivable,
  then all sound triples about $N$ are derivable in unrealizability logic.
\end{lemma}

When creating a proof tree in unrealizability logic, 
one will encounter triples where the center element is 
the right-hand side of a production (e.g., $\circled{1}$ of Figure~\ref{fig:s2_prooftree}).
Consequently, Lemma 4.3 is stated for both nonterminals and right-hand sides of productions.
Lemma~\ref{thm:precision} can be proved by showing that 
the postcondition of the conclusion of each rule 
\emph{precisely} captures the semantics of the corresponding operator.

The most interesting case in the proof of Lemma~\ref{thm:precision} 
is the case with loops; all other cases can be
proved via structural induction. 
For loops, because unrealizability logic contains the rule
$\mathsf{While\_Exact}$, 
which precisely captures the semantics of loops, any sound triple about a loop can be
proved via $\mathsf{While\_Exact}$ and an application of the $\mathsf{Weaken}$ rule.

While Lemma~\ref{thm:precision} would be sufficient as a proof of 
completeness in normal Hoare logic, it is insufficient for 
unrealizability logic because we are working over an RTG;
without the use of the $\hp$ rule, the proof tree will become 
infinite.
The next step in showing completeness remedies this fact
by proving that through use of the $\hp$ rule, we can prove the 
strongest triple for a nonterminal within a finite number of steps.

\begin{lemma}[Derivability of the Strongest Triple]
  \label{lem:strongest_triple}
  Let $N$ be a nonterminal from a RTG $G$ and $\vec{z}$ be a set of auxiliary symbolic variables. 
  Let $\aQ_0\equiv\semantics{N}(\vec{x} = \vec{z})$; that is, 
  $\aQ_0$ is a formula that precisely captures the behavior of the set of programs $L(N)$ on the symbolic vector-state 
  $x_1 = z_1 \wedge  \cdots x_n = z_n$.\footnote{
    We assume that state is defined over a single variable $x$; it is trivial to extend this to multiple-variable states.
  }
  We refer to $\utriple{\vec{x} = \vec{z}}{N}{\aQ_0}$ as the \emph{strongest triple} for $N$.
  Then $\vdash \utriple{\vec{x} = \vec{z}}{N}{\aQ_0}$ is derivable.
\end{lemma}

The proof of Lemma~\ref{lem:strongest_triple} performs induction 
over the number of hypotheses in the context: 
the proof relies on the fact that
one only needs to insert the strongest triple for each nonterminal into the context 
(because the strongest triple can derive any other triple), 
thus the $\hp$ rule need only be applied at most $|G|$ times 
(the number of nonterminals inside $G$).
When all $|G|$ hypotheses are established, one relies on Lemma~\ref{thm:precision} to show that 
the strongest hypothesis can actually be derived.

The proof of Lemma~\ref{lem:strongest_triple}, as well as the final step in arguing completeness, 
requires a lemma that shows one can derive any triple from the strongest triple.

\begin{lemma}[Derivability of General Triples]
  \label{lem:general_triple}
  Let $G$ be a grammar and $N$ be any nonterminal in $G$. 
  Given the strongest triple $H_N$ for the nonterminal $N$ 
  (as defined in Lemma~\ref{lem:strongest_triple}),
  if $\Gamma \vdash H_N$,
  then $\Gamma \models \utriple{\aP}{N}{\aQ} \implies \Gamma \vdash \utriple{\aP}{N}{\aQ}$.
\end{lemma}

Lemma~\ref{lem:general_triple} can be proved using the substitution rules $\subone$ and $\subtwo$.
Relative completeness (Theorem~\ref{thm:completeness}) then follows from 
Lemma~\ref{lem:strongest_triple} and Lemma~\ref{lem:general_triple}.

\subsection{Undecidability of Unrealizability and Decidable Fragments}
\label{SubSe:Decidability}

Although we have proved relative completeness, proving unrealizability is 
an undecidable problem, even when limited to synthesis problems 
\emph{without} loops.

\begin{theorem}[Undecidability of Unrealizability]
  \label{thm:undecidability}
  Let $\syu$ be a synthesis problem with a grammar $\gu$ that does \emph{not} 
  contain productions of the form $S::=\Ewhile{B}{S_1}$.
	Checking whether $\syu$ is unrealizable is an undecidable problem.
\end{theorem}

The proof of Theorem~\ref{thm:undecidability} relies on translating
an arbitrary program for a two-counter machine into a program-synthesis problem 
that is realizable iff the original program halts.
This translation is performed by taking each instruction of the original 
program and translating it into a production that performs the same 
computation.
The goal of the resultant synthesis problem is to set a special variable $h$ to $1$:
by translating only the final $\mathsf{Halt}$ instruction into a production 
that sets $h$ to $1$, 
the translated synthesis problem is realizable iff the original program halts.
As the halting problem for two-counter machines is undecidable, it follows that 
unrealizability is undecidable as well.

%program and constructs a synthesis problem from it; 
%it then shows that solving the halting problem for the two-counter machine can be 
%reduced to showing unrealizability of the synthesis problem.
%However, as the halting problem for two-counter machines is, in general, 
%undecidable, it follows that proving unrealizability for synthesis problems is 
%also undecidable.

Theorem~\ref{thm:undecidability} 
shows that proving unrealizability of synthesis problems is undecidable, 
even when while loops are out of the question.
%(which would, in terms of simple verification, make things much simpler).
%Follwing the proof of Theorem~\ref{thm:undecidability}, this is, in essence, 
%because the recursive structure of a RTG can mimic the `looping' structure 
%of a program---thus even without loops, there is an element of recursion
%that leads to undecidability.
On the other hand, when one is limited to synthesis over finite domains, 
proving unrealizability becomes decidable---even when the synthesis problem 
may contain an unbound number of loops.

\begin{theorem}[Decidability of Unrealizability over Finite Domains]
  \label{thm:finite_decidability}
  Determining whether a synthesis problem $\syfin$ is unrealizable, where the grammar $\gfin$
  is valid with respect to $G_{impv}$, is \emph{decidable} if the semantics of
	programs in $\gfin$ is defined over a finite domain.
\end{theorem}

The proof of Theorem~\ref{thm:finite_decidability} relies on the fact that if the 
domain is finite, 
one can perform grammar-flow analysis~\cite{gfa} to obtain the greatest-fixed point 
of values that a nonterminal may generate.
In this paper, because both expressions and statements are of type 
$\mathsf{State} \rightarrow \mathsf{State}$, 
the values that nonterminals generate are pairs of states, representing input-output relations;
this computation 
is guaranteed to terminate 
because the domain is finite.

Furthermore, if $\syfin$ is unrealizable, then there exists a proof of this fact in 
unrealizability logic.

%% file: 5examples.tex
\section{The Power of Unrealizability Logic}
\label{Se:Examples}

In this section, we give some example proofs that 
highlight the capabilities of unrealizability logic.

\subsection{Dealing with an Infinite Number of Examples}
\label{Se:DealingWithAnInfiniteNumberOfExamples}

As stated in \S\ref{Se:MotivatingExample} and \S\ref{Se:UnrealizabilityLogic}, 
one of the major features of unrealizabilty logic is that it is capable 
of dealing with synthesis problems which require an infinite number of examples 
to prove unrealizability.
We illustrate one such example in this section.

%Previous approaches, such as \nope~\cite{nope}, \nay~\cite{nay}, and \messy~\cite{semgus} 
%rely on counterexample-guided inductive synthesis, and 
%are thus limited to cases where unrealizability can be proved with finitely many examples.

\begin{example}[Proof with Infinitely Many Examples]
  \label{ex:inf_example}
  Consider a synthesis problem $\syite$, with the grammar
  $\gite$ given below:
  $$
  \begin{array}{lcl}
    \nstart & \rightarrow & \Eifthenelse{B}{A}{\nstart} \mid A \\
    B       & \rightarrow & y == N \\ 
    N       & \rightarrow & 0 \mid N + 1 \\ 
    A       & \rightarrow & \Eassign{x}{N}
  \end{array}
  $$

  $\syite$ aims to synthesize a function $f$ which takes as input 
  a state $(x, y)$, and return a state where $x$ is equivalent to $y$.

  $\syite$ is unrealizable, and one needs an infinite number of examples to 
  prove this fact.
  This is because a specification over $n$ examples for $\syite$ will be 
  realizable by a term that contains $n$ If-then-Elses.
  However, a term of finite size in $\gite$ can only assign to $x$ a finite 
  number of constants; thus it is impossible to have a function in $\gite$ 
  that sets $x$ to $y$ in the general case.
  $\syite$ is, in fact, an imperative variant of 
  an example often used to show that there are synthesis problems that cannot be 
  proven unrealizable using only a finite number of examples~\cite{nope}; 
  previous
  approaches~\cite{nope, nay, semgus} thus cannot prove $\syite$ unrealizable.
  We show that a proof tree for $\syite$ can be constructed in unrealizability logic, 
    using FO-PA extended with infinite arrays as the assertion language.

  Consider a set of inputs $(x, y)$ given by the predicate 
  $\forall i \in \mathbb{N}. x_i = -1 \wedge y_i = i$.
  That is, the input is an 
  infinite set of examples where $y$ spans over the 
  positive integers, and $x$ is assigned the fixed value $-1$ 
  (this implies $x_i \neq y_i$ for every $i \in \mathbb{N}$).
  The output specification for this input is given as 
  $\forall i. x_i = i \wedge y_i = i$; the infinite vector-state 
  where $x_i = y_i = i$ for the $i$-th example.
  This infinite set of examples does \emph{not} encompass the entire input state; 
  however, it is sufficient to prove unrealizability.

  To prove unrealizability, first negate the postcondition to obtain the 
  triple that we wish to prove:
  $$
    \utriple{\forall i. x_i = -1 \wedge y_i = i}{\nstart}
    {\exists i. x_i \neq i \vee y_i \neq i}
  $$
  Instead of proving this triple directly, we will prove the following 
  triple (from which we can obtain the target triple via $\mathsf{Weaken}$):
  $$
    \bigutriple
    {\forall i. y_i = i \wedge \begin{array}{l} \text{only a finite no. of } i \\ \text{such that } x_i = y_i \end{array}}
    {\nstart}
    {\forall i. y_i = i \wedge \begin{array}{l} \text{only a finite no. of } i \\ \text{such that } x_i = y_i \end{array}}
  $$
  Let $\aI$ denote the condition `$\forall i. y_i = i \wedge \text{only a finite no. of } i \ \text{such that } x_i = y_i$'.
  To see the implication, observe that: \rone $\forall i. x_i = -1 \wedge y_i = i \implies \aI$ as there are $0$ (a finite number) 
  $i$s for which $x_i = y_i$; 
  \rtwo $\aI \implies \exists i. x_i \neq i \vee y_i \neq i$, because if $x_i = y_i$ for only a finite number of $i$, then 
  there must exist some $i$ for which $x_i \neq i$.

  We will prove that $\utriple{\aI}{\nstart}{\aI}$ holds by introducing this triple as a hypothesis for $\nstart$ via 
  the $\hp$ rule.
  Let $\Gamma_{\aI}$ denote the context containing only the triple $\utriple{\aI}{\nstart}{\aI}$.
  $$
    \infer[\hp]{\vdash \utriple{\aI}{\nstart}{\aI}}
    {
      \Gamma_{\aI} \vdash \utriple{\aI}{A}{\aI} \quad
      \Gamma_{\aI} \vdash \utriple{\aI}{\Eifthenelse{B}{A}{\nstart}}{\aI}
    }
  $$
  Consider first the base case $\utriple{\aI}{A}{\aI}$, where $A$ is the simple assignment $\Eassign{x}{N}$.
  The hypothesis in this case can be proved simply via $\assignrule$ and $\weaken$:
  $$
    \infer[\weaken]{\Gamma_{\aI} \vdash \utriple{\aI}{A}{\aI}}{
      \infer[\weaken]{\Gamma_{\aI} \vdash \utriple{\aI}{A}
      {\exists \vec{x'}. \aI[\vec{x'} / \vec{x}] \wedge \exists k. k \geq 0 \wedge \forall i. x_i = k}}{
      \infer[\assignrule]{\Gamma_{\aI} \vdash \utriple{\aI}{A}
      {\exists \vec{x'}. \aI[\vec{x'} / \vec{x}] \wedge \exists k. k \geq 0 \wedge \vec{e_t} = k \wedge \vec{x} = \vec{e_t}}}{
        \infer[\weaken]{\Gamma_{\aI} \vdash \utriple{\aI}{N}{\aI \wedge \exists k. k \geq 0 \wedge \vec{e_t} = k}}{
        \infer[\hp]{\Gamma_{\aI} \vdash \utriple{\aI}{N}{
          \exists \vec{e_t'}. \aI[\vec{e_t'} / \vec{e_t}] \wedge \exists k. k \geq 0 \wedge \vec{e_t} = k }}
        {\cdots\vspace{2pt}}
    }}}}
  $$
  The ellipsis over the top-level $\hp$ rule indicates that we can prove that $N$ results in a value 
  above $0$ via the $\hp$ rule, through 
  the induction hypothesis 
  $\utriple{\aI}{N}{\exists \vec{e_t'}. \aI[\vec{e_t'} / \vec{e_t}] \wedge \exists k. k \geq 0 \wedge \vec{e_t} = k}$
  (the exact reasoning is omitted).
  The final application of $\weaken$ works as, if all $x_i = k$ for some $k \geq 0$, then there is at most 
  one (a finite number) $i$ for which $x_i = y_i$.

  The inductive case requires an application of the $\siterule$ rule.
  We first write:
{\small
  $$
  \infer[\weaken]{\Gamma_{\aI} \vdash \utriple{\aI}{\Eifthenelse{B}{A}{\nstart}}{\aI}}{
    \infer[\siterule]{\Gamma_{\aI} \vdash \bigutriple{\aI}{\Eifthenelse{B}{A}{\nstart}}
    {
      \exists \vec{v_1}, \vec{v_2}, \vec{v_1'}, \vec{v_2'}.
      \hspace{-1mm}
      \begin{array}{l}
        \aI[\vec{v_1'}[i] / \vec{v}[i] \ \text{where} \ \vec{b_{t_1}}[i] = \Efalse] \wedge \\
        \aI[\vec{v_2'}[i] / \vec{v}[i] \ \text{where} \ \vec{b_{t_2}}[i] = \Etrue]
      \end{array}
      \hspace{-2mm}\wedge (\vec{v_1} = \vec{v_2})
    }}
    {
      \Gamma_{\aI} \vdash \utriple{\aI}{B}{\aI} \quad 
      \Gamma_{\aI} \vdash \utriple{\aI \wedge \vec{v_1} = \vec{v}}{A}{\aI} \quad 
      \Gamma_{\aI} \vdash \utriple{\aI \wedge \vec{v_2} = \vec{v}}{\nstart}{\aI} \quad
    }
  }
  $$
}
  Observe that the first premise $\utriple{\aI}{B}{\aI}$ seems to do nothing.
  This is the result of weakening the postcondition generated by $B$, and effectively `forgetting' 
  the value of the branch condition as done in the following proof tree, where the postcondition 
  of the premise represents the exact postcondition one would have obtained through a precise reasoning of $\eqrule$:
  {\small
  $$
    \infer[\weaken]{\Gamma_{\aI} \vdash \utriple{\aI}{B}{\aI}}{
      \infer[\eqrule]{\Gamma_{\aI} \vdash \bigutriple{\aI}{B}
      {\exists \vec{b_t'}, \vec{v_1}, \vec{v_2}, \vec{v_1'}, \vec{v_2'}. 
        \aI \wedge 
        \begin{array}{l}
        \aI[\vec{v_1'} / \vec{v}] \wedge  \vec{e_{t_1}'} = \vec{y} \: \wedge  \\
        \aI[\vec{v_2'} / \vec{v}] \wedge  \exists k. \vec{e_{t_2}'} = k 
        \end{array} \wedge
        \begin{array}{l}
          (\vec{v} = \vec{v_1} = \vec{v_2}) \: \wedge \\
        \vec{b_t} = (\vec{e_{t_1}'} == \vec{e_{t_2}'})
        \end{array}
      }
      }
    {
      \cdots
%      \infer[]{
%        \Gamma_{\aI} \vdash \utriple{\aI \wedge \aI[\vec{v_1} / \vec{v}] \wedge (\vec{v_1} = \vec{v})}{y}{\aI \wedge \aI[\vec{v_1} / \vec{v}] \wedge (\vec{v_1} = \vec{v}) \wedge \vec{e_t} = \vec{y}}
%      }{\cdots}
%      \quad
%      \infer[]{
%        \Gamma_{\aI} \vdash \utriple{\aI \wedge \aI[\vec{v_2} / \vec{v}] \wedge (\vec{v_2} = \vec{v})}{y}{\aI \wedge \exists k. \vec{e_t} = k}
%      }{\cdots}
    }
  }
  $$
  }
  We take such an approach to take advantage of the fact that the branch condition is actually irrelevant to 
  the proof of unrealizability; this example shows that sometimes triples inside the proof tree 
  need not be precise to prove unrealizability (which can, as shown, greatly simplify predicates).

  The second premise is an instance of $A$; one can prove this premise by using the same proof tree
  as we used to prove $\Gamma_{\aI} \vdash \utriple{\aI}{A}{\aI}$ in the base case, and apply an
  additional $\weaken$ to it.
  $$
    \infer[\weaken]{\Gamma_{\aI} \vdash \utriple{\aI \wedge \vec{v_1} = \vec{v}}{A}{\aI}}
    {
      \infer[]{\Gamma_{\aI} \vdash \utriple{\aI}{A}{\aI}}
      {\cdots}
    }
  $$

  Finally, the third premise can be derived with a simple combination of $\applyhp$ and $\weaken$, where 
  $\applyhp$ applies the \emph{induction hypothesis} that $\utriple{\aI}{\nstart}{\aI}$.
  $$
    \infer[\weaken]{\Gamma_{\aI} \vdash \utriple{\aI \wedge \vec{v_2} = \vec{v}}{\nstart}{\aI}}
    {
      \infer[\applyhp]{\Gamma_{\aI} \vdash \utriple{\aI}{\nstart}{\aI}}
      {}
    }
  $$
  
  Returning to the application of $\siterule$, observe that $\aI$ lacks any occurence of 
  $\vec{v_1}$ and $\vec{v_2}$.
  This is again because our triples for $A$ and $\nstart$ in the premises were not exact; $\aI$ is an 
  \emph{overapproximation} of the set of states that may occur when, e.g., executing $\nstart$ over the 
  input precondition $\aI \wedge \vec{v_2} = \vec{v}$.
  Although the derived postcondition is thus also not exact, it is still a sound overapproximation that is 
  \emph{precise enough} for the proof.
  To see this, 
  observe that mixing examples from two states where 
  `$\forall i. y_i = i \wedge \text{only a finite no. of } i \ \text{such that } x_i = y_i$' holds (the postcondition 
  of the conclusion of $\siterule$)
  still results in a state where `$\forall i. y_i = i \wedge \text{only a finite no. of } i \ \text{such that } x_i = y_i$' holds.
  Thus the final application of $\weaken$ that derives 
  $\utriple{\aI}{\Eifthenelse{B}{A}{\nstart}}{\aI}$ is a valid application.
  %TODO: Discuss imprecision here

  By proving that $\Gamma \vdash \utriple{\aI}{A}{\aI}$ and $\Gamma \vdash \utriple{\aI}{\Eifthenelse{B}{A}{\nstart}}{\aI}$, 
  we have proven the induction hypothesis; 
  thus $\utriple{\aI}{\nstart}{\aI}$ holds, which can further be 
  weakened down into the target triple 
  $\utriple{\forall i. x_i = -1 \wedge y_i = i}{\nstart}{\exists i. x_i \neq i \vee y_i \neq i}$.
  Hence $\syite$ is \emph{unrealizable}, and we have proved this fact quite 
  simply via some imprecise reasoning.
\end{example}

In Example~\ref{ex:inf_example}, we used predicates such as 
$\forall i. x_i = i \wedge y_i = i$, or 
`$\text{only a finite no. of } i \ \text{such that } x_i = y_i$'; 
whether or not such predicates are supported 
is dictated by the choice of the assertion language.
As mentioned in \S\ref{Se:UnrealizabilityLogic}, the choice of assertion language is 
parametric in unrealizability logic, as long as it contains the operators 
used in the rules.

\subsection{Loops and Expressing Proof Strategies from Other Frameworks}

In this section, we give an example of how loops are dealt with in 
unrealizability logic, and also show that the reasoning behind other frameworks 
for proving unrealizability (e.g., Nay~\cite{nay} or MESSY~\cite{semgus}) can 
be captured as a \emph{proof strategy} for completing a unrealizability logic proof
tree.

\begin{example}[Proof with Loops]
  \label{ex:easy_loop}
  Consider a synthesis problem $\sysum$ where the goal is to synthesize a function 
  $f$ that takes as input a state $(x, y)$, and assigns to $y$ the sum of all integers 
  between $1$ and $x$.
  Let as assume that $\sysum$ is given the following grammar $\gsum$:
  $$
  \begin{array}{lcl}
    \nstart & \rightarrow & \Ewhile{B}{S} \\
    B       & \rightarrow & E < E \\ 
    S       & \rightarrow & \Eassign{x}{E} \mid \Eassign{y}{E} \mid \Eseq{S}{S} \\ 
    E       & \rightarrow & x \mid y \mid E + E \mid E - E
  \end{array}
  $$
  Then, the problem $\sysum$ is unrealizable because 
  if $x$ and $y$ are both even, then $E$ can only produce even values 
  (note that $E$ does not contain productions such as $E + 1$).
  This fact conflicts with cases where, e.g., $x = 2$, in which case $y$ must be assigned $3$.
  $\sysum$ is introduced as an unrealizable problem in~\citet{semgus}, where 
  the solver MESSY exploits this fact to prove that $\sysum$ is 
  indeed unrealizable when given the single input example $(x, y) = (2, 0)$.
  We show that this argument can be mimicked directly as a proof tree for unrealizability logic.

  Let $\aI$ denote the condition $\smod{x}{0}{2} \wedge \smod{y}{0}{2}$, i.e., that $x$ and $y$ are 
  both even (because we use a single example, we temporarily drop the vector notation).
  We wish to use $\aI$ as an invariant for the loop $\Ewhile{B}{S}$.
  Begin with an application of $\whilerule$ for $\nstart$, 
  where, similar to Example~\ref{ex:inf_example}, 
  the loop condition is irrelevant to the proof and thus can be skipped:
  $$
    \infer[\whilerule]{ \vdash \utriple{\aI}{\Ewhile{B}{S}}{\aI \wedge \vec{b_t} = \Efalse}}
    {
      \vdash \utriple{\aI}{B}{\aI} \quad 
      \vdash \utriple{\aI}{S}{\aI}
    }
  $$
  The first premise can be proved via $\weaken$, as we did in 
  Example~\ref{ex:inf_example}.
  The implication condition of the $\whilerule$ rule has been omitted, as in this case 
  $\aI_B' \equiv \aI_B \equiv \aI$ and thus the implication is trivial.

  We wish to prove that $\vdash \utriple{\aI}{S}{\aI}$; let us introduce 
  $\utriple{\aI}{S}{\aI}$ as a hypothesis.
  We denote the context containing only this triple as $\Gamma_S$.
  We then have the proof obligation:
  $$
    \infer[\hp]{\vdash \utriple{\aI}{S}{\aI}}
    {
      \Gamma_S \vdash \utriple{\aI}{\Eassign{x}{E}}{\aI} \quad 
      \Gamma_S \vdash \utriple{\aI}{\Eassign{y}{E}}{\aI} \quad 
      \Gamma_S \vdash \utriple{\aI}{\Eseq{S}{S}}{\aI}
    }
  $$
  The first two premises can be proved by introducing a hypothesis
  $\utriple{\smod{x}{0}{2} \wedge \smod{y}{0}{2}}{E}{\smod{e_t}{0}{2}}$.
  The third premise can be proved via two applications of $\applyhp$.
  This completes that $\utriple{\aI}{S}{\aI}$, and therefore that 
  $\utriple{\aI}{\Ewhile{B}{S}}{\aI \wedge \vec{b_t} = \Efalse}$.

  Finally, we apply $\mathsf{Weaken}$ to 
  $\utriple{\aI}{\Ewhile{B}{S}}{\aI \wedge \vec{b_t} = \Efalse}$ to
  obtain the triple 
  $\utriple{x = 2 \wedge y = 0}{\nstart}{y \neq 3}$.
  We have thus proved that $\sysum$ is unrealizable, mainly by using an invariant 
  that is valid across the entire set of possible loop bodies.
\end{example}

In \S\ref{Se:MotivatingExample}, specifically Example~\ref{ex:basic_unreal}, we 
highlighted the importance of finding good hypotheses and invariants 
for completing a proof in 
unrealizability logic.
In Example~\ref{ex:easy_loop} above,
knowing the invariant $\smod{x}{0}{2} \wedge \smod{y}{0}{2}$ resulted in a very 
simple proof tree.
The algorithms implemented by external solvers (such as Spacer~\cite{spacer} for MESSY, 
or the semilinear-set approach from Nay~\cite{nay})
may be thought of as \emph{proof strategies} 
for finding these hypotheses or invariants: 
they remain parametric of the logic itself, 
while providing critical information to complete the proof trees.
Although not the focus of this paper, we hope that this view will 
allow future work in automating unrealizability logic to draw from a significant 
body of work in grammar-flow analysis~\cite{gfa} and other program/constraint-solving techniques.

\subsection{Working with Symbolic Examples and Auxiliary Variables}

In \S\ref{SubSe:InputOutputPairs}, we saw that auxiliary variables are 
insufficient to model unrealizability.
However, auxiliary variables can still be used 
to set the input examples of an unrealizability triple to be \emph{symbolic}.
A triple derived in unrealizability logic using symbolic examples 
indicates that the said triple must hold for \emph{any instantiation} of 
the symbolic examples.
In unrealizability logic, this can be used to avoid having to search 
for hard-to-find examples, instead completing a symbolic proof tree 
and checking at the end whether there are examples that can be 
used to show unrealizability.

\begin{example}[Proof with Symbolic Variables]
  \label{ex:auxiliary}
  Recall the synthesis problem $\syite$ from Example~\ref{ex:inf_example}, where the goal 
  is to synthesize a function that assigns to $x$ the (initial) value of $y$.
  We will consider a variant of $\syite$, named $\syconst$, where the goal is the same 
  but the supplied grammar 
  $\gconst$ is different:
  $$
  \begin{array}{lcl}
    \nstart & \rightarrow & \Eifthenelse{B}{A}{\nstart} \mid A \\
    B       & \rightarrow & y == E \\ 
    E       & \rightarrow & 0 \mid E + 1 \\
    N       & \rightarrow & 1 \mid 2 \mid \cdots \mid 999 \\ 
    A       & \rightarrow & \Eassign{x}{N}
  \end{array}
  $$
  In $\gconst$, the nonterminal $N$ may generate only a fixed set of integers from 
  $1$ to $999$ (as opposed to $\ginf$, which could generate any positive integer).
  $\syconst$ is unrealizable, and this time only requires one example to prove so: 
  for example, $(x, y) = (-1, 1000)$.
  However, this specific example may be difficult to find, because it uses large 
  constants.
  We show that one can avoid having to explicitly find this example 
  by completing a proof tree in unrealizability logic using a single \emph{symbolic example}---which 
  may then be instantiated (perhaps using a constraint solver)---to 
  find a concrete example for which $\syconst$ is unrealizable.

  Our goal this time is to prove the following triple (we again drop the vector-subscripts here 
  as we only have one example):
  $$
    \utriple{x = -1 \wedge y = \yaux}{\nstart}{x < 1000 \wedge y = \yaux}
  $$
  This triple states that: starting from a \emph{single} example $(x, y) = (-1, \yaux)$, we will 
  only be able to reach states in which $x < 1000$.
  Note that the given single example is made symbolic 
  through the use of the auxiliary variable $\yaux$.
  This time, we will use the following hypothesis for $\nstart$:
  \[
    \utriple
    {x < 1000 \wedge y = \yaux}
    {\nstart}
    {x < 1000 \wedge y = \yaux}
  \]
  Denote the triple above as $\aI$.
  In a similar process to the one used in Example~\ref{ex:inf_example}, 
  one can see that the hypothesis holds for the base case (assignment) 
  (we omit the application of the $\hp$ rule):
  $$
    \infer[\weaken]{\Gamma_{\aI} \vdash \utriple{\aI}{A}{\aI}}{
      \infer[\weaken]{\Gamma_{\aI} \vdash \utriple{\aI}{A}
      {\exists x'. \aI[x' / x] \wedge x < 1000}}{
      \infer[\assignrule]{\Gamma_{\aI} \vdash \utriple{\aI}{A}
      {\exists x'. \aI[x' / x] \wedge \exists k. k < 1000 \wedge e_t = k \wedge x = e_t}}{
        \infer[\weaken]
        {\Gamma_{\aI} \vdash \utriple{\aI}{N}{\aI \wedge \exists k.  k < 1000 \wedge e_t = k}}{
        \infer[\hp]{\Gamma_{\aI} \vdash \utriple{\aI}{N}{
          \exists e_t'. \aI[e_t' / e_t] \wedge \exists k. 0 < k < 1000 \wedge e_t = k }}
        {\cdots}
    }}}}
  $$
  And also for the inductive case, If-Then-Else (where we omit the reasoning for the premises):
{\small
  $$
  \infer[\weaken]{\Gamma_{\aI} \vdash \utriple{\aI}{\Eifthenelse{B}{A}{\nstart}}{\aI}}{
    \infer[\siterule]{\Gamma_{\aI} \vdash \bigutriple{\aI}{\Eifthenelse{B}{A}{\nstart}}
    {
      \exists v_1, v_2, v_1', v_2'.
      \begin{array}{l}
        \aI[v_1'[i] / v[i] \ \text{where} \ b_{t_1}[i] = \Efalse] \wedge \\
        \aI[v_2'[i] / v[i] \ \text{where} \ b_{t_2}[i] = \Etrue]
      \end{array}
      \wedge (v_1 = v_2)
    }}
    {
      \Gamma_{\aI} \vdash \utriple{\aI}{B}{\aI} \quad 
      \Gamma_{\aI} \vdash \utriple{\aI \wedge v_1 = v}{A}{\aI} \quad 
      \Gamma_{\aI} \vdash \utriple{\aI \wedge v_2 = v}{\nstart}{\aI} \quad
    }
  }
  $$
}
  Thus $
    \utriple{x < 1000 \wedge y = \yaux}
            {\nstart}
            {x < 1000 \wedge y = \yaux}
  $ is a valid triple.
  
  At this point, observe that 
  $x < 1000 \wedge y = \yaux$ does not immediately imply 
  the negation of the specification, i.e., $x \neq \yaux \vee y \neq \yaux$.
  However, when setting $\yaux = 1000$, it does become clear that 
  $x < 1000 \wedge y = \yaux \implies x \neq 1000 \vee y \neq 1000$.

  To understand this more formally, recall that as shown in 
  Example~\ref{ex:univ_problem}, the quantification 
  for auxiliary variables happens \emph{outside} the triple:
  $$
    \forall \yaux. 
    \utriple{x < 1000 \wedge y = \yaux}
            {\nstart}
            {x < 1000 \wedge y = \yaux}
  $$
  This, in turn, means that the triple holds for all configurations of $\yaux$, 
  \emph{each of which constitutes a different example}.
  For synthesis, the program must work for \emph{all} examples: 
  thus it is sufficient that there \emph{exists} an 
  example for which the synthesis problem is unrealizable, i.e., the derived postcondition implies 
  the negation of the specification.
  Thus to check unrealizability in this scenario where we have used an auxiliary variable, 
  we can check whether the 
  following formula is valid:
  $$
    \exists \yaux. 
    (x < 1000 \wedge y = \yaux \implies x \neq 1000 \vee y \neq 1000)
  $$
  And this formula is certainly valid, as witnessed by $\yaux = 1000$.
  Thus $\syconst$ is unrealizable.

\end{example}

In general, when using a specification with symbolic examples denoted using the variables 
$\vec{\vaux}$, one can check whether $\exists \vec{\vaux}. \: \aQ \implies \lnot \mathcal{O}$ 
is a valid formula, where $\aQ$ is the derived postcondition and 
$\mathcal{O}$ denotes the desired postcondition of the synthesis problem.
As shown in Example~\ref{ex:auxiliary}, the existential quantifier over $\vec{\vaux}$ asks 
whether there \emph{exists} a concrete instantiation of the symbolic examples
such that $\aQ \implies \lnot \mathcal{O}$.
Although this approach will not be able to prove unrealizability if the supplied number 
of symbolic examples is \emph{less} than the number of examples required to show unrealizability 
(as in Example~\ref{ex:univ_problem}),
it will succeed in proving unrealizability if the supplied number of examples is sufficient.
This can be very useful if the examples required to prove unrealizability are difficult to 
find, as in Example~\ref{ex:auxiliary}.

Note that it is still \emph{sound} if one decides to drop the existential quantifier and 
simply ask whether $\aQ \implies \lnot \mathcal{O}$; adding the existential simply makes 
the final query more precise.

%%% Local Variables:
%%% mode: latex
%%% TeX-master: "paper.tex"
%%% End:

%% file: 6related.tex
\section{Related and Future Work}
\label{Se:RelatedWork}

\mypar{Unrealizability}
There has been limited work that focuses on proving the unrealizability of synthesis problems.
\nay~\cite{nay} and \nope~\cite{nope} can prove unrealizability for syntax-guided synthesis (\sygus) 
problems where the input grammar only contains expressions.
Both \nay and \nope reduce an unrealizability problem to a 
program-verification problem, and present techniques for solving the reduced problem.
\citet{KampP21} use some of the ideas presented in \nope to design specialized unrealizability-checking 
algorithms for problems involving bit-vector arithmetic.
When a grammar is not given as part of the input, CVC4~\cite{cvc4} is also capable of detecting unrealizability.
These works only focus on expression-synthesis problems. 
\messy~\cite{semgus} proposes a general algorithm for proving whether a given \semgus problem is unrealizable. 
\semgus is a general framework for specifying synthesis problems, 
  which also allows one to define synthesis problems involving imperative constructs. 
\messy is currently the only tool that can prove unrealizability for problems involving imperative programs.
\citet{unrealwitness} have a technique for proving unrealizability, which they use as part of a 
synthesis technique; however, their technique is limited to a very 
specific class of functional programs and specifications.

While previous approaches all provide ways to solve variants of unrealizability problems, 
these approaches are embedded within a specific system that employs a fixed solver or algorithm.
With the exception of \nay and its use of grammar-flow analysis, these
tools do not produce a proof artifact that can be separately verified.
For example, in \messy, one is at the mercy of an external constraint solver, 
which makes it difficult for researchers to develop new solvers
tailored towards unrealizability, or even understand why certain problems can be proved 
unrealizable while others cannot.
In contrast, unrealizability logic provides a general, logical system for (both human and machine-based) reasoning about 
unrealizability. 
While \nay can produce proof artifacts, it is limited to \sygus problems over expressions;
similar to what we showed in Example~\ref{ex:easy_loop}, the GFA algorithm of \nay may be 
understood as a proof strategy to discover hypotheses over nonterminals 
(where in this case, the assertion language is over semilinear sets).
%One may also argue that unrealizability logic distills the fundamental concepts hidden 
%within the aforementioned previous approaches---e.g., overapproximating sets of values, 
%reasoning about sets of programs at once, and such---into a single logical system, 
%similar to how Hoare logic~\cite{incorrectness} distilled fundamental concepts 
%behind bug-catching that were hidden behind bug-catching tools such as Infer~\cite{infer}.

\mypar{Hyperproperties}
The way we synchronize between multiple examples in unrealizability logic, 
using vector-states, is similar to
the concept of \emph{hyperproperties}, which are, 
in essence, sets of properties~\cite{hyperproperties}.
Hyperproperties are used to model, 
for example, $k$-safety properties, which are properties that should 
hold over $k$ separate runs of a program~\cite{cartesian} 
(for example, transitivity is a 3-safety property).

%\tom{Once we settle on the right explanation of the ``fundamental difference,'' we might
%want to move this material much earlier in the paper.}
There is a subtle but fundamental difference between properties in unrealizability logic 
and standard hyperproperties: 
properties in unrealizability logic are `properties over vector-states'; 
that is, `properties over (sets of states)', whereas 
standard hyperproperties are `sets of (properties over states)'.
The difference between the two is highlighted when considering 
nondeterminism:
when verifying $k$-safety properties for a nondeterministic program, 
one will want to let different states execute on different paths.
In contrast, unrealizability logic introduces vector-states to synchronize 
the paths---corresponding to one specific program within the set $S$---that 
each constituent state in a vector-state follows.
For example, given a grammar
$E \rightarrow \Eassign{x}{x + 1} \mid \Eassign{x}{x + 2}$, the unrealizability triple
$\utriple{x_1 = 0 \wedge x_2 = 10}{E}{(x_1 = 1 \wedge x_2 = 11) \vee (x_1 = 2 \wedge x_2 = 12)}$ 
is derivable.
However, a standard hyperproperty approach would likely wish to treat the above triple as 
invalid (taking a nondeterministic-program interpretation of the different productions of $E$).

%\tom{
%Two points still unclear:
%(a) `Properties over (sets of states)' and `sets of (properties over states)'
%both boil down to `sets of sets of states,' so where does the difference lie?
%(b) The example suggests that the difference is really in the different semantics of nondeterminism
%for the two triples $\utriple{P}{E}{Q}$ and $\triple{P}{E}{Q}$.
%Probably a separate point: $k$-safety hyperproperties can be understood in terms
%of the number of traces that it takes to refute a given property.
%Do we have any analog?
%}
%\jinwoo{I hope the rewrite makes things clearer. I think the immediate analog is 
%that a $k$-unrealizability problem requires $k$ examples to prove unrealizability?
%}

\mypar{Hoare Logic for Recursive Procedures}
The rules of unrealizability logic have 
much in common with the rules of Hoare logic extended towards 
recursive procedures.
However, as discussed in \S\ref{SubSe:WhyNotHoare}, one requires many 
features to fully support the range of features in a synthesis problem; 
for example, nondeterminism, mutual recursion, both local and global variables, 
and infinite data structures.
Combinations of some of these features have been studied previously, such as 
local variables and mutual recursion~\cite{varrecursion}, 
or nondeterminism and recursion~\cite{nipkow}.
There is a vast amount of work on variants of Hoare logic; 
\citet{50years} provides a survey of how the original Hoare logic~\cite{hoare} 
has evolved throughout the years.

Despite this body of work, we are unaware of a study of a system 
that contains \emph{all} of the features listed above, and proves
soundness, completeness, and decidability results as we have.
Also as discussed in \S\ref{SubSe:WhyNotHoare}, even though one does have 
an extended Hoare logic, such a logic would hide the fact that we are 
dealing with synthesis problems and trying to prove unrealizability---whereas 
unrealizability logic takes advantage of this fact through rules like 
$\grmdisj$.

\mypar{Supporting Nondeterministic Statements}
As previously discussed in this section and \S\ref{SubSe:WhyNotHoare}, 
there is a similarity between program-synthesis problems and 
nondeterministic, recursive procedures that has been exploited in previous 
approaches to proving unrealizability~\cite{nope}.
A natural question that this similarity leads to is: 
can nondeterministic statements
be supported in unrealizability logic as well?

In general, nondeterminism in program synthesis varies according to 
whether one takes an \emph{angelic} interpretation, 
where the specification is considered met if there exists a nondeterministic 
execution of the synthesized program that satisfies the specification, 
or a \emph{demonic} interpretation, where all executions of the synthesized program 
must meet the specification.

Supporting nondeterministic statements in unrealizability logic is simpler for 
angelic nondeterminism, in which case one can merely add a rule that generates the 
set of all possible vector-states generated by the nondeterministic statement.
For example, given a statement $\Eassign{x}{\mathsf{nondet}(V)}$, which nondeterministically 
assigns a value from the set $V$ to $x$, a (simplified) rule for nondeterminism might be 
(where $|\vec{x}|$ denotes the length of the vector-state $\vec{x}$):
$$
  \infer[\mathsf{Nondet}]{\Gamma \vdash \utriple{\aP}{\Eassign{x}{\mathsf{nondet}(V)}}
    {\exists \vec{x'}. \aP[\vec{x'} / \vec{x}] \wedge \vec{x} \in V^{|\vec{x}|}}}
  {}
$$
The postcondition of the $\mathsf{Nondet}$ rule is an overapproximation of all states that may 
be generated by $\mathsf{nondet}$.
Thus, using $\mathsf{Nondet}$ in a proof tree in tandem with other unrealizability-logic rules
will also derive an overapproximation of the set of reachable states; by showing that this set 
does not intersect with the desired specification, 
one can guarantee that the synthesis problem is unrealizable under the angelic interpretation.

However, $\mathsf{Nondet}$ is incomplete when used to prove unrealizability under a demonic 
interpretation: the demonic interpretation only requires that there \emph{exists} a nondeterministic 
choice that fails to meet the specification, 
but a proof generated using $\mathsf{Nondet}$ will report unrealizability only when 
\emph{all} possible nondeterministic choices fail to meet the specification.
To prove unrealizability in the demonic sense, one would require 
a mechanism for reasoning about the 
set of produced vector-states 
simultaneously---for example, another level of vectors.

\mypar{Unrealizability and Program Synthesis}
As briefly discussed in \S\ref{Se:Introduction}, one of the ultimate goals of 
studying unrealizability is to \emph{prune the search space} of a synthesis problem, 
by showing that certain subsets of the search space do not contain a desired solution.
Although this work is the first to distill the fundamental concepts behind unrealizability 
into a logical system, there exist synthesizers that have already 
utilized the concept of pruning unrealizable parts of the search space to some extent: 
for example, Neo~\cite{conflict} discovers unrealizable subsets of the search space 
by analyzing conflicts, and avoids traversing these subsets during the search procedure.

We hope that the logical characterization of unrealizability provided in this paper 
will lay the foundation for future attempts to utilize unrealizability towards 
synthesizing programs.

%The rules of unrealizability logic have much in common with the rules of 
%Hoare logic extended toward recursive procedures.
%In particular, the rules for introducing a triple as a hypothesis, 
%and using the hypothesis in a proof tree akin to structural induction, 
%appears in previous work on Hoare logic 
%for recursive procedures~\cite{50years, nipkow}.
%This similarity can be traced to the fact that program synthesis can be encoded as 
%a verification of a 
%nondeterministic, recursive program \cite{nope}.
%
%A key difference between Hoare logics over recursive procedures and 
%unrealizability logic is the concept of vector-states, and a semantics 
%modified to work with vector-states.
%As illustrated in \S\ref{Se:DealingWithAnInfiniteNumberOfExamples},
%extending the logic to 
%work with multiple (possibly infinite) related input-output pairs 
%raises unique challenges due to the fact that 
%a relation among input, output, and the program must be tracked.
%%% Local Variables:
%%% mode: latex
%%% TeX-master: "paper.tex"
%%% End:

%% file: 7conclusion.tex
\section{Conclusion}
\label{Se:Conclusion}

We presented \emph{unrealizability logic}, the first proof system for 
overapproximating the execution of an infinite set of programs.
Unrealizability logic is both sound and relatively complete; 
it is also the first approach that 
allows one to prove unrealizability for 
synthesis problems that require infinitely many inputs to be proved unrealizable.
We believe unrealizability  logic will prove to be essential in further developments
having to do with unrealizability.

A natural question that follows from this paper is the design
of a \emph{realizability logic}: 
``Can a similar logic be constructed for program synthesis, i.e.,
for proving \emph{realizability}?''
Because program synthesis requires a guarantee that a certain (vector-)state is 
\emph{reachable}, one must
devise suitable
principles of \emph{underapproximation}
(like those discussed in  reverse-Hoare (aka incorrectness) logic~\cite{reverse,incorrectness})
instead of overapproximation as used in this paper.
We believe the results presented in this paper will also be useful in designing a \textit{realizability logic}.

%First, in \S\ref{SubSe:WhileLoopsAndCompleteness}, we presented a sound proof rule for while-loops, but
%we could not prove its completeness.
%We leave as an open problem whether unrealizability logic is complete 
%with respect to synthesis problems containing while-loops
%(or whether a complete logic with finitely many rules exists).

%Second, this paper focused on proving unrealizability through the principle of overapproximation, 
%which allowed us to guarantee that if the desired output state is \emph{not} in the postcondition, 
%the synthesis problem is unrealizable.
%A question that begs to be answered is,
%``Can a similar logic be constructed for program synthesis, i.e.,
%for proving \emph{realizability}?''
%Because program synthesis requires a guarantee that a certain state is 
%\emph{reachable}, one must
%devise suitable
%principles of \emph{underapproximation}
%(like the ones discussed in  reverse-Hoare (aka incorrectness) logic~\cite{reverse,incorrectness})
%instead of overapproximation.

%%% Local Variables:
%%% mode: latex
%%% TeX-master: "paper.tex"
%%% End:

%% file: appendix_proofs.tex
\section{Proofs}
\label{App:Proofs}

\input{pf_logic.tex}

\input{pf_soundness.tex}

\input{pf_precision.tex}
\input{pf_strongest_triple.tex}

\input{pf_general_triple.tex}

\input{pf_completeness.tex}
\input{pf_decidability.tex}

%% file: pf_logic.tex
We will start by proving lemmas related to using unrealizability logic 
(Theorem~\ref{thm:unreal_hoare}, Lemma~\ref{lem:expression_invariant}).

\begin{reptheorem}{thm:unreal_hoare}
  The \emph{unrealizability triple} $\utriple{P}{S}{Q}$ for a precondition $P$,
  postcondition $Q$, and set of programs $S$, 
  holds iff for each program $\stmt \in S$,
  the Hoare triple $\triple{P}{\stmt}{Q}$ holds.
\end{reptheorem}

\begin{proof}
  Note here that $P, Q$ are not predicates over vector-states, but ordinary predicates.
  For an unrealizability triple, this may be interpreted as vector-states of length 1; the
  vector-state semantics are identical to usual semantics in this case.
  The proof then follows directly from the definitions.

  $\Rightarrow$: By definition of validity of $\utriple{P}{S}{Q}$,
  $\semantics{S}(P) = \bigcup_{\stmt \in S} \semantics{\stmt}(P) \subseteq Q$.
  Clearly $\semantics{\stmt}(P) \subseteq \bigcup_{\stmt \in S} \semantics{\stmt}(P)$ for any $\stmt \in S$,
  thus by definition of $\semantics{\stmt}(P)$, $\semantics{\stmt}(p) \in Q$ for all $\stmt \in S$ and $p \in P$.
  Thus $\triple{P}{\stmt}{Q}$ holds.

  $\Leftarrow$: $\forall \stmt \in S. \: \triple{P}{\stmt}{Q}$ implies that
  $\forall \stmt \in S. (\forall p \in P. \semantics{\stmt}(p) \in Q)$.
  Thus $\forall \stmt \in S. \semantics{\stmt}(P) \subseteq Q$, and then $\semantics{S}(P) \subseteq Q$.
  Thus $\utriple{P}{S}{Q}$ holds.
\end{proof}

\begin{replemma}{lem:expression_invariant}
  Given an expression nonterminal $E$,
  if an unrealizability triple $\Gamma \vdash \utriple{\aP}{E}{\aQ}$ is derivable using the rules of unrealizability logic
  under some context $\Gamma$,
  then for every $\exp\in L(E)$ and for every $\sigma \in \aP$,
  the formula $\aQ[\semantics{\exp}(\sigma)/\vec{e_t}]$ is true
  assuming all triples in $\Gamma$ are true
  (where in this case, $\semantics{\exp}(\sigma)$ refers to the standard semantics that evaluates
  $\exp$ to a value).
\end{replemma}
\begin{proof}
  By induction on the length of the proof tree, and case analysis on the final rule used to derive $\utriple{\aP}{E}{\aQ}$.
  For example, for proof trees of length $1$ where the final applied rule is $\mathsf{Zero}$, the rule holds as the postcondition
  assigns $\vec{0}$ to $\vec{e_t}$.
  For longer proof trees, for example where the final rule applied is $\plusrule$, the induction hypothesis holds for the premises;
  then, examining the postcondition of the conclusion shows that the
  induction hypothesis still holds (as $\vec{e_t} = \vec{e_{t1}} + \vec{e_t'}$.
  (For constituent state that are not executed, $\vec{e_t}$ is unbound and thus $\aQ[\semantics{e} / \vec{e_t}]$ still holds).
\end{proof}

%% file: pf_soundness.tex
We will now move to proving soundness (Theorem~\ref{thm:soundness}).

\begin{reptheorem}{thm:soundness}[Soundness]
  Given a regular tree grammar $G$ and a set of programs $N$ 
  (which is either the language of a nonterminal in $G$ or the set of programs generated by the 
   RHS of a production in $G$), 
  the following is true: $\vdash \utriple {\aP}{N}{\aQ} \implies \models \utriple{\aP}{N}{\aQ}$.
\end{reptheorem}

\begin{proof}
  Soundness is proved by performing structural induction on each inference rule in unrealizability 
  logic: we will show soundness by showing that the conclusion of each rule is sound assuming the 
  premises are sound.

  The soundness proof will also entail that some of the rules are \emph{precise}: that is, 
  if $\utriple{\aP}{N}{\aQ}$ is the conclusion of an inference rule, that $\semantics{\aP}(N) = \aQ$.
  Although preciseness of the rules is not required in the proof of soundness, we will refer 
  these parts of the soundness proof when proving completeness later.

  We proceed by case analysis, on expression rules first.
  Throughout the proof, we will use $\pi$ to denote states from the precondition $\aP$, and
  $\sigma$ to denote states from the postcondition $\aQ$.

  $\zero, \one, \varrule, \truerule, \falserule$: These serve as the base cases for the 
  entire set of rules.
  We take $\varrule$ as an example.

	For $\varrule$, $N$ is the singleton set $\set{x}$. 
  Recall that our semantics for expressions are not to simply compute a set of values, but to 
  assign the computed values to a reserved auxiliary (vector-)variable $\vec{e_t}$. 
  Then by definition of the vectorized semantics
 	and the extended expression semantics,
  \begin{align*}
    \semantics{N}(\aP) & = \semantics{x}(\aP) \\
                       & = \bigcup_{\pi \in \aP} \semantics{x}(\pi) \\
                       & = \bigcup_{\pi \in \aP} \lrangle{\semantics{x}(\state_1), \cdots} 
                       \tag{$\state_1$ is a single example-state in $\pi$}\\
                       & = \bigcup_{\pi \in \aP} \lrangle{\state_1[e_{t_1} \mapsto x_1], \cdots }\\
                       & = \bigcup_{\pi \in \aP} {\pi[\vec{e_t} \mapsto \vec{x}]}\\
                       & = \exists \vec{e_t'}. \aP[\vec{e_t'} / \vec{e_t}] \wedge \vec{e_t} = \vec{x}
  \end{align*}
  Since $\semantics{N}(\aP) \in \exists \vec{e_t'}. \aP[\vec{e_t'} / \vec{e_t}] \wedge \vec{e_t} = \vec{x}$, 
  the postcondition of the conclusion, this rule is sound.
  Also note that the sequence is an equality; the rule is also \emph{precise}.
  $\zero, \one, \truerule,$ and $\falserule$ can be proved sound and precise in a similar manner.

  $\notrule$: Assuming the premise $\utriple{\aP}{B}{\aQ}$ is sound (via the induction hypothesis), 
  $\notrule$ can be proved sound through the following equation:
  \begin{align*}
    \semantics{N}(\aP) & = \semantics{!B}(\aP) \\
                       & = \bigcup_{\bexp \in B} \bigcup_{\pi \in \aP} \semantics{! \bexp}(\pi) 
 													 \tag{$!\bexp \in !B \iff \bexp \in B$}\\
                       & = \bigcup_{\bexp \in B} \bigcup_{\pi \in \aP} \lrangle{\semantics{!\bexp}(\state_1), \cdots } \\
                       & = \bigcup_{\bexp \in B} \bigcup_{\pi \in \aP} 
                           \lrangle{\semantics{b}(\state_1)[b_{t_1} \mapsto ! \semantics{\bexp}(\state_1)[b_{t_1}]], \cdots } \\
                       & \subseteq \bigcup_{\sigma \in \aQ} \lrangle{\sigma[\vec{b_t} \mapsto ! \sigma[\vec{b_t}]]}
                           \tag{$\bigcup_{\bexp \in B} \bigcup_{\pi \in \aP} \semantics{\bexp}(\pi) \subseteq \aQ$, \text{ by I.H.}}\\
                       & = \exists \vec{b_t'}. \aQ[\vec{b_t'} / \vec{b_t}] \wedge \vec{b_t} = \lnot \vec{b_t'}
  \end{align*}
  $\mathsf{Bin}$, $\mathsf{Comp}$, $\mathsf{And}$: These are rules for binary operators; as discussed in 
  \S\ref{Se:UnrealizabilityLogic}, these rules require extra synchronization between the postconditions of the 
  premises.

  We take $\plusrule$ as an example, assuming that the two premises are sound via the induction hypothesis.
  We will first compute the set of states that may arise from $E_1 + E_2$ (without committing the update to 
  $\vec{e_t}$).
  First observe that:
  \begin{align*}
    \semantics{N}(\aP) & = \semantics{E_1 + E_2}(\aP) \\
                       & = \bigcup_{e_1 \in E_1, e_2 \in E_2} \bigcup_{\pi \in \aP} \semantics{e_1 + e_2}(\pi) \\
                       & = \bigcup_{e_1 \in E_1, e_2 \in E_2} \bigcup_{\pi \in \aP} \semantics{e_1}(\pi) + \semantics{e_2}(\pi) \\
  \end{align*}
  %By the extended semantics of $+$, we also have that:
  %\begin{align*}
  %  \bigcup_{e_1 \in E_1, e_2 \in E_2} \bigcup_{\pi \in \aP} \semantics{e_1}(\pi) + \semantics{e_2}(\pi) = 
  %    \bigcup_{e_1 \in E_1, e_2 \in E_2} \bigcup_{\pi \in \aP} 
  %    \pi[\vec{e_t} \mapsto (\semantics{e_1}(\pi) + \semantics{e_2}(\pi))[\vec{e_t}]]
  %\end{align*}
  %We will denote the outer expression part as $\upd(\pi, \semantics{e_1}(\pi) + \semantics{e_2}(\pi))$.

  %We will drop the update part, and focus on the inner expression $\semantics{e_1}(\pi) + \semantics{e_2}(\pi)$ for now.

  Observe that it is impossible to say that simply 
  \begin{align*}
    \bigcup_{e_1 \in E_1, e_2 \in E_2} \bigcup_{\pi \in \aP} \semantics{e_1}(\pi) + \semantics{e_2}(\pi) & = 
    \bigcup_{e_1 \in E_1} \bigcup_{\pi \in \aP} \semantics{e_1}(\pi) + \bigcup_{e_2 \in E_2} \bigcup_{\pi \in \aP} \semantics{e_2}(\pi)
  \end{align*}
  Because as noted multiple times, the equation on the RHS fails to capture that the same $\pi$ enters $e_1$ 
  and $e_2$.

  To remedy this fact, we first introduce the ghost state discussed in \S\ref{Se:MotivatingExample}.
  Given two (vector-)states over disjoint vocabularies $\pi$ and $\sigma$, define $\ext{\pi}{\sigma}$ to be 
  the extension of $\pi$ and $\sigma$, i.e., a state where if a variable $v \in \pi$ or $v \in \sigma$, then 
  $v \in \ext{\pi}{\sigma}$ and $\ext{\pi}{\sigma}(v) = \pi(v)$ (or $\sigma(v)$).
  Conversely, we define two kinds of projections: 
  the exclusive projection of $\pi$ on a vocabulary $V(\sigma)$, denoted as $\proj{\pi}{V(\sigma)}$, 
  as the state obtained by \emph{excluding all} variables in $V(\sigma)$ from $\pi$.
  We also define an inclusive projection of $\pi$ on a vocabulary $V(\sigma)$, denoted as 
  $\proji{\pi}{V(\sigma)}$, as the state obtained by 
  \emph{including only} variables in $V(\sigma)$ from $\pi$.

  Now observe that:
  {\small
  \begin{align*}
    \bigcup_{e_1 \in E_1, e_2 \in E_2} \bigcup_{\pi \in \aP} \semantics{e_1}(\pi) + \semantics{e_2}(\pi)
    = \bigcup_{\tiny \begin{array}{c} e_1 \in E_1 \\ e_2 \in E_2 \end{array}} \!\!\!\!
    \bigcup_{\tiny \begin{array}{c} \pi \in \aP \\ \pi_1 = \pi[\vec{v_1} / \vec{v}] \\ \pi_2 = \pi[\vec{v_2} / \vec{v}]\end{array}}
      \!\!\!\!
      \proj{\semantics{e_1}(\ext{\pi}{\pi_1}}{\vec{v_1}} + \proj{{e_2}(\ext{\pi}{\pi_2}}{\vec{v_2}}
  \end{align*}
  }
  We slightly abuse notation and write $\vec{v_1}$ to indicate $V(\sigma_1)$; same goes for 
  $\vec{v_2}$.

  We now wish to pull the addition outside of the set union.
  Define an extended plus operator, $\oplus$, that works on sets of extended vector-states as follows: 
  $\aP_1 \oplus \aP_2 = \aQ$ iff for every $\ext{\sigma}{\sigma'} \in \aQ$, there exists 
  $\ext{\pi_1}{\pi_1'} \in \aP_1$ and $\ext{\pi_2}{\pi_2'} \in \aP_2$ such that $\pi_1 + \pi_2 = \sigma$ (according to 
  the extended assignment semantics) and 
  $\proji{\ext{\pi_1}{\pi_1'}}{\vec{v_1}} = 
  \proji{\ext{\pi_2}{\pi_2'}}{\vec{v_2}}[\vec{v_1} / \vec{v_2}] = \sigma'[\vec{v_1} / V(\sigma)]$.
  Intuitively speaking, $\oplus$  wishes to add vector-states 
  where the extended part---$\pi_1'$ and $\pi_2'$, i.e., the `ghost state'---is identical.
  However, observe that $\vec{v_1}$ and $\vec{v_2}$ are fixed sets of variables,
  and have, in general, 
  no relation at all with the vocabulary 
  of $\pi_1, \pi_1', \pi_2, $ and $\pi_2'$, which are arbitrary states 
  (one may treat $\vec{v_1}$ and $\vec{v_2}$ as being baked in to the definition of $\oplus$).
  Thus checking the equivalence of the `ghost part' will only take effect when $\pi_1'$ and $\pi_2'$ are defined over the 
  vocabulary $\vec{v_1}$ and $\vec{v_2}$; otherwise, $\oplus$ will take a imprecise cartesian-product style addition.\footnote{
    It is tempting to check simply that $\pi_1 = \pi_2$, but this will fail if the premise triples are 
    imprecise and decide to forget the extended part of the state, as we have done in Example~\ref{ex:inf_example}.
  }

  Now observe that:
   {\small
  \begin{align*}
    & \bigcup_{\tiny \begin{array}{c} e_1 \in E_1 \\ e_2 \in E_2 \end{array}} \!\!\!\!
    \bigcup_{\tiny \begin{array}{c} \pi \in \aP \\ \pi_1 = \pi[\vec{v_1} / \vec{v}] \\ \pi_2 = \pi[\vec{v_2} / \vec{v}]\end{array}}
      \!\!\!\!
      \proj{\semantics{e_1}(\ext{\pi}{\pi_1}}{\vec{v_1}}) + \proj{\semantics{e_2}(\ext{\pi}{\pi_2}}{\vec{v_2}}) \\
    = \ \ & \proj{\bigcup_{e_1 \in E_1} \!\!\!\!
    \bigcup_{\tiny \begin{array}{c} \pi \in \aP \\ \pi_1 = \pi[\vec{v_1} / \vec{v}] \end{array}} \!\!\!\!
      \semantics{e_1}(\ext{\pi}{\pi_1})
      \oplus
    \bigcup_{e_2 \in E_2} \!\!\!\!
    \bigcup_{\tiny \begin{array}{c} \pi \in \aP \\ \pi_2 = \pi[\vec{v_2} / \vec{v}] \end{array}} \!\!\!\!
      \semantics{e_2}(\ext{\pi}{\pi_2})}{\vec{v_1} \cup \vec{v_2}}
  \end{align*}
  }
  The equality holds because at this point, $\pi_1$ and $\pi_2$ are actually defined over the 
  vocabularies $\vec{v_1}$ and $\vec{v_2}$.

  At this point, we may finally invoke the induction hypothesis: observe that the union over possible precondition 
  states $\pi \in \aP, \pi_1 = \pi[\vec{v_1} / \vec{v}]$ is identical to the predicate-form precondition of the premise 
  $\aP \wedge \aP[\vec{v_1} / \vec{v}] \wedge (\vec{v_1} = \vec{v})$.
  We consider a state $\sigma' \in \aQ$ is guaranteed to be an extended state of form 
  $\ext{\sigma}{\sigma_1}$ (even if $\aQ$ does not refer to variables such as $\vec{v_1}$, 
  one may still treat the extended part as being unbound).
  This yields:
    {\small
  \begin{align*}
  & \proj{\bigcup_{e_1 \in E_1} \!\!\!\!
    \bigcup_{\tiny \begin{array}{c} \pi \in \aP \\ \pi_1 = \pi[\vec{v_1} / \vec{v}] \end{array}} \!\!\!\!
      \semantics{e_1}(\ext{\pi}{\pi_1})
      \oplus
    \bigcup_{e_2 \in E_2} \!\!\!\!
    \bigcup_{\tiny \begin{array}{c} \pi \in \aP \\ \pi_2 = \pi[\vec{v_2} / \vec{v}] \end{array}} \!\!\!\!
      \semantics{e_2}(\ext{\pi}{\pi_2})}{\vec{v_1} \cup \vec{v_2}} \\ 
    \subseteq \ \ & 
     \proj{
      \bigcup_{\ext{\sigma_1}{\sigma_1'} \in \aQ_1} \ext{\sigma_1}{\sigma_1'} \oplus 
      \bigcup_{\ext{\sigma_2}{\sigma_2'} \in \aQ_2} \ext{\sigma_2}{\sigma_2'}
    }{\vec{v_1} \cup \vec{v_2}}
  \end{align*}
  }
  %Also observe that on non-auxiliary variables, $\sigma_1 = \sigma_2 = \pi$ (as expression nonterminals 
  %only modify the auxiliary variables).
  We now wish to move the $\oplus$ back inside the union operator as ordinary $+$.
  Recall that $\oplus$ has the variables $\vec{v_1}$ and $\vec{v_2}$ baked into the definition;
  thus we have:
   \begin{align*}
     & \proj{
      \bigcup_{\ext{\sigma_1}{\sigma_1'} \in \aQ_1} \ext{\sigma_1}{\sigma_1'} \oplus 
      \bigcup_{\ext{\sigma_2}{\sigma_2'} \in \aQ_2} \ext{\sigma_2}{\sigma_2'}
    }{\vec{v_1} \cup \vec{v_2}} \\
     = \ \ & \proj{
      \bigcup_{\tiny
        \begin{array}{c}
          \ext{\sigma_1}{\sigma_1'} \in \aQ_1\\
          \ext{\sigma_2}{\sigma_2'} \in \aQ_2\\
          \proji{\ext{\sigma_1}{\sigma_1'}}{\vec{v_1}} = \proji{\ext{\sigma_2}{\sigma_2'}}{\vec{v_2}}[\vec{v_1} / \vec{v_2}] 
        \end{array}
      }\!\!\!\!
      \ext{\sigma_1}{\sigma_1'}+ \ext{\sigma_2}{\sigma_2'}
    }{\vec{v_1} \cup \vec{v_2}}
  \end{align*}
  Let us simply write 
  $\proji{\ext{\sigma_1}{\sigma_1'}}{\vec{v_1}} = \proji{\ext{\sigma_2}{\sigma_2'}}{\vec{v_2}}[\vec{v_1} / \vec{v_2}]$  
  as $\vec{v_1} = \vec{v_2}$: note that the two are equivalent.

  The expression above computes the set of states that we will draw the vector-variable $\vec{e_t}$ from; 
  now we will update the original state $\pi \in \aP$ with the new vector variable.
  Note that it is sound (and precise!) to perform the update only when 
  $\vec{v} = \vec{v_1} = \vec{v_2}$ 
  (that is, when $\pi = \proji{\ext{\sigma_1}{\sigma_1'}}{\vec{v_1}}[V(\pi) / \vec{v_1}] = 
  \proji{\ext{\sigma_2}{\sigma_2'}}{\vec{v_2}}[V(\pi) / \vec{v_2}]$), 
  as one only wants to update an original state with a value 
  obtained from executing that original state.
  Factoring the update in, we obtain:
  {\small
   \begin{align*}
     & \proj{
      \bigcup_{\tiny
        \begin{array}{c}
          \pi \in \aP \\
          \ext{\sigma_1}{\sigma_1'} \in \aQ_1\\
          \ext{\sigma_2}{\sigma_2'} \in \aQ_2\\
          \vec{v} = \vec{v_1} = \vec{v_2} 
        \end{array}
      }\!\!\!\!
      \pi[\vec{e_t} \mapsto
        \ext{\sigma_1}{\sigma_1'} + \ext{\sigma_2}{\sigma_2'}[e_t]]
      }{\vec{v_1} \cup \vec{v_2}} \\ 
     = \ \ & \proj{
      \bigcup_{\tiny
        \begin{array}{c}
          \pi \in \aP \\
          \ext{\sigma_1}{\sigma_1'} \in \aQ_1\\
          \ext{\sigma_2}{\sigma_2'} \in \aQ_2\\
          \vec{v_1} = \vec{v_2} 
        \end{array}
      }\!\!\!\!
      \pi[\vec{e_t} \mapsto
        \ext{\sigma_1}{\sigma_1'}[\vec{v_1'} / \vec{v}][\vec{e_{t_1}'}] + \ext{\sigma_2}{\sigma_2'}[\vec{v_2'} / \vec{v}][\vec{e_{t_2}'}]]
      }{\vec{v_1'} \cup \vec{v_2'} \cup \vec{v_1} \cup \vec{v_2}} 
  \end{align*}
  }
  The equality holds because although the variables have been renamed; the set of states (i.e., when interpreted as a simple ordered tuple) 
  renames the same.
  The second formula also uses the `standard' semantics of $+$ on simple values instead of our extended semantics, 
  hence the reference to $\vec{e_{t_1}}'$ and $\vec{e_{t_2}}'$ before the +.

  Finally, expressing the last line as a formula, we obtain: 
  $$
    \exists \vec{e_t', v_1', v_2', v_1, v_2}. 
      (\aP \wedge \aQ_1[\vec{v_1'} / \vec{v}] \wedge \aQ_2[\vec{v_2'} / \vec{v}] \wedge (\vec{v} = \vec{v_1} = \vec{v_2}))[\vec{e_t'} / \vec{e_t}]
      \wedge \vec{e_t} = \vec{e_{t_1}'} + \vec{e_{t_2}'}
  $$
  Which is the postcondition of $\plusrule$.

  $\assignrule$: Assuming the premise $\utriple{\aP}{E}{\aQ}$ is precise, 
  \begin{align*}
    \semantics{N}(\aP) & = \semantics{\Eassign{x}{E}}(\aP) \\
                       & = \bigcup_{\exp \in E} \bigcup_{\pi \in \aP} \semantics{\Eassign{x}{\exp}}(\pi) \\
                       & = \bigcup_{\exp \in E} \bigcup_{\pi \in \aP} \lrangle{\semantics{\Eassign{x}{\exp}}(\state_1), \cdots} \\
                       & = \bigcup_{\exp \in E} \bigcup_{\pi \in \aP} \lrangle{\semantics{\exp}(\state_1)
                           [x \mapsto \semantics{\exp}(\state_1)[e_{t_1}]], \cdots } \\
                       & \subseteq \bigcup_{\sigma \in \aQ} \lrangle{\sigma[\vec{x} \mapsto \sigma[\vec{e_t}]], \cdots}
                           \tag{$\bigcup_{\exp \in E} \bigcup_{\pi \in \aP} \semantics{\exp}(\pi) \subseteq \aQ$, \text{ by I.H.}} \\
                       & = \exists \vec{x'}. \aQ[\vec{x'} / \vec{x}] \wedge \vec{x} = \vec{e_t}
  \end{align*}
	Thus $\assignrule$ is sound.

  $\seqrule$: Assume (by the induction hypothesis) that the two premises $\utriple{\aP}{S_1}{\aQ}$ and 
	$\utriple{\aQ}{S_2}{\aR}$ are sound. Then
  \begin{align*}
    \semantics{N}(\aP) & = \semantics{\Eseq{S_1}{S_2}}(\aP) \\
                       & = \bigcup_{\stmt_1 \in S_1} \bigcup_{\stmt_2 \in S_2} \bigcup_{\pi \in \aP} 
													 \semantics{\Eseq{\stmt_1}{\stmt_2}}(\pi) \\
											 & = \bigcup_{\stmt_1 \in S_1} \bigcup_{\stmt_2 \in S_2} \bigcup_{\pi \in \aP}
													 \semantics{\stmt_2}(\semantics{\stmt_1}(\pi)) \\
                       & \subseteq \bigcup_{\stmt_2 \in S_2} \bigcup_{\sigma \in \aQ} \semantics{\stmt_2}(\sigma)
                       \tag{$\bigcup_{\stmt_1 \in S_1} \bigcup_{\pi \in \aP} \semantics{\stmt_1} \subseteq \aQ$, \text{ by I.H.}} \\
                       & \subseteq \aR
                       \tag{$\bigcup_{\stmt_2 \in S_2} {\bigcup \sigma \in \aQ} \semantics{\stmt_2} \subseteq \aR$, by I.H.}
  \end{align*}
  Thus $\seqrule$ is sound.

  $\siterule$: Like $\plusrule$, $\siterule$ requires re-composing states that have 
  gone through different nonterminals.
  We start with:

  \begin{align*}
    \semantics{N}(\aP) & = \!\!\!\! \bigcup_{\tiny \begin{array}{c} b \in B \\ s_1 \in S_1 \\ s_2 \in S_2 \end{array}} \!\!\!\! 
      \bigcup_{\pi \in \aP} \semantics{\Eifthenelse{b}{s_1}{s_2}}(\pi)
  \end{align*}
  Under our extended expression semantics, the branch condition $b$ 
  first executes to assign a value to the auxiliary variable $\vec{b_t}$, 
  then i) the updated state executes through the two branches, 
  ii) the result gets merged according to the branch condition.

  As we did in the case for $\plusrule$, the merge operation can be defined using $\mathsf{ext}$ and $\mathsf{proj}$.
  Consider the following formula:
  \begin{align*}
    \mathsf{ext}( & {\proj{\tbranch[\vec{v_1'}[i] / \vec{v}[i] \ \text{where} \ \semantics{b}(\pi) [\vec{b_t}][i] = \Efalse}
    {\vec{v_1'}}}, \\
                  & {\proj{\fbranch[\vec{v_2'}[i] / \vec{v}[i] \ \text{where} \ \semantics{b}(\pi) [\vec{b_t}][i] = \Etrue}
                  {\vec{v_2'}}})
  \end{align*}
  The first $\mathsf{proj}$ in this formula 
  executes $\tbranch$ (i.e., the $\Etrue$ branch), then \emph{projects out} those cases where $\vec{b_t}[i] = \Efalse$---that is, 
  the examples that should not have gone through the $\Etrue$ branch.\footnote{
    The conditional substitution can be defined simply using if-then-else, e.g., 
    $\Eifthenelse{\vec{b_t}[i] = \Efalse}{\vec{v_1'}[i]}{\vec{v}[i]}$.
  }
  Similarly, the second $\mathsf{proj}$ will execute the $\Efalse$ branch and project out examples that should go through the 
  $\Etrue$ branch instead.
  The $\mathsf{ext}$ operation then reunites these examples together---one can quickly see how this set of states 
  is equivalent to the `standard' accepted states generated by If-Then-Else, i.e., 
  $\lrangle{\semantics{\Eifthenelse{b}{s_1}{s_2}(\state_1)}, \cdots}$.
  We will denote this operator as $\star$; e.g., $\tbranch \star \fbranch$ 
  is equivalent to the formula above.
  We now have:
  \begin{align*}
   \!\!\!\! \bigcup_{\tiny \begin{array}{c} b \in B \\ s_1 \in S_1 \\ s_2 \in S_2 \end{array}} \!\!\!\! 
     \bigcup_{\pi \in \aP} \semantics{\Eifthenelse{b}{s_1}{s_2}}(\pi) 
   = \bigcup_{b \in B} \!\! \bigcup_{\tiny \begin{array}{c} s_1 \in S_1 \\ s_2 \in S_2 \end{array}} \!\!\!\! 
     \bigcup_{\pi \in \aP} \tbranch \star \fbranch\\
  \end{align*}
  We now wish to invoke the induction hypothesis over the first premise triple $\utriple{\aP}{B}{\aP_B}$ 
  (which is why we pulled the union over $b \in B$ out).
  This yields:
  \begin{align*}
  \bigcup_{b \in B} \!\! \bigcup_{\tiny \begin{array}{c} s_1 \in S_1 \\ s_2 \in S_2 \end{array}} \!\!\!\! 
  \bigcup_{\pi \in \aP} \tbranch \star \fbranch 
    \subseteq 
    \bigcup_{\tiny \begin{array}{c} s_1 \in S_1 \\ s_2 \in S_2 \end{array}} \bigcup_{\beta \in \aP_B} 
      \tbranchbeta \star \fbranchbeta
  \end{align*}
  Observe that in the RHS, $\star$ is expanded to: 
  \begin{align*}
    \mathsf{ext}( & {\proj{\tbranchbeta[\vec{v_1'}[i] / \vec{v}[i] \ \text{where} \ \beta[\vec{b_t}][i] = \Efalse]}
    {\vec{v_1'}}}, \\
                  & {\proj{\fbranchbeta[\vec{v_2'}[i] / \vec{v}[i] \ \text{where} \ \beta[\vec{b_t}][i] = \Etrue]}
                  {\vec{v_2'}}})
  \end{align*}
  This is due to the induction hypothesis being applied to the conditional substitution part as well.

  We now wish to bring the $\star$ operator outside of the set union.
  To do so, we must take a similar approach as we did with $\plusrule$, and first introduce two additional ghost states 
  (one for each branch).
  Observe that:
  \begin{align*}
    & \bigcup_{\tiny \begin{array}{c} s_1 \in S_1 \\ s_2 \in S_2 \end{array}} \bigcup_{\beta \in \aB} 
      \tbranchbeta \star \fbranchbeta \\
  = & \!\!\!\! \bigcup_{\tiny \begin{array}{c} s_1 \in S_1 \\ s_2 \in S_2 \end{array}} \!\!\!\! 
    \bigcup_{ \tiny \begin{array}{c} \beta \in \aP_B \\ 
                                     \beta_1 = \beta[\vec{v_1} / \vec{v}] \\ \beta_2 = \beta[\vec{v_2} / \vec{v}] \end{array}} 
    \tbranchext \star \fbranchext\\
  \end{align*} 
  At this point, observe that: 
   \begin{align*}
     \mathsf{ext}( & {\proj{\tbranchext[\vec{v_1'}[i] / \vec{v}[i] \ \text{where} \ \beta[\vec{b_t}][i] = \Efalse]}
    {\vec{v_1'}}}, \\
                  & {\proj{\fbranchext[\vec{v_2'}[i] / \vec{v}[i] \ \text{where} \ \beta[\vec{b_t}][i] = \Etrue]}
    {\vec{v_2'}}}) \\ 
     = \mathsf{ext}( & {\proj{\tbranchext[\vec{v_1'}[i] / \vec{v}[i] \ \text{where} \ \beta_1[\vec{b_{t_1}}][i] = \Efalse]}
    {\vec{v_1'}}}, \\
                  & {\proj{\fbranchext[\vec{v_2'}[i] / \vec{v}[i] \ \text{where} \ \beta_2[\vec{b_{t_2}}][i] = \Etrue]}
    {\vec{v_2'}}}) \\ 
  \end{align*}
  Because $\beta = \beta_1[\vec{v_1} / \vec{v}]$.
  We will shift to this latter definition of $\star$ here.

  Like we did with $\oplus$, define an extended version of $\star$, 
  denoted $\ostar$, that will compute $\star$ over two extended states 
  $\ext{\pi}{\pi_1}$ and $\ext{\pi}{\pi_2}$ iff 
  $\proji{\ext{\pi_1}{\pi_1'}}{\vec{v_1}}[\vec{v_1} / \vec{v_2}] = \proji{\ext{\pi_2}{\pi_2'}}{\vec{v_2}}$.

  Then the following holds:
  {\small
  \begin{align*}
    & \!\!\!\! \bigcup_{\tiny \begin{array}{c} s_1 \in S_1 \\ s_2 \in S_2 \end{array}} \!\!\!\! 
    \bigcup_{ \tiny \begin{array}{c} \beta \in \aP_B \\ 
                    \beta_1 = \beta[\vec{v_1} / \vec{v}] \\ \beta_2 = \beta[\vec{v_2} / \vec{v}] \end{array}} 
    \tbranchext \star \fbranchext\\
    = & 
    \bigcup_{s_1 \in S_1} \!\!\!\! \bigcup_{\tiny \begin{array}{c} \beta \in \aP_B \\ 
                                   \beta_1 = \beta[\vec{v_1} / \vec{v}]\end{array}} \!\!\!\! \tbranchext \ostar
      \bigcup_{s_2 \in S_2} \!\!\!\! \bigcup_{\tiny \begin{array}{c} \beta \in \aP_B \\ 
                                   \beta_2 = \beta[\vec{v_2} / \vec{v}]\end{array}} \!\!\!\! \fbranchext
  \end{align*}
  }
  Observe that $\semantics{s_1}(\ext{\beta}{\beta_1}) = \ext{\semantics{s_1}(\beta)}{\beta_1}$ 
  (same for $s_2$ and $\beta_2$).
  Realign the extension operator back into the semantics, then apply the induction hypotheses for the statements to obtain:
  \begin{align*}
  &\bigcup_{s_1 \in S_1} \!\!\!\! \bigcup_{\tiny \begin{array}{c} \beta \in \aP_B \\ \beta_1 = \beta[\vec{v_1} / \vec{v}]\end{array}} \!\!\!\! 
    \proj{\ext{\semantics{s_1}(\beta)}{\beta_1}}{\vec{v_1}} \ostar
   \bigcup_{s_2 \in S_2} \!\!\!\! \bigcup_{\tiny \begin{array}{c} \beta \in \aP_B \\ \beta_2 = \beta[\vec{v_2} / \vec{v}]\end{array}} \!\!\!\! 
     \proj{\ext{\semantics{s_2}(\beta)}{\beta_2}}{\vec{v_2}} \\ 
  \subseteq & 
    \bigcup_{\ext{\sigma_1}{\sigma_1'} \in \aQ_1} \proj{\ext{\sigma_1}{\sigma_1'}}{\vec{v_1}} \ostar \bigcup_{\ext{\sigma_2}{\sigma_2'} \in \aQ_2} \proj{\ext{\sigma_2}{\sigma_2'}}{\vec{v_2}}
  \end{align*}
  Move $\ostar$ back inside the set union, and pull the $\mathsf{proj}$ operation outside of the set union to obtain:
   \begin{align*}
     \proj{\!\!\!\! \bigcup_{\tiny \begin{array}{c} \ext{\sigma_1}{\sigma_1'} \in \aQ_1 \\ \ext{\sigma_2}{\sigma_2'} \in \aQ_2 \\ \vec{v_1} = \vec{v_2} \end{array}} \!\!\!\! 
      \ext{\sigma_1}{\sigma_1'} \star \ext{\sigma_2}{\sigma_2'}}
      {\vec{v_1} \cup \vec{v_2}}
  \end{align*}
  Re-expand $\star$ to obtain:
   \begin{align*}
     \mathsf{proj}( \bigcup_{\tiny \begin{array}{c} \ext{\sigma_1}{\sigma_1'} \in \aQ_1 \\ \ext{\sigma_2}{\sigma_2'} \in \aQ_2 \\ \vec{v_1} = \vec{v_2} \end{array}}
       \mathsf{ext}( & {\proj{\ext{\sigma_1}{\sigma_1'}[\vec{v_1'}[i] / \vec{v}[i] \ \text{where} \ \vec{b_{t_1}}[i] = \Efalse]}{\vec{v_1'}}}, \\
                  & {\proj{\ext{\sigma_2}{\sigma_2'}[\vec{v_2'}[i] / \vec{v}[i] \ \text{where} \ \vec{b_{t_2}}[i] = \Etrue]}{\vec{v_2'}}})
                  , { \vec{v_1} \cup \vec{v_2}})
  \end{align*}
  Move all $\mathsf{proj}$ operations to the front, and expand $\mathsf{ext}$ into its predicate form to obtain the desired postcondition: 
  $$
      \exists \vec{v_1}, \vec{v_2}, \vec{v_1'}, \vec{v_2'}.
      \aQ_1[\vec{v_1'}[i] / \vec{v}[i] \ \text{where} \ \vec{b_{t_1}}[i] = \Efalse] \wedge
      \aQ_2[\vec{v_2'}[i] / \vec{v}[i] \ \text{where} \ \vec{b_{t_2}}[i] = \Etrue]
      \wedge (\vec{v_1} = \vec{v_2}) 
  $$

  $\whilerule$: 
Let us recall the rule:
{\small
$$  
\infer[\mathsf{While}] {\Gamma \vdash \utriple{\aI}{\Ewhile{B}{S}}{\aI_B \wedge \vec{b_t} = \Ef}}{
      \begin{array}{l}
      \vec{b_{loop}}, \vec{v_1}, \vec{v_2} \ \text{fresh} \\
      \Gamma \vdash \utriple{\aI}{B}{\aI_B} \\
      \Gamma \vdash \utriple{\aI_B \wedge \vec{b_{loop}} =
        \vec{b_t} \wedge \vec{v} = \vec{v_1}}{S}{\aI_B'}
      \end{array} \quad
      \begin{array}{l}
        (\exists \vec{v_1}, \vec{v_2}, \vec{v_1'}, \vec{v_2'}, \vec{b_t}. \aI_B'[\vec{v_{1}'}[i] / \vec{v}[i] \
\text{where} \  \vec{b_{loop}}[i] = \Efalse] \wedge \\
        (\aI_B \wedge \vec{b_{loop}} = \vec{b_t}\wedge \vec{v} = \vec{v_2})
        [\vec{v_2'}[i] / \vec{v}[i] \ \text{where} \ \vec{b_{loop}}[i] = \Etrue] \wedge \\
        (\vec{v_1} = \vec{v_2}))
        \implies
        (\exists \vec{b_t}. \aI_B)
      \end{array}
}
$$
}

  Here, $\aI, \aI_B, \aI_B'$ may be treated as invariants over the vector-states.
  This rule is highly similar to the standard Hoare logic rule for loops:
  it checks whether executing the set of loop bodies $S$, preserves the invariant encoded by 
  $\aI_B$.
  The precondition of executing $S$ only allows in states where the loop condition has evaluated to 
  $\Etrue$.
  The complex part is checking whether $\aI_B'$, the result of executing $S$ on $\aI_B$, satisfies an 
  invariant relation with $\aI_B$.

  Checking that $\aI_B'$ and $\aI_B$ satisfy an invariant relation is similar to what we did for 
  $\siterule$: we create a 'merged invariant' where states that have passed the loop guard $B$ 
  are drawn from $\aI_B'$, while states that have failed to pass the loop guard (and thus should not
  execute the loop body) are drawn from $\aI_B$ instead.
  The premise of the implication check on the right encodes this operation as a formula, in a similar 
  encoding to the postcondition for $\siterule$.
  Then the implication merely checks whether this merged invariant implies $\aI_B$, the original invariant,
  minus the auxiliary variable $b_t$.
  Under this intuition, it is easy to see the soundness of this rule:
  the correctness of the premise of the implication can be shown in a manner identical to the $\siterule$,
  which in turn guarantees that $\aI_B$ is an actual invariant for the set of loop bodies $S$.
  Since $\aI_B$ is an actual invariant, it follows that $\aI_B \wedge \vec{b_t} = \Ef$ is a valid overapproximation 
  of the possible post-states 
  (although, unlike in Hoare logic, this rule is not precise).

  $\mathsf{While\_Exact}$:
  The soundness of $\mathsf{While\_Exact}$ follows from the fact that the postcondition of 
  $\mathsf{While\_Exact}$ is \emph{precisely} the set of vector-states that may occur when 
  executing a pre-vector-state $\sigma_{\mathit{start}} \in \aP$.
  This is because a post-state $\sigma$ is in the post-condition if and only if it satisfies the 
  following conditions:
  \begin{itemize}
    \item There exists a corresponding $\sigma_{\mathit{start}} \in \aP$, such that, 
    \item For each $j$ in the vector-length of $\sigma$,
    \item There exists $b \in B$ and $s \in S$ such that 
      $\sem{\Ewhile{b}{s}}(\sigma_{\mathit{start}}[j]) = \sigma[j]$ 
      (which is captured by the formula inside the quantifiers).
  \end{itemize}
  Thus the postcondition of $\mathsf{While\_Exact}$ is actually merely a generalization of the 
  formula that expresses the semantics of single loops for single states~\cite{winskel} 
  towards sets of programs 
  and vector-states, and it thus follows that this postcondition precisely captures 
  the set of possible post vector-states.

	We now move onto structural rules in the logic.

  $\weaken$: %Since $\aP' \subseteq \aP$, $\aQ \subseteq \aQ'$, 
             %and $N_1 \subseteq N$, and $\semantics{N}(\aP) \in \aQ$ from the premise (induction hypothesis),
  Soundness follows from the fact that 
  $\semantics{N_1}(\aP') \subseteq \semantics{N}(\aP') \subseteq \semantics{N}(\aP) \subseteq \aQ \subseteq \aQ'$ 
  if $\aP \subseteq \aP'$, $N_1 \subseteq N$, and $\aQ' \subseteq \aQ$.

  $\conj$: 
  From $\semantics{N}(\aP) \! \subseteq \! \aQ_1 \wedge \semantics{N}(\aP) \! \subseteq \! 
  \aQ_2 \Rightarrow \semantics{N}(\aP) \subseteq \aQ_1 \wedge \aQ_2$.

  $\grmdisj$: From $\semantics{(N_1 \cup N_2)}(\aP) = \semantics{N_1}(\aP) \cup \semantics{N_2}(\aP) \subseteq \aQ$.

  $\inv$: Clearly, $\semantics{N}(\aP) = \aP$ if $N$ does not modify variables in $\aP$.

  $\subone$: By defining the triple as a formula, e.g., $\semantics{N}(\aP) \subseteq \aQ$, the entire formula can be
  alpha-renamed.
  That $\semantics{N}(\aP) \subseteq \aQ$ implies that $(\semantics{N}(\aP) \subseteq \aQ) [\vec{y} / \vec{z}]$; 
  since $N$ does not have intersecting variables with
  $\vec{y}$ and $\vec{z}$, the substitution leaves $N$ unchanged and thus leads 
  to $\semantics{N}(\aP[\vec{y} / \vec{z}]) \subseteq \aQ [\vec{y} / \vec{z}]$, which in turn is equivalent to 
 	$\utriple{\aP[\vec{y} / \vec{z}]}{N}{\aQ[\vec{y} / \vec{z}]}$. 

  $\subtwo$: $\subtwo$ is easier to prove by dividing cases.
  First consider that $\vec{z} \cap (\vars(N) \cup \vars(\aQ)) = \phi$.
  \begin{itemize}
    \item If $\vec{z} \cap \vars(\aP) = \phi$, then $\utriple{\aP}{N}{\aQ} \equiv \utriple{\aP[\vec{y} / \vec{z}]}{N}{\aQ}$ and is thus sound.
    \item If $\vec{y} \cap \vars(N) = \phi$, then $\utriple{\aP[\vec{y} / \vec{z}]}{N}{\aQ[\vec{y} / \vec{z}]} \equiv 
      \utriple{\aP[\vec{x} / \vec{z}]}{N}{\aQ}$ is a valid triple by alpha-renaming; thus the rule
      is sound.
    \item If $\vec{z} \cap \vars(\aP) \neq \phi$ and $\vec{y} \cap \vars(N) \neq \phi$: 
      Let $\pi$ be a vector-state in $\aP[\vec{y} / \vec{z}]$.
      We show that there exists a $\sigma \in \aQ$ where $\semantics{N}(\pi) = \sigma$.
      First observe that there exists some $\pi' \in \aP$ such that $\pi' = \pi$ on all variables 
      not in $\vec{z}$: this is because $\aP[\vec{y} / \vec{z}]$ is stronger than $\aP$ on variables not in $\vec{z}$.
      Thus there exists $\sigma' \in \aQ$ such that $\sigma'$ is `partially correct': 
      $\semantics{N}(\pi) = \sigma'$ on all variables not in $\vec{z}$.

      For variables $z \in \vec{z}$ of $\pi$, because $\vec{z} \cup \vars{N} = \phi$, $N$ leaves these variables 
      untouched; also because $\vec{z} \cup \vars{N} = \phi$, if 
      there exists $\sigma' \in \aQ$ for some specific configuration of non-$\vec{z}$ variables, 
      there must also exist $\sigma \in \aQ$ such that $\sigma = \pi$ on variables in $\vec{z}$.
      Thus there exists $\sigma \in \aQ$ such that $\semantics{N}(\pi) = \sigma$ for every 
      $\pi \in \aP[\vec{y} / \vec{z}]$; thus the rule is sound.
  \end{itemize}

	$\hp, \applyhp$: By soundness of mathematical structural induction.

  This covers all rules in unrealizability logic, and thus concludes the proof of soundness.
\end{proof}

%% file: pf_precision.tex
Having proved soundness, we will now start to prove completeness---starting 
from Lemma~\ref{thm:precision}.
As stated, the proof of Lemma~\ref{thm:precision} will rely on a structural 
induction proof similar to the proof of Theorem~\ref{thm:soundness}, with 
the $\while$ rule as an exception.

\begin{replemma}{thm:precision}[Completeness of Expression and Statement Rules]
  Let $G$ be a regular tree grammar and $N$ be a set of programs 
  generated by the RHS of a production.

  Then the expression and statement rules of unrealizability logic 
  \emph{preserve completeness}, in the sense that, 
  if all (sound) required premise triples are derivable, 
  then all (sound) triples about $N$ are derivable in unrealizability logic.
\end{replemma}
\begin{proof}
%  The overarching idea behind Theorem~\ref{thm:precision} is that we wish to prove 
%  `completeness' in the absence of the hypothesis rule.

  To prove preservation of completeness, we will rely on a somewhat counterintuitive lemma:
  \begin{lemma}
    Let us say that a triple $\utriple{\aP}{N}{\aQ}$ is reverse-sound if 
    $\semantics{\aP}(N) \supseteq \aQ$.
    Then the expression and statement rules (bar $\whilerule$) preserve reverse-soundness; that is, 
    if the premise triples are reverse-sound, then the conclusion is also reverse-sound.
  \end{lemma}
  \begin{proof}
    The proof of reverse-soundness is almost identical to the proof of soundness in Theorem~\ref{thm:soundness}:
    in particular, observe that the expansion of the semantics for each case is done via equality except for 
    the lines in which we invoke the induction hypothesis.
    The only difference comes from the direction of the subset relation due to the induction hypothesis; the rest of the 
    expansion is identical.

    $\zero, \one, \varrule, \truerule, \falserule$: Take $\varrule$ as an example.
    As shown in the base case section for the proof of Theorem~\ref{thm:soundness}, 
    $\semantics{N}(\aP) = \semantics{x}(\aP) = \exists  \vec{e_t'}. \aP[\vec{e_t'} / \vec{e_t}] \wedge \vec{e_t} = \vec{x}$, 
    so this rule is clearly precise (the base case does not invoke the induction hypothesis).

    $\notrule$: In this case, the induction hypothesis states that the premise triple $\utriple{\aP}{B}{\aQ}$ 
    is reverse-sound (instead of sound).
    Thus the semantic expansion becomes:
  \begin{align*}
    \semantics{N}(\aP) & = \semantics{!B}(\aP) \\
                       & = \bigcup_{\bexp \in B} \bigcup_{\pi \in \aP} \semantics{! \bexp}(\pi)
                           \tag{$!\bexp \in !B \iff \bexp \in B$}\\
                       & = \bigcup_{\bexp \in B} \bigcup_{\pi \in \aP} \lrangle{\semantics{!\bexp}(\state_1), \cdots } \\
                       & = \bigcup_{\bexp \in B} \bigcup_{\pi \in \aP}
                           \lrangle{\semantics{b}(\state_1)[b_{t_1} \mapsto ! \semantics{\bexp}(\state_1)[b_{t_1}]], \cdots } \\
                       & \supseteq \bigcup_{\sigma \in \aQ} \lrangle{\sigma[\vec{b_t} \mapsto ! \sigma[\vec{b_t}]]}
                           \tag{$\bigcup_{\bexp \in B} \bigcup_{\pi \in \aP} \semantics{\bexp}(\pi) \subseteq \aQ$, \text{ by I.H.}}\\
                       & = \exists \vec{b_t'}. \aQ[\vec{b_t'} / \vec{b_t}] \wedge \vec{b_t} = \lnot \vec{b_t'}                        
	\end{align*}
	And this shows that $\notrule$ preserves reverse-soundness.

    $\mathsf{Bin}$, $\mathsf{Comp}$, $\andrule$: Like $\notrule$, the only place where the semantic expansion is a subset is when 
    applying the induction hypothesis; assuming reverse-soundness via the induction hypothesis yields the 
    desired result.

  $\assignrule, \seqrule, \siterule$: The same reasoning applies for the statement nonterminals as well: 
  assume that the premise is reverse-sound by the induction hypothesis, then apply the 
	same kind of semantic expansion in the soundness argument
	for the respective cases to show that the rules are reverse-sound.
  \end{proof}
  The fact that the rules are reverse-sound and that they are also sound implies that the rules are \emph{precise}: 
  that is, if the provided premise triples are such that $\semantics{N}(\aP) = \aQ$, then 
  the conclusion will also have that $\semantics{N}(\aP) = \aQ$.
  From this, one can apply $\weaken$ to derive any desired triple.

  The overarching idea behind Theorem~\ref{thm:precision} is that we wish to prove
  completeness of \emph{only} the expression and statement 
  rules of unrealizability logic, \emph{without} having to reason about nonterminals
  or recursive structures in the grammar---this is why we proved that the rules 
  (bar $\whilerule$) are precise.
  Note that this theorem alone does \emph{not} imply the completeness of 
  unrealizability logic; 
  this theorem will later be used in the proof of Lemma~\ref{lem:strongest_triple}, 
  which performs the induction over the structure of the grammar, 
  to prove (relative) completeness of the whole logic.

  The only constructor left to consider is $\mathsf{While}$: 
  the preciseness of loops follows from the rule $\mathsf{While\_{Exact}}$, which 
  \rone does not have any premises and 
  \rtwo as shown in the proof of Theorem~\ref{thm:soundness}, 
  the rule is precise.
  Thus we will always be able to construct $\aQ$ such that $\sem{\Ewhile{B}{S}}(\aP) = \aQ$, 
  for some given $\aP$ and $\Ewhile{B}{S}$.
\end{proof}

%% file: pf_strongest_triple.tex
Following the proof of Lemma~\ref{thm:precision}, we proceed to 
prove Lemma~\ref{lem:strongest_triple} by induction on the number of 
hypotheses in the context.

\begin{replemma}{lem:strongest_triple}[Derivability of the Strongest Triple]
  Let $N$ be a nonterminal from a RTG $G$ and $\vec{z}$ be a set of auxiliary symbolic variables.
  Let $\aQ_0 \equiv \semantics{N}(\vec{x} = \vec{z})$; that is,
  $\aQ_0$ is a formula that precisely captures the behavior of the set of programs $L(N)$ on the symbolic vector-state
  $x_1 = z_1 \wedge \cdots x_n = z_n$.
  We refer to $\utriple{\vec{x} = \vec{z}}{N}{\aQ_0}$ as the \emph{strongest triple} for $N$.
  Then $\vdash \utriple{\vec{x} = \vec{z}}{N}{\aQ_0}$ is derivable.
\end{replemma}

\begin{proof}
  To prove derivability of the strongest triple within a finite number of steps, we fix a particular
  \emph{proof strategy} that we take.
  The gist of the proof is that a strongest triple can be proved by inserting a finite number of
  hypotheses into the context; thus, one can perform induction on the number of hypotheses inserted,
  i.e., the number of $\hp$ rules applied to the proof tree at a certain point.

  The proof strategy is that upon encountering a new nonterminal $N'$
  for the \emph{first} time in a proof tree, we immediately apply $\hp$ to $N$ and insert
  the strongest triple $\utriple{\vec{x} = \vec{z}}{N'}{\aQ_{0, N'}}$ as a hypothesis into the context $\Gamma$.
  (We refer to this triple as the strongest triple, because as later shown in Theorem~\ref{lem:general_triple}, 
  one can derive any triple starting from the strongest triple).
  This is the only time we will introduce hypotheses into the context, and thus
  the maximum number of hypotheses that may be introduced is $|G|$ (the number of nonterminals in $G$).

  We inductively show that this construction guarantees a finite proof tree that
  proves the strongest triple for any new $N' \in G$ (and thus, $N$ as well).
  (To give a simpler proof, we again assume that $\fb$ is fixed to $\vec{\Et}$; it is again
  straightforward to extend the proof to consider other $\fb$).

  \textbf{Base case: $|G| - 1$ hypotheses in $\Gamma$.}
  The base case is when $N'$ is the last new nonterminal that has been encountered, and hypotheses for all other
  nonterminals have already been inserted into the context $\Gamma$.
  (The induction does not start at $0$, because we are constructing the proof tree bottom-up.)
  Then, following the proof strategy, one must introduce the strongest triple for $N'$ as a hypothesis to produce
  a new context $\Gamma_{N'}$; $\Gamma_{N'}$ will have strongest triples for all $|N'|$ nonterminals as hypotheses.
  For the sake of example, let us assume that $N' \rightarrow \Eassign{x}{E_1} \mid \cdots$.
  Then the application of the $\hp$ rule is:
  $$
    \infer[\hp]{\Gamma \vdash \utriple{\vec{x} = \vec{z}}{N'}{\aQ_{0, N'}}}
    {
      \Gamma_{N'} \vdash \utriple{\vec{x} = \vec{z}}{\Eassign{x}{E_1}}{\aQ_{0, N'}} \quad
      \cdots
    }
  $$
  At this point, observe the following:
  \begin{itemize}
    \item To derive a triple for $\Eassign{x}{E_1}$, one must rely on the $\assignrule$ rule, which
      has as a premise a triple for $E_1$.
    \item Since $E_1$ is a nonterminal in $G$, the precise triple for $E_1$ is already in $\Gamma_{N'}$,
      we may apply $\applyhp$ to resolve this premise.
    \item As shown in Theorem~\ref{thm:precision}, the rule $\assignrule$ preserves completeness: thus,
      one can derive a precise triple for the behavior of $\Eassign{x}{E_1}$ as well.
    \item Then, due to Lemma~\ref{lem:general_triple}, we can prove any valid triple for $\Eassign{x}{E_1}$.
      Of these is the target triple $\utriple{\vec{x} = \vec{z}}{\Eassign{x}{E_1}}{\aQ_{0, N'}}$; we have shown that
      this is derivable.
  \end{itemize}
  The same reasoning applies for other possible productions of $N'$: since precise hypotheses about all
  nonterminals are present in $\Gamma_{N'}$, and the non-structural rules of unrealizability logic are precise,
  one can always derive $\utriple{\vec{x} = \vec{z}}{RHS}{\aQ_{0, N'}}$ for an $RHS$ of $N'$.

  Thus the strongest hypothesis for $N'$ is derivable when there are $|G| - 1$ hypotheses in $\Gamma$.

  \textbf{Inductive case: $0 \leq n < |G| - 1$ hypotheses in $\Gamma$.}
  The inductive case is when $N'$ is not the last nonterminal to be encountered, and there are
  only $n$ nonterminals in the context.
  In this case, one can apply reasoning identical to the base case, with the addition that
  strongest triples for nonterminals \emph{not} in the context can be derived within a finite number
  of steps due to the induction hypothesis.

  The minimum number of hypotheses is when $\Gamma$ is the empty context: thus, we have proved that
  there always exists a finite proof tree deriving the strongest triple for any nonterminal $N$,
  starting from an empty context.
\end{proof}

%% file: pf_general_triple.tex
We now prove Lemma~\ref{lem:general_triple}, which, in conjunction with 
Lemma~\ref{lem:strongest_triple}, implies relative completeness.

\begin{replemma}{lem:general_triple}[Derivability of General Triples]
  Let $G$ be a grammar and $N$ be any nonterminal in $G$.
  Given the strongest triple $H_N$ for the nonterminal $N$
  (as defined in Lemma~\ref{lem:strongest_triple}),
  if $\Gamma \vdash H_N$,
  then $\Gamma \models \utriple{\aP}{N}{\aQ} \implies \Gamma \vdash \utriple{\aP}{N}{\aQ}$.
\end{replemma}

\begin{proof}
  Again, we prove the case where the states for $N$ are defined over a single variable $x$; it is easy to
  extend the proof toward multiple-variable states.

  Let $\vec{x}$ be the vector-variable for $x$, and $\vec{z}$, $\vec{u}$ be fresh vectors of
  symbolic variables (of the same length as $\vec{x}$).
  From $\inv$, one has that
  $\utriple{\aP[\vec{u} / \vec{z}][\vec{z} / \vec{x}]}{N}{\aP[\vec{u} / \vec{z}][\vec{z} / \vec{x}]}$
  for any $\aP$.
  Also, by assumption, one can derive the strongest triple for $N$, $H_N = \utriple{\vec{x} = \vec{z}}{N}{\aQ_0}$.

  By $\conj$ and $\weaken$, one can derive:
  $$
    \Gamma \vdash \utriple{\vec{x} =
    \vec{z} \wedge \aP[\vec{u} / \vec{z}][\vec{z} / \vec{x}]}{N}{\aQ_0 \wedge \aP[\vec{u} / \vec{z}][\vec{z} / \vec{x}]}
  $$
  Let $\aQ_{in}$ be the postcondition $\aQ_0 \wedge \aP[\vec{u} / \vec{z}][\vec{z} / \vec{x}$] of this intermediate triple.
  We show that $\aQ_{in}$ implies $\aQ[\vec{u} / \vec{z}]$.
  Notice that, as $\aQ_0$ is exact, for every $\sigma \in \aQ_{in}$,
  there must exist $\pi \in \vec{x} = \vec{z}$ and $\stmt \in L(N)$ such that
  $\semantics{\stmt}(\pi) = \sigma$.

  Now, if $\pi \not \in \aP[\vec{u} / \vec{z}][\vec{z} / \vec{x}]$,
  then $\sigma \not \in \aP[\vec{u} / \vec{z}][\vec{z} / \vec{x}]$ as well due to the soundness of $\mathsf{Inv}$,
  which is a contradiction.
  Thus $\pi \in (\vec{x} = \vec{z}) \wedge \aP[\vec{u} / \vec{z}][\vec{z} / \vec{x}]$,
  and as $(\vec{x} = \vec{z}) \wedge \aP[\vec{u} / \vec{z}][\vec{z} / \vec{x}] \implies \aP[\vec{u} / \vec{z}]$
  (as the second sub has no effect), $\pi \in \aP[\vec{u} / \vec{z}]$.

  Recall that $\semantics{\stmt}(\pi) = \sigma$ for some $\sigma$.
  Because $\pi \in \aP[\vec{u} / \vec{z}]$, and the soundness of $\subone$,
  the assumption that $\models \utriple{\aP}{N}{\aQ}$ indicates that
  $\utriple{\aP[\vec{u} / \vec{z}]}{N}{\aQ[\vec{u} / \vec{z}]}$ is valid, and thus $\sigma \in \aQ[\vec{u} / \vec{z}]$.
  This holds for arbitrary $\sigma \in \aQ_{in}$, thus $\aQ_{in} \implies \aQ[\vec{u} / \vec{z}]$.

  Thus, by $\weaken$, one can derive:
  $$
    \Gamma \vdash \utriple{\vec{x} = \vec{z} \wedge \aP[\vec{u} / \vec{z}][\vec{z} / \vec{x}]}{N}{\aQ[\vec{u} / \vec{z}]}
  $$
  From $\subtwo$, by substituting $\vec{z}$ for $\vec{x}$, one gets:
  $$
    \Gamma \vdash \utriple{\aP[\vec{u} / \vec{z}]}{N}{\aQ[\vec{u} / \vec{z}]}
  $$
  Finally, apply $\subone$ to substitute $\vec{u}$ for $\vec{z}$;
  this shows that $\Gamma \vdash \utriple{\aP}{N}{\aQ}$ is derivable.
\end{proof}

The proof of Theorem~\ref{lem:strongest_triple} and Theorem~\ref{lem:general_triple}
is similar to the proof of completeness for recursive Hoare logic detailed in~\cite{10years}, 
which also relies on a context to store hypotheses about recursive procedures.

%% file: pf_completeness.tex
Finally, completeness follows directly from Lemma~\ref{lem:strongest_triple} and Lemma~\ref{lem:general_triple}.

\begin{reptheorem}{thm:completeness}[Completeness]
  Given a grammar $G$ and a nonterminal $N \in G$, $\models \utriple{\aP}{N}{\aQ} \implies \vdash \utriple{\aP}{N}{\aQ}$.
\end{reptheorem}

\begin{proof}
  Because the strongest triple $H_N$ for $N$ is derivable (Lemma~\ref{lem:strongest_triple}),
  completeness follows from Lemma~\ref{lem:general_triple}.
\end{proof}

%% file: pf_decidability.tex
We now move on to proving theorems related to decidability.

\begin{reptheorem}{thm:undecidability}[Undecidability]
  Let $\syu$ be a synthesis problem with a grammar $\gu$ that does \emph{not}
  contain productions of the form $S::=\Ewhile{B}{S_1}$.
  Checking whether $\syu$ is unrealizable is an undecidable problem.
\end{reptheorem}

\begin{proof}
  Let $C$ be a two-counter machine with registers labeled $x$ and $y$, which supports the instructions
  $\cinc{r}$ (increment the register $r$), $\cdec{r}$ (decrement the register $r$),
  and $\cjez{r}{l}$ (jump to the location $l$ if the value stored in $r = 0$).
  We proceed by reducing the halting problem for $C$ (which is an undecidable problem) to realizability for $\syu$.

  Let $P_C$ be an arbitrary two-counter machine program, defined as a list of instruction-location pairs
  $(i_1, 1), (i_2, 2), \cdots (i_n, n)$ (i.e., the location of the $k$-th instruction is $k$) and an
  initial configuration for the registers, $x_c$ and $y_c$.
  From $P_C$, one can construct a grammar $G_C$,
  by translating each pair $(i_k, l_k)$ as follows.

  \begin{itemize}
    \item First, introduce a new nonterminal $S_k$ in $G_C$.
    \item If $i_k \equiv \cinc{r}$, add
      $S_k \rightarrow \Eseq{\Eassign{r}{r + 1}}{S_{k + 1}} \mid \Eskip$ to $G_C$.
    \item If $i_k \equiv \cdec{r}$, add
      $S_k \rightarrow \Eseq{\Eassign{r}{r - 1}}{S_{k + 1}} \mid \Eskip$ to $G_C$.
    \item If $i_k \equiv \cjez{r}{l}$, add
      $S_k \rightarrow \Eifthenelse{x = 0}{S_l}{S_{k + 1}} \mid \Eskip$ to $G_C$.
    \item Note that each $S_k$ up to this point has two productions, one that goes to $\Eskip$ and
      one that does not.
      We name each non-$\Eskip$ production as $Pr_k$.
    \item Finally, add the nonterminal $S_{n + 1}$ and the productions
      $S_{n + 1} \rightarrow \Eassign{h}{1}$ and $\nstart \rightarrow S_1$ to $G_C$.
  \end{itemize}
  Define the synthesis problem $sy_C$ over three variables $x$, $y$, and $c$, with the grammar as $G_C$,
  the input specification as $(x, y, h) = (x_c, y_c, 0)$, and the output specification as
  $h = 1$ (we do not care about the values of $x$ and $y$ in the output).

  Given a trace of locations $L = [l_1, l_2, \cdots]$
  obtained by executing some $P_C$, one can always find a corresponding, valid series of productions in $G_C$.
  This can be seen inductively: the starting production is $\nstart \rightarrow S_1$, so we start at $S_1$,
  while $P_C$ starts at the location $1$.
  The inductive argument is as follows:
  \begin{itemize}
    \item Assume that we are executing instruction $i_k$ at location $k$, and the current nonterminal is $S_k$.
    \item If $i_k$ is $\cinc{r}$, the next instruction to execute is $k + 1$.
      In addition, $S_k \rightarrow \Eseq{\Eassign{r}{r +1}}{S_{k + 1}} \mid \Eskip$; take the
      non-skip production, $Pr_k$, and generates $S_{k + 1}$ as the next nonterminal.
      (Similar for $\cdec{r}$.)
    \item If $i_k$ is $\cjez{r}{l}$, the non-skip production
      $Pr_k$ is $\Eifthenelse{r = 0}{S_l}{S_{k + 1}}$.
      If the execution takes the true branch, send $S_{k + 1}$ to $\Eskip$ and continue inductively on $S_l$.
      If the execution takes the false branch, send $S_{l}$ to $\Eskip$ and continue inductively on $S_{k + 1}$.
  \end{itemize}

  From this correspondence, the following lemma holds:

  \begin{lemma}
    \label{lem:real_if_halts}
    $sy_C$ is realizable iff $P_C$ halts.
  \end{lemma}

  \begin{proof}
    $sy_C$ is realizable $\implies$ $P_C$ halts:
    If $sy_C$ is realizable through a program $t_C \in G_C$, notice that there is a variant of $t_C$, $t_C'$, also
    in $G_C$, for which the `untaken' branches of if-then-elses are sent to $\Eskip$s instead.

    This is because i) $sy_C$ is a single-example synthesis problem, and thus examples only pass through one branch; and
    ii) all productions except those for $S_{n + 1}$ have $\Eskip$ as an option.
    Then the subscripts of the non-skip productions used to derive $t_C'$ serve as a witness execution trace that
    $P_C$ terminates.

    $P_C$ halts $\implies$ $sy_C$ is realizable:
    As $P_C$ halts, the trace of executions $L = [l_1 \cdots l_m]$ is finite; in addition, $l_m$ is always $n$.
    Following the construction above, at location $n$, the current nonterminal is $S_n$.
    If $i_n$ is $\cinc{r}$ or $\cdec{r}$, then $S_n$ takes the non-skip production and the next active nonterminal is $S_{n + 1}$.
    If $i_n$ is $\cjez{l}{r}$, note that it cannot take the branch as the trace $l_m$ is the final instruction;
    thus, according to the construction above, one sends $S_l$ to $\Eskip$ and continues with $S_{n + 1}$.

    In both cases, $S_{n + 1}$ may only produce $\Eassign{h}{1}$, which is the final statement in the program,
    so $sy_C$ is realizable.
  \end{proof}
  From Lemma~\ref{lem:real_if_halts}, one can see that realizability is undecidable,
  because the halting problem for two-counter machines is also undecidable.
\end{proof}

We now give the proof that when limited to finite domains, proving unrealizability becomes 
a decidable problem (even those containing while loops).

\begin{reptheorem}{thm:finite_decidability}[Decidability of Unrealizability over Finite Domains]
  Determining whether a synthesis problem $\syfin$ is realizable, where the grammar $\gfin$
  is valid with respect to $G_{impv}$, is \emph{decidable} if the semantics of
  programs in $\gfin$ is defined over a finite domain $D$.
\end{reptheorem}

\begin{proof}
  Let $|D|$ denote the size of the domain, and $|\gfin|$ the number of nonterminals in 
  $\gfin$.
  Given the starting nonterminal $\nstart \in \gfin$, we argue that 
  performing grammar-flow analysis on $\gfin$ to compute the set of producible values 
  as a greatest fixed point terminates in a finite number of steps.

  The GFA procedure is simple: for each nonterminal $N \in \gfin$, construct a map
  $T_N$ that maps an input (vector-)states to the set of all possible output (vector-)states.
  $T_N$ is guaranteed to have a finite number of entries, because 
  the domain is finite: thus there are at most a finite number of distinct vector-states 
  ($|D|^v$ for $v$ variables) that may be input.
  Likewise, the set $T_N(\pi)$ for an arbitrary input vector-state $\pi$ is also 
  guaranteed to be finite.

  Then it is easy to see the GFA procedure terminates: starting from the single input 
  example and $0$-ary operators, compute the possible set of outputs, 
  and use them as possible inputs for computing the results of other productions recursively.
  Note that at this point, computing the results of productions is also a decidable 
  process that terminates, because while loops either i) terminate or ii) display a fixed 
  pattern of states (implying nontermination) within a finite number of steps.
  The process is guaranteed \emph{not} to remove any possible sets from the output; 
  thus at each step, the process is guaranteed to i) add nothing to an output set, or 
  ii) add something to an output step.

  If there comes a point where nothing is added to an output set for all nonterminals,
  then the output sets have reached a fixed point and GFA may terminate.
  Otherwise, GFA will keep adding to output sets until the output set has saturated: 
  this is guaranteed to terminate in a finite amount of steps, as the output set 
  has at most finite cardinality.
  Thus GFA will terminate in a finite amount of steps, and the output sets computed 
  in this method are also precise: thus unrealizability amounts to checking whether 
  the desired output is in $T_{\nstart}(I)$ (where $I$ is the vectorized input example).
  Thus unrealizability is decidable.
 % Let $|D|$ denote the size of the domain, and $|\gfin|$ denote the number of 
 % nonterminals in $\gfin$.
 % As stated in \S\ref{SubSe:Decidability}, the key idea of the proof is to expand 
 % $\gfin$ to contain nonterminals of the form $N_{1, 1}, N_{1, 2}, \cdots, N_{D, D}$; 
 % formally, this may be done through transforming $\gfin$ into 
 % an extended grammar $\gfinw$ which annotates the input and output states of the 
 % nonterminal.
 %
 % Let $v$ be the number of variables inside the program, and $n$ be the number 
 % of examples for $\syfin$.
 % Note that as the domain $D$ is finite, 
 % $n$ is also guaranteed to be finite (a maximum of $|D|^v$, where $v$ is the number of 
 % variables in the program).
 % For each nonterminal $N \in \gfin$, generate a new nonterminal 
 % $N_{I, O} \in \gfinw$, where $I \in |D|^{v^n}$ and $O \in |D|^{v^n}$.
 % Intuitively, $N_{I, O}$ captures the case where an input vector-state $I$ enters 
 % $N$ and results in the output $O$.
 % With these nonterminals, the productions may be calculated by computing the 
 % semantics of productions, which are again decidable 
 % due to the fact that $D$ is finite.
 % For example, for every production 
 % $S \rightarrow \Eseq{S_1}{S_2}$, one can add the productions:
 % $S_{I, O} \rightarrow \Eseq{S_{1_{I, K}}}{S_{2_{K, O}}}$
 % for every $I, K, O$.

 % Then the question of realizability is reduced to checking whether 
 % $L(\nstart_{I, O})$ is empty or not for input and output vector-states 
 % $I$ and $O$ (which is decidable).
\end{proof}